\documentclass[a4paper,final,UKenglish, numberwithinsect, cleveref, autoref, thm-restate]{lipics-v2021}
\pdfoutput=1

\nolinenumbers

\EventEditors{Serge Haddad and Daniele Varacca}
\EventNoEds{2}
\EventLongTitle{32nd International Conference on Concurrency Theory (CONCUR 2021)}
\EventShortTitle{CONCUR 2021}
\EventAcronym{CONCUR}
\EventYear{2021}
\EventDate{August 23--27, 2021}
\EventLocation{Virtual Conference}
\EventLogo{}
\SeriesVolume{203}
\ArticleNo{32}

\usepackage[utf8]{inputenc}

\usepackage{microtype}
\usepackage{amsmath}
\usepackage{amssymb}
\usepackage{stmaryrd}
\usepackage{xspace}
\usepackage{mathtools}
\usepackage{booktabs}
\usepackage{pifont}
\usepackage{nicefrac}
\usepackage{array} %
\usepackage{tikz}
\usepackage[noadjust]{cite}

\usepackage[marginclue,footnote,final]{fixme}
\FXRegisterAuthor{tw}{atw}{TW}%
\FXRegisterAuthor{sm}{asm}{SM}%
\FXRegisterAuthor{ls}{als}{LS}%

\usetikzlibrary{calc,cd,fit,decorations,arrows,arrows.meta}

\usepackage{twshowkeys}

\newcommand{\resetCurThmBraces}{%
\gdef\curThmBraceOpen{(}%
\gdef\curThmBraceClose{)}}
\resetCurThmBraces
\newcommand{\removeThmBraces}{%
\gdef\curThmBraceOpen{}%
\gdef\curThmBraceClose{}}
\resetCurThmBraces

\newenvironment{notheorembrackets}{\removeThmBraces}{\resetCurThmBraces}

\usepackage{etoolbox}
\patchcmd{\thmhead}{(#3)}{\curThmBraceOpen #3\curThmBraceClose }{}{}

\usepackage{hyperref}
\hypersetup{hidelinks,final,bookmarks}
\addto\extrasUKenglish{%
}

\newcommand{\takeout}[1]{\empty}

\makeatletter
\@ifpackageloaded{stix}{%
}{%
  \DeclareFontEncoding{LS2}{}{\noaccents@}
  \DeclareFontSubstitution{LS2}{stix}{m}{n}
  \DeclareSymbolFont{stix@largesymbols}{LS2}{stixex}{m}{n}
  \SetSymbolFont{stix@largesymbols}{bold}{LS2}{stixex}{b}{n}
  \DeclareMathDelimiter{\lBrace}{\mathopen} {stix@largesymbols}{"E8}%
                                            {stix@largesymbols}{"0E}
  \DeclareMathDelimiter{\rBrace}{\mathclose}{stix@largesymbols}{"E9}%
                                            {stix@largesymbols}{"0F}
}
\makeatother
\newcommand{\bag}[2][]{\ensuremath{{#1\lBrace #2#1\rBrace}}}

\newcommand{\Set}{\ensuremath{\mathsf{Set}}\xspace}
\newcommand{\pred}{\ensuremath{\mathsf{pred}}\xspace}
\newcommand{\id}{\ensuremath{\mathsf{id}}\xspace}
\newcommand{\pr}{\ensuremath{\mathsf{pr}}}
\newcommand{\inl}{\ensuremath{\mathsf{in}_{1}}}
\newcommand{\inr}{\ensuremath{\mathsf{in}_{2}}}
\newcommand{\inj}{\ensuremath{\mathsf{in}}}
\newcommand{\fpair}[1]{\ensuremath{\langle #1 \rangle}}
\newcommand{\set}[2][]{\ensuremath{{#1\{#2#1\}}}}
\newcommand{\BoolComb}{\ensuremath{{\mathcal{B}\!\ell}}\xspace}
\newcommand{\BoolMonoid}{\ensuremath{{\mathbb{B}}}\xspace}

\newcommand{\grad}[1]{\monmod{#1}}
\newcommand{\gradI}{\monmod{\raisebox{-1pt}{\ensuremath{\scriptstyle I}}}}
\newcommand{\monmod}[1]{\ensuremath{{\langle{\raisebox{1pt}{\ensuremath{\scriptstyle\mathord{=}#1}}}\rangle}}}

\newcommand{\hearts}{\heartsuit}

\newcommand{\Pow}{\ensuremath{\mathcal{P}}}
\newcommand{\CO}{\ensuremath{\mathcal{O}}}
\newcommand{\Powf}{\ensuremath{{\mathcal{P}_\textsf{f}}}}
\newcommand{\Dist}{\ensuremath{{\mathcal{D}}}}
\newcommand{\Bag}{\ensuremath{\mathcal{B}}}
\newcommand{\unzip}{\ensuremath{\mathsf{unzip}}}

\newcommand{\textqt}[1]{\text{`#1'}}
\newcommand{\N}{\ensuremath{\mathbb{N}}}
\newcommand{\R}{\ensuremath{\mathbb{R}}}
\newcommand{\Z}{\ensuremath{\mathbb{Z}}}
\newcommand{\monoto}{\rightarrowtail}
\newcommand{\monodown}{%
\begin{tikzpicture}[baseline=(B)]
  \node[rotate=-90,inner xsep=0pt,inner ysep=2pt]
     (M) {\ensuremath{\rightarrowtail}};
  \coordinate (B) at ([yshift=3pt]M.east);
\end{tikzpicture}}

\newcommand{\Formulae}{\ensuremath{\mathcal{M}}}

\newcommand{\etal}[1][.]{\text{et~al.}}
\newcommand{\arity}[3][]{\ensuremath{\mathord{\raisebox{1pt}{\ensuremath{#1#2}}\mkern-1.5mu/\mkern-1.5mu{\raisebox{-1pt}{\ensuremath{#1#3}}}}}}
\newcommand{\semantics}[2][]{\ensuremath{#1\llbracket #2 #1\rrbracket}}

\newcommand{\fmod}[1]{\ensuremath{{\ulcorner #1\urcorner}}}
\newcommand{\itemref}[2]{\autoref{#1}.{\ref{#2}}}
\newcommand{\notni}{\ensuremath{\not\mkern1mu\ni}}

\newcommand{\refAllEdgeCases}{%
  \caseref{edgeS}--\caseref{edgeBS}\xspace%
}

\newcommand{\scissors}{\ding{34}} %

\tikzset{
  coalgebra/.style={
    block line/.style={
      draw=black!50,
      line width=1.2pt,
    },
    block/.style={
      block line,
      rounded corners=3pt,
      inner sep=1pt,
      minimum height=6mm,
      minimum width=6mm,
    },
    scissors line/.style={
      draw=black!50,
      text=black!50,
      font=\footnotesize,
      line width=0.8pt,
      shorten <= -4pt,
      shorten >= -4pt,
      dotted,
    },
    description/.style={
      fill=white,
      inner sep=1pt,
    },
    state/.style={
      text depth=0pt,
      outer sep=0pt,
      inner sep=4pt,
    },
    transition/.style={
      -{latex},
      line width=0.8pt,
      draw=black,
      preaction = {draw,-,draw=white,line width=4.6pt,line cap=round},
    },
    path with edges/.style={
      every edge/.append style={transition}
    },
  },
}

\newcommand{\descto}[3][]{\arrow[phantom]{#2}[#1]{\text{\footnotesize{}\begin{tabular}{c}#3\end{tabular}}}}

\usetikzlibrary{decorations.markings}
\usetikzlibrary{decorations.pathreplacing}

\newcolumntype{L}{>{$}l<{$}} %

\theoremstyle{plain}
\newtheorem{oprob}[theorem]{Open Problem}
\theoremstyle{definition}
\newtheorem{notation}[theorem]{Notation}
\newtheorem{defn}[theorem]{Definition}
\newtheorem{expl}[theorem]{Example}
\newtheorem{rem}[theorem]{Remark}
\newtheorem*{rem*}{Remark}
\newtheorem{algorithm}[theorem]{Algorithm}

\newcommand{\appendixsect}[2]{%
  \section*{Details for \autoref{#1} (#2)}
  \addcontentsline{toc}{section}{Details for \autoref{#1} (#2)}%
}
\newenvironment{proofappendix}[2][Proof of]{%
\subsubsection*{#1~\autoref{#2}.}%
\addcontentsline{toc}{subsection}{#1~\autoref{#2}}%
}{}

\newcommand{\xra}[1]{\xrightarrow{~#1~}}

\title{Explaining Behavioural Inequivalence Generically \texorpdfstring{\\}{}%
  in Quasilinear Time}
\titlerunning{Explaining Behavioural Inequivalence Generically in Quasilinear Time}
\author{Thorsten Wißmann}%
{Radboud University, Nijmegen, The Netherlands}
{uni@thorsten-wissmann.de}%
{https://orcid.org/0000-0001-8993-6486}{%
  Work forms part of the NWO TOP project 612.001.852 and the DFG-funded project COAX (MI 717/5-2)}

\author{Stefan Milius}{Friedrich-Alexander-Universit\"{a}t
  Erlangen-N\"{u}rnberg,
  Germany}{stefan.milius@fau.de}{https://orcid.org/0000-0002-2021-1644}
{Work forms part of the DFG-funded project CoMoC (MI~717/7-1)}

\author{Lutz Schröder}{Friedrich-Alexander-Universit\"{a}t
  Erlangen-N\"{u}rnberg, Germany}{lutz.schroeder@fau.de}{http://orcid.org/0000-0002-3146-5906}{Work forms part of the DFG-funded project CoMoC (SCHR~1118/15-1)}
\authorrunning{T.~Wißmann, S.~Milius, and L.~Schröder}
\Copyright{T.~Wißmann, S.~Milius, and L.~Schröder}
\relatedversion{}
\relatedversiondetails{Full Version with Appendix}{https://arxiv.org/abs/2105.00669}




\begin{document}
\maketitle              %
\begin{abstract}
  We provide a generic algorithm for constructing formulae that
  distinguish behaviourally inequivalent states in systems of various
  transition types such as nondeterministic, probabilistic or
  weighted; genericity over the transition type is achieved by working
  with coalgebras for a set functor in the paradigm of universal
  coalgebra. For every behavioural equivalence class in a given
  system, we construct a formula which holds precisely at the states
  in that class. The algorithm instantiates to deterministic finite
  automata, transition systems, labelled Markov chains, and systems of
  many other types. The ambient logic is a modal logic featuring
  modalities that are generically extracted from the functor; these
  modalities can be systematically translated into custom sets of
  modalities in a postprocessing step. The new algorithm builds on an
  existing coalgebraic partition refinement algorithm. It runs in time
  $\mathcal{O}((m+n) \log n)$ on systems with~$n$ states and~$m$ transitions,
  and the same asymptotic bound applies to the dag size of the
  formulae it constructs.  This improves the bounds on run time and
  formula size compared to previous algorithms even for previously
  known specific instances, viz.~transition systems and Markov chains;
  in particular, the best previous bound for transition systems was
  $\mathcal{O}(m n)$.

\end{abstract}

\section{Introduction}
For finite transition systems, the Hennessy-Milner theorem guarantees
that two states are bisimilar if and only if they satisfy the same
modal formulae. This implies that whenever two states are not
bisimilar, then one can find a modal formula that holds at one of the
states but not at the other.  Such a formula explains the difference
of the two states' behaviour and is thus usually called a
\emph{distinguishing formula}~\cite{Cleaveland91}.  For example, in
the transition system in \autoref{fig:exampleTS}, the formula
$\Box\Diamond\top$ distinguishes the states $x$ and $y$ because $x$
satisfies $\Box\Diamond\top$ whereas~$y$ does not. Given two states in
a finite transition system with $n$ states and $m$ transitions, the
algorithm by Cleaveland~\cite{Cleaveland91} computes a distinguishing
formula in time $\CO(m n)$. The algorithm builds on the
Kanellakis-Smolka partition refinement
algorithm~\cite{KanellakisSmolka83,KanellakisS90}, which computes the
bisimilarity relation on a transition system within the same time
bound.

\begin{figure}[t]%
  \begin{minipage}{.45\textwidth}
    \hspace{1mm}
  \begin{tikzpicture}[coalgebra,x=1.5cm,baseline=(x.base)]
    \begin{scope}[every node/.append style={
        state,
        label distance=-.5mm,
        outer sep=3pt,
        inner sep=0pt,
        text depth=0pt,
        font=\normalsize,%
      }
      ]
      \node[label=left:$x$] (x)  at (0, 0) {$\bullet$};
      \node (x1)  at (1, 0) {$\bullet$};
      \node[label=right:$y$] (y)  at (3, 0) {$\bullet$};
      \node[] (z)  at (2, 0) {$\bullet$};
    \end{scope}
    \path[path with edges,overlay]
      (x) edge (x1)
      (x) edge[out=65,in=115,looseness=6.5,overlay] (x)
      (x1) edge[out=65,in=115,looseness=6.5,overlay] (x1)
      (x1) edge (z)
      (y) edge (z)
      (y) edge[out=65,in=115,looseness=6.5,overlay] (y)
      ;
  \end{tikzpicture}
    \caption{Example of a transition system}
    \label{fig:exampleTS}
  \end{minipage}%
  \hfill
  \begin{minipage}{.45\textwidth}
  \begin{tikzpicture}[coalgebra,x=1.5cm,baseline=(x.base)]
    \begin{scope}[every node/.append style={
        state,
        label distance=-.5mm,
        outer sep=3pt,
        inner sep=0pt,
        text depth=0pt,
        font=\normalsize,%
      }
      ]
      \node[label=right:$y$] (y)  at (1, 0) {$\bullet$};
      \node[label=left:$x$] (x)  at (-2, 0) {$\bullet$};
      \node[] (z1)  at (0, 0) {$\bullet$};
      \node[] (z2)  at (-1, 0) {$\bullet$};
    \end{scope}
    \path[path with edges,overlay,every node/.append style={description}]
      (y) edge node {1} (z1)
      (x) edge[bend right=10] node {0.5} (z2)
      (z2) edge[bend right=10] node {1} (z1)
      (x) edge[bend left=20] node {0.5} (z1)
      ;
  \end{tikzpicture}
    \caption{Example of a Markov chain}
    \label{fig:exampleMarkov}
  \end{minipage}%
\end{figure}

Similar logical characterizations of bisimulation exist for other
system types. For instance, Desharnais
\etal~\cite{DesharnaisEA98,desharnaisEA02} characterize probabilistic
bisimulation on (labelled) Markov chains, in the sense of Larsen and
Skou~\cite{LarsenS91} (for each label, every state has either no
successors or a probability distribution on successors).  In their
logic, a formula $\Diamond_{\ge p}\phi$ holds at states that have a
transition probability of at least $p$ to states satisfying
$\phi$. For example, the state $x$ %
in \autoref{fig:exampleMarkov}
satisfies~$\Diamond_{\ge 0.5}\Diamond_{\ge 1}\top$ but $y$ does
not. Desharnais \etal~provide an algorithm that computes
distinguishing formulae for labelled Markov chains in run time
(roughly) $\CO(n^4)$.

\takeout{} %

In the present work, we construct such counterexamples
generically for a variety of system types. We achieve genericity over
the system type by modelling state-based systems as coalgebras for a
set functor in the framework of universal
coalgebra~\cite{Rutten00}. Examples of coalgebras for a set functor include
transition systems, deterministic automata, or weighted systems
(e.g.~Markov chains). Universal coalgebra provides a generic notion of
behavioural equivalence that instantiates to standard notions for
concrete system types, e.g.~bisimilarity (transtion systems), language
equivalence (deterministic automata), or probabilistic bisimilarity (Markov
chains).  Moreover, coalgebras come equipped with a generic notion of
modal logic that is parametric in a choice of modalities whose
semantics is constructed so as to guarantee invariance
w.r.t.~behavioural equivalence; under easily checked conditions, such
a \emph{coalgebraic modal logic} in fact characterizes behavioural
equivalence in the same sense as Hennessy-Milner logic characterizes
bisimilarity~\cite{Pattinson04,Schroder08}. Hence, as soon as suitable
modal operators are found, coalgebraic modal formulae serve as
distinguishing formulae.

In a nutshell, the contribution of the present paper is an algorithm
that computes distinguishing formulae for behaviourally inequivalent
states in \emph{quasilinear time}, and in fact \emph{certificates} that uniquely describe
behavioural equivalence classes in a system, in coalgebraic
generality. We build on an existing efficient coalgebraic partition
refinement algorithm~\cite{concurSpecialIssue}, thus achieving run
time $\CO(m\log n)$ on coalgebras with~$n$ states and~$m$ transitions
(in a suitable encoding). The dag size of formulae is
also $\CO(m\log n)$ (for tree size, exponential lower bounds are
known~\cite{FigueiraG10}); even for labelled
transition systems, we thus improve the previous best bound
$\CO(m n)$~\cite{Cleaveland91} for both run time and formula
size.\twnote{} We systematically extract the requisite modalities from
the functor at hand, requiring binary and nullary modalities in the
general case, and then give a systematic method to translate these
generic modal operators into more customary ones (such as the standard
operators of Hennessy-Milner logic).

We subsequently identify a notion of \emph{cancellative} functor that
allows for additional optimization. E.g.~functors modelling weighted
systems are cancellative if and only if the weights come from a
cancellative monoid, such as $(\Z,+)$, or $(\R,+)$ as used in
probabilistic systems.  For cancellative functors, much simpler
distinguishing formulae can be constructed: the
binary modalities can be replaced by unary ones, and only conjunction
is needed in the propositional base. On labelled Markov chains, this
complements the result that a logic with only conjunction and
different unary modalities (mentioned above) suffices for the
construction of distinguishing formulae (but not
certificates)~\cite{desharnaisEA02} (see
also~\cite{Doberkat09}).\twnote{}

\takeout{}%

\subparagraph*{Related Work} Cleaveland's
algorithm~\cite{Cleaveland91} for labelled transition systems is is
based on Kanellakis and Smolka's partition refinement
algorithm~\cite{KanellakisS90}. The coalgebraic partition refinement
algorithm we employ~\cite{concurSpecialIssue} is instead related to the more
efficient Paige-Tarjan algorithm~\cite{PaigeTarjan87}. %
König et al.~\cite{KoenigEA20} extract formulae from winning
strategies in a bisimulation game in coalgebraic generality; their
algorithm runs in $\CO(n^4)$ and does not support negative transition
weights.  Characteristic formulae for behavioural equivalence classes
taken across \emph{all} models require the use of fixpoint
logics~\cite{DorschEA18}.  The mentioned algorithm by Desharnais~et
al.~for distinguishing formulae on labelled Markov
processes~\cite[Fig.~4]{desharnaisEA02} %
is based on Cleaveland's. No complexity analysis is made but the
algorithm has four nested loops, so its run time is roughly
$\CO(n^4)$. Bernardo and Miculan~\cite{BernardoMiculan19} provide a
similar algorithm for a logic with only disjunction. There are further
generalizations along other axes, e.g.~to behavioural
preorders~\cite{CelikkanCleaveland95}. The TwoTowers tool set for the
analysis of stochastic process algebras~\cite{BernardoEA98,Bernardo04}
computes distinguishing formulae for inequivalent processes, using
variants of Cleaveland's algorithm. Some approaches construct
alternative forms of certificates for inequivalence, such as Cranen et
al.'s notion of evidence~\cite{cranen_et_al:LIPIcs:2015:5408} or
methods employed on business process models, based on model
differences and event
structures~\cite{Dijkman08,ArmasCervantesEA13,ArmasCervantesEA14}.

\section{Preliminaries}
\label{preliminaries}
We first recall some basic notation. We denote by $0=\emptyset$,
$1=\{0\}$, $2 = \{0,1\}$ and $3 = \{0,1,2\}$ the sets representing the
natural numbers $0$, $1$, $2$ and $3$. For every set~$X$, there is a
unique map $!\colon X\to 1$.  We write~$Y^X$ for the set of functions
$X\to Y$, so e.g.~$X^2\cong X\times X$.\twnote{} In particular,~$2^X$ is the set of $2$-valued
\emph{predicates} on~$X$, which is in bijection with the
\emph{powerset}~$\Pow X$ of~$X$, i.e.~the set of all subsets of~$X$;
in this bijection, a subset $A\in\Pow X$ corresponds to its
\emph{characteristic function} $\chi_A\in 2^X$, given by $\chi_A(x)=1$
if $x\in A$, and $\chi(x)=0$ otherwise.  We generally indicate
injective maps by~$\monoto$. Given maps $f\colon Z\to X$,
$g\colon Z\to Y$, we write~$\fpair{f,g}$ for the map $Z\to X\times Y$
given by $\fpair{f,g}(z) = (f(z), g(z))$.  We denote the disjoint
union of sets~$X$,~$Y$ by $X+Y$, with canonical inclusion maps
$\inl\colon X\monoto X+Y$ and $\inr\colon Y\monoto X+Y$. More
generally, we write $\coprod_{i\in I} X_i$ for the disjoint union of
an $I$-indexed family of sets $(X_i)_{i\in I}$, and
$\inj_i\colon X_i\monoto \coprod_{i\in I} X_i$ for the $i$-th
inclusion map. For a map $f\colon X\to Y$ (not necessarily
surjective), we denote by $\ker(f) \subseteq X\times X$ the
\emph{kernel} of $f$, i.e.~the equivalence relation
\begin{equation}
  \ker(f) := \{ (x,x') \in X\times X \mid f(x) = f(x')\}.
  \label{eqKer}  
\end{equation}

\begin{notation}[Partitions] 
  \label{eqEqClass}
  Given an equivalence relation~$R$ on~$X$, we write $[x]_R$ for the
  equivalence class $\{ x'\in X\mid (x,x')\in R\}$ of $x\in X$. If~$R$
  is the kernel of a map $f$, we simply write $[x]_{f}$ in lieu of
  $[x]_{\ker(f)}$. The intersection $R\cap S$ of equivalence relations
  is again an equivalence relation. The partition corresponding to~$R$
  is denoted by $X/R = \{[x]_R\mid x\in X\}$. Note that
  $[-]_R\colon X\to X/R$ is a surjective map and that
  $R = \ker([-]_R)$.
\end{notation}
\noindent A \emph{signature} is a set $\Sigma$, whose elements are
called \emph{operation symbols}, equipped with a function
$a\colon \Sigma\to \N$ assigning to each operation symbol its
\emph{arity}. We write $\arity{\sigma}{n}\in \Sigma$ for
$\sigma\in \Sigma$ with $a(\sigma) = n$. We will apply the same
terminology and notation to collections of modal operators.

\subsection{Coalgebra} \emph{Universal coalgebra}~\cite{Rutten00}
provides a generic framework for the modelling and analysis of
state-based systems. Its key abstraction is to parametrize notions and
results over the transition type of systems, encapsulated as an
endofunctor on a given base category. %
Instances cover, for example, deterministic automata, labelled (weighted) transition
systems, and Markov chains.
\begin{defn}
  A \emph{set functor} $F\colon \Set\to \Set$ assigns to every set $X$
  a set $FX$ and to every map $f\colon X\to Y$ a map
  $Ff\colon FX\to FY$ such that identity maps and composition are
  preserved: $F\id_X = \id_{FX}$ and $F(g\cdot f) = Fg\cdot Ff$.  An
  \emph{$F$-coalgebra} is a pair $(C,c)$ consisting of a set $C$ (the
  \emph{carrier}) and a map $c\colon C\to FC$ (the
  \emph{structure}). When $F$ is clear from the context, we simply speak of a
  \emph{coalgebra}.
\end{defn}
\noindent In a coalgebra $c\colon C\to FC$, we understand the carrier
set~$C$ as consisting of \emph{states}, and the structure~$c$ as
assigning to each state $x\in C$ a structured collection of successor
states, with the structure of collections determined by~$F$. In this
way, the notion of coalgebra subsumes numerous types of state-based
systems, as illustrated next.

\begin{expl}\label{ex:powerset}
  \begin{enumerate}
  \item The \emph{powerset functor} $\Pow$ sends a set~$X$ to its
    powerset $\Pow X$ and a map $f\colon X\to Y$ to the map
    $\Pow f = f[-]\colon \Pow X\to \Pow Y$ taking direct images. A
    $\Pow$-coalgebra $c\colon C\to \Pow C$ is precisely a transition
    system: It assigns to every state $x\in C$ a set $c(x) \in \Pow C$
    of \emph{successor} states, inducing a transition relation~$\to$
    given by $x\to y$ iff $y\in c(x)$.
    Similarly, coalgebras for the finite powerset functor
    $\Powf$ (with~$\Powf X$ being the set of finite subsets of~$X$)
    are finitely branching transition systems.

  \item Coalgebras for the functor $FX=2\times X^A$, where $A$ is a
    fixed input alphabet, are deterministic automata (without an
    explicit initial state). Indeed, a coalgebra structure
    $c = \langle f, t\rangle\colon C\to 2\times C^A$ consists of a
    finality predicate $f\colon C \to 2$ and a transition map $C
    \times A \to C$ in curried form $t\colon C\to C^A$.

  \item\label{exSignature} Every signature $\Sigma$ defines a
    \emph{signature functor} that maps a set~$X$ to the set
    \[\textstyle
      T_\Sigma X = \coprod_{\arity[\scriptstyle]{\sigma}{n} \in \Sigma} X^n,
    \]
    whose elements we may understand as flat $\Sigma$-terms
    $\sigma(x_1,\ldots,x_n)$ with variables from~$X$. The action of
    $T_\Sigma$ on maps $f\colon X\to Y$ is then given by
    $(T_\Sigma f) (\sigma(x_1,\ldots,x_n)) = \sigma(f(x_1), \ldots,$
    $f(x_n))$. For simplicity, we write $\sigma$ (instead of
    $\inj_\sigma$) for the coproduct injections, and~$\Sigma$ in lieu
    of $T_\Sigma$ for the signature functor.
    States in $\Sigma$-coalgebras describe possibly infinite
    $\Sigma$-trees.

  \item For a commutative monoid $(M,+,0)$, the \emph{monoid-valued functor}
  $M^{(-)}$~\cite{GummS01} is given by
  \begin{equation}
    M^{(X)} := \{ \mu \colon X\to M\mid \text{$\mu(x) = 0$ for all but finitely many
      $x\in X$}\}
    \label{eqFinSupport}
  \end{equation}
  on sets $X$; for a map $f\colon X\to Y$, the map
  \(M^{(f)}\colon M^{(X)}\to M^{(Y)} \) is defined by
  \begin{equation*}
    (M^{(f)})(\mu)(y) = \textstyle\sum_{x\in X, f(x) = y} \mu(x).
  \end{equation*}
  A coalgebra $c\colon C\to M^{(C)}$ is a finitely branching weighted
  transition system, where $c(x)(x')\in M$ is the transition weight
  from $x$ to~$x'$. For the Boolean monoid
  $\BoolMonoid = (2,\vee, 0)$, we recover $\Powf =
  \BoolMonoid^{(-)}$. Coalgebras for $\R^{(-)}$, with~$\R$ understood
  as the additive monoid of the reals, are $\R$-weighted transition
  systems. The functor
  \begin{equation*}\textstyle
    \Dist X = \{\mu\in \R_{\ge 0}^{(X)}\mid \sum_{x\in X}\mu(x) = 1\},
  \end{equation*}
  which assigns to a set $X$ the set of all finite probability
  distributions on $X$ (represented as finitely supported probability
  mass functions), is a subfunctor of~$\R^{(-)}$.

\item Functors can be composed; for instance, given a set~$A$ of
  labels, the composite of $\Pow$ and the functor $A\times (-)$ (whose
  action on sets maps a set~$X$ to the set $A\times X$) is the functor
  $FX=\Pow(A\times X)$, whose coalgebras are $A$-labelled transition
  systems. Coalgebras for $(\Dist(-)+1)^A$ have been termed
  \emph{probabilistic transition systems}~\cite{LarsenS91} or
  \emph{labelled Markov chains}~\cite{desharnaisEA02}, and coalgebras
  for $(\Dist((-)+1)+1)^A$ are \emph{partial labelled Markov
    chains}~\cite{desharnaisEA02}.  Coalgebras for
  $SX = \Powf(A\times \Dist X)$ are variously known as \emph{simple
    Segala systems} or \emph{Markov decision processes}.
  \end{enumerate}
\end{expl}

\takeout{}%

\noindent We have a canonical notion of \emph{behaviour} on $F$-coalgebras:

\smallskip
\noindent
\begin{minipage}[T]{.82\textwidth}
\begin{defn}
  An $F$-coalgebra \emph{morphism} $h\colon (C,c)\to (D,d)$ is a map
  $h\colon C\to D$ such that $d\cdot h = Fh\cdot c$. States $x,y$ in
  an $F$-coalgebra $(C,c)$ are \emph{behaviourally equivalent}
  ($x\sim y$) if there exists a coalgebra morphism $h$ such that
  $h(x) = h(y)$.
\end{defn}
\end{minipage}%
  \hfill
  \begin{minipage}[T]{.18\textwidth}
    \vspace{-5mm}
    \hfill
    \begin{tikzcd}[sep=5mm,baseline=(C.base),outer sep=0pt]
      |[alias=C]|
      C
      \arrow{r}{c}
      \arrow[swap]{d}{h}
      & FC
      \arrow{d}{Fh}
      \\
      D
      \arrow{r}{d}
      & FD
    \end{tikzcd}
  \end{minipage}\medskip

  \noindent Thus, we effectively define the behaviour of a state as
  those of its properties that are preserved by coalgebra
  morphisms. The notion of behavioural equivalence subsumes standard
  branching-time equivalences:
\begin{expl} \label{exCoalg}
  \begin{enumerate}
  \item \label{coalgts} For $F\in\{\Pow,\Powf\}$, behavioural
    equivalence on $F$-coalgebras, i.e.~on transition systems, is
    \emph{bisimilarity} in the usual sense.
    
  \item For deterministic automata as coalgebras for $FX = 2 \times
    X^A$, two states are behaviourally equivalent iff they accept the same
    formal language.

  \item\label{coalgSignature} For a signature functor $\Sigma$, two
    states of a $\Sigma$-coalgebra are behaviourally equivalent iff
    they describe the same $\Sigma$-tree.

  \item For labelled transition systems as coalgebras for
    $FX = \Pow(A \times X)$, coalgebraic behavioural equivalence
    precisely captures Milner's strong
    bisimilarity~\cite{AczelMendler89}.\twnote{}

  \item For weighted and probabilistic systems, coalgebraic
    behavioural equivalence instantiates to weighted and probabilistic
    bisimilarity, respectively~\cite[Cor.~4.7]{RuttenDV99}, \cite[Thm.~4.2]{BartelsEA04}.

  \end{enumerate}
\end{expl}

\begin{rem}
  \label{behEqDifferntCoalg}\label{trnkovahull}
  \begin{enumerate}
  \item The notion of behavioural equivalence extends
    straightforwardly to states in different coalgebras, as one can
    canonically define the disjoint union of coalgebras.
  \item We may assume without loss of generality that a set
    functor~$F$ preserves injective maps~\cite{trnkova71} (see
    also~\cite[8.1.12--17]{AdamekEA21}), that is,~$Ff$ is injective
    whenever~$f$ is.  \smnote{}
  \end{enumerate}
\end{rem}

\subsection{Coalgebraic Logics}
\label{sec:coalgLogic}

\takeout{}  We briefly review basic
concepts of coalgebraic modal
logic~\cite{Pattinson03,Schroder08}. Coalgebraic modal logics are
parametric in a functor~$F$ determining the type of systems underlying
the semantics, and additionally in a choice of modalities interpreted
in terms of \emph{predicate liftings}. For now, we use $F = \Pow$ as a
basic example, deferring further examples to \autoref{domainSpecific}.

\subparagraph*{Syntax} The syntax of coalgebraic modal logic is
parametrized over the choice of signature $\Lambda$ of \emph{modal
  operators} (with assigned arities). Then, \emph{formulae} $\phi$ are
generated by the grammar\\[2mm]
\(
{\hspace{5mm}}
    \phi_{1},\ldots,\phi_{n} ::= \top
    ~\vert~ \neg \phi_1
    ~\vert~ \phi_1 \wedge \phi_2
    ~\vert~ \hearts(\phi_1,\ldots,\phi_n)
    \qquad
    (\arity{\hearts}{n} \in \Lambda).
\)
\begin{expl} \label{exPowModalities} For $F = \Pow$, one often takes
  $\Lambda=\{\arity{\Diamond}{1}\}$; the induced syntax is that of
  (single-action) Hennessy-Milner logic. As usual, we write $\Box \phi :\equiv
  \neg \Diamond \neg \phi$.
\end{expl}

\subparagraph*{Semantics} We interpret formulae as sets of states in
$F$-coalgebras. This interpretation arises by assigning
to each modal operator $\arity{\hearts}{n}\in \Lambda$ an $n$-ary
\emph{predicate
  lifting}~$\semantics{\hearts}$~\cite{Pattinson03,Schroder08}, i.e.~a
family of maps $\semantics{\hearts}_X\colon (2^{X})^n \to 2^{FX}$, one
for every set~$X$, such that the \emph{naturality} condition
\begin{equation}\label{eq:naturality}
  Ff^{-1}\big[\semantics{\hearts}_Y(P_1,\ldots,P_n)\big]
  = \semantics{\hearts}_X(f^{-1}[P_1],\ldots,f^{-1}[P_n])
\end{equation}
for all $f\colon X\to Y$ and all $P_1,\ldots,P_n\in 2^X$ (for
categorically-minded readers,~$\semantics{\hearts}$ is a natural
transformation $(2^{(-)})^n\to 2^{F^{\mathsf{op}}}$); the idea being
to lift given predicates on states to predicates on structured
collections of states. Given these data, the \emph{extension} of a
formula $\phi$ in an~$F$-coalgebra $(C,c)$ is a predicate
$\semantics{\phi}_{(C,c)}$, or just~$\semantics{\phi}$, on~$C$,
recursively defined by
  \begin{align*}
    & \semantics{\top}_{(C,c)} = C
      \qquad
    \semantics{\phi\wedge\psi}_{(C,c)} = \semantics{\phi}_{(C,c)} \cap \semantics{\psi}_{(C,c)} 
      \qquad
    \semantics{\neg\phi}_{(C,c)} = C\setminus\semantics{\phi}_{(C,c)}
    \\
    &\semantics{\hearts(\phi_1,\ldots,\phi_n)}_{(C,c)}
    = c^{-1}\big[\semantics{\hearts}_C\big(\semantics{\phi_1}_{(C,c)}, \ldots, \semantics{\phi_n}_{(C,c)}\big)\big]
    \qquad\text{($\arity{\hearts}{n}\in \Lambda$)}
  \end{align*}
  (where we apply set operations to predicates with the evident
  meaning). We say that a state~$x\in C$ \emph{satisfies}~$\phi$ if
  $\semantics{\phi}(x)=1$. Notice how the clause for modalities says
  that $x$ satisfies $\hearts(\phi_1,\ldots,\phi_n)$ iff $c(x)$
  satisfies the predicate obtained by lifting the predicates
  $\semantics{\phi_1}, \ldots, \semantics{\phi_n}$ on~$C$ to a
  predicate on $FC$ according to~$\semantics{\hearts}$.

  \begin{expl} \label{exPowPredLift} Over $F=\Pow$, we
    interpret~$\Diamond$ by the predicate lifting
    \[
    \semantics{\Diamond}_X\colon 2^X \to 2^{\Pow X},\quad
    P\mapsto~ \{ K \subseteq X \mid \exists x \in K\colon x\in P\}
                = \{K\subseteq X \mid K\cap P\neq \emptyset\},
                \label{eqDiamondLifting}
    \]
  The arising notion of satisfaction over $\Pow$-coalgebras $(C,c)$ is
  precisely the standard one:
    $x\in \semantics{\Diamond \phi}_{(C,c)}$
    iff $y\in \semantics{\phi}_{(C,c)}$ for some transition $
      x\to y$.
\end{expl}
The naturality condition~\eqref{eq:naturality} of predicate liftings
guarantees invariance of the logic under coalgebra morphisms, and
hence under behavioural equivalence:
\begin{proposition}[Adequacy~\cite{Pattinson03,Schroder08}]%
  \label{bisimInvariant}
  Behaviourally equivalent states satisfy the same formulae: $x\sim y$
  implies that for all formulae~$\phi$, we have $x\in \semantics{\phi}$
  iff $y\in \semantics{\phi}$.  \smnote{}
\end{proposition}
\noindent In our running example $F=\Pow$, this instantiates to the
well-known fact that modal formulae are bisimulation-invariant, that
is, bisimilar states in transition systems satisfy the same formulae
of Hennessy-Milner logic.

\section{Constructing Distinguishing Formulae}
\label{sec:main}

A proof method certifying behavioural equivalence of states $x,y$ in a
coalgebra is immediate by definition: One simply needs to exhibit a
coalgebra morphism $h$ such that $h(x) = h(y)$. In fact, for many
system types, it suffices to relate~$x$ and~$y$ by a coalgebraic
\emph{bisimulation} in a suitable sense
(e.g.~\cite{AczelMendler89,Rutten00,GorinSchroeder13,MartiVenema15}), generalizing
the Park-Milner bisimulation principle~\cite{Milner89,Park81}. It is
less obvious how to certify behavioural \emph{inequivalence}
$x\not\sim y$, showing that such a morphism~$h$ does \emph{not}
exist. By \autoref{bisimInvariant}, one option is to exhibit a
(coalgebraic) modal formula~$\phi$ that is satisfied by~$x$ but not
by~$y$. In the case of (image-finite) transition systems, such a
formula is guaranteed to exist by the Hennessy-Milner theorem, which
moreover is known to generalize to
coalgebras~\cite{Pattinson04,Schroder08}. More generally, we consider
separation of \emph{sets} of states by formulae, following
Cleaveland~\cite[Def.~2.4]{Cleaveland91}:
\begin{defn}\label{defDistinguish}
  Let $(C,c)$ be an $F$-coalgebra. A formula $\phi$
  \emph{distinguishes} a set $X\subseteq C$ from a set $Y\subseteq C$
  if $X \subseteq \semantics{\phi}$ and
  $Y \cap \semantics{\phi} = \emptyset$. In case $X = \{x\}$ and
  $Y = \{y\}$, we just say that \emph{$\phi$ distinguishes $x$ from
    $y$}. We say that $\phi$ is a \emph{certificate} of~$X$ if~$\phi$
  distinguishes~$X$ from $C\setminus X$, that is if
  $\semantics{\phi}=X$.
\end{defn}
Note that~$\phi$ distinguishes~$X$ from~$Y$ iff $\neg\phi$
distinguishes~$Y$ from~$X$. Certificates have also been referred to as
\emph{descriptions}~\cite{FigueiraG10}.  If $\phi$ is a certificate of
a behavioural equivalence class $[x]_\sim$, then by definition~$\phi$
distinguishes~$x$ from $y$ whenever $x \not\sim y$. To obtain
distinguishing formulae for behaviourally inequivalent states in a
coalgebra, it thus suffices to construct certificates for all
behavioural equivalence classes, and our algorithm does just that. Of
course, every certificate must be at least as large as a smallest
distinguishing formula. However, already
on transition systems, distinguishing formulae and certificates
have the same asymptotic worst-case size (cf.~\autoref{worstcase}).

A natural approach to computing certificates for behavioural
equivalence classes is to extend algorithms that compute these
equivalence classes. In particular, \emph{partition refinement}
algorithms compute a sequence $C/R_0, C/R_1,\ldots$ of consecutively
finer partitions (i.e.~$R_{i+1}\subseteq R_i$) on the state space,
where every \emph{block} $B\in C/R_i$ is a union of behavioural
equivalence classes, and the final partition is precisely
$C/\mathord{\sim}$. Indeed, Cleaveland's algorithm for computing
certificates on labelled transition systems~\cite{Cleaveland91}
correspondingly extends Kanellakis and Smolka's partition refinement
algorithm~\cite{KanellakisSmolka83,KanellakisS90}, which runs in
$\CO(m n)$ on systems with $n=|C|$ states and $m$ transitions. Our
generic algorithm will be based on a more efficient partition
refinement algorithm.

\subsection{Paige-Tarjan with Certificates}
\label{paigeTarjan}
Before we turn to constructing certificates in coalgebraic generality,
we informally recall and extend the Paige-Tarjan
algorithm~\cite{PaigeTarjan87}, which computes the partition modulo
bisimilarity of a given transition system with $n$ states and $m$
transitions in time $\CO((m+n) \log n)$. We fix a given finite
transition system, viewed as a $\Pow$-coalgebra $c\colon C\to \Pow C$.

The algorithm computes two sequences $(C/P_i)_{i\in \N}$ and
$(C/Q_i)_{i\in \N}$ of partitions of~$C$ (with~$Q_i,P_i$ equivalence
relations), where only the most recent partition is held in memory
and~$i$ indexes the iterations of the main loop.  Throughout the
execution, $C/P_i$ is finer than $C/Q_i$ (that is, $P_i\subseteq Q_i$
for all~$i$), and the algorithm terminates when $P_i= Q_i$.
Intuitively,~$P_i$ is `one transition ahead' of $Q_i$: if $Q_i$
distinguishes states~$x$ and~$y$, then~$P_i$ is based on
distinguishing transitions to~$x$ from transitions to~$y$.

Initially, $C/Q_0 := \{C\}$ consists of only one block and $C/P_0$ of
two blocks: the live states and the deadlocks (i.e.~states with no
outgoing transitions). If $P_i\subsetneqq Q_i$, then there is a block
$B\in C/Q_i$ that is the union of at least two blocks in~$C/P_i$.
In such a situation, the algorithm chooses $S \subseteq B$ in $C/P_i$
to have at most half the size of $B$ and then splits the block $B$
into $S$ and $B\setminus S$ in the partition $C/Q_i$:
\[
  C/Q_{i+1} = (C/Q_i \setminus \{B\}) ~\cup~ \{S, B\setminus S\}.
\]
This is correct because every state in $S$ is already known to be
behaviourally inequivalent to every state in $B\setminus S$. By the
definition of bisimilarity, this implies that every block $T\in C/P_i$
with some transition to~$B$ may contain behaviourally inequivalent
states as illustrated in \autoref{fig:parttree}; that is,~$T$ may need
to be split into smaller blocks, as follows:

\newcommand{\caseref}[1]{\ref{#1}}
\begin{enumerate}[({C}1)]
\makeatletter
\renewcommand{\labelenumi}{\theenumi}
\renewcommand{\theenumi}{\text{\sffamily\color{lipicsGray}(C\arabic{enumi})}}
\renewcommand{\p@enumi}{}
\makeatother
\item\label{edgeS} states in~$T$ with successors in $S$ but not in $B\setminus
  S$ (e.g.~$x_1$ in \autoref{fig:parttree}),
\item\label{edgeBoth} states in~$T$ with successors in $S$ and $B\setminus S$ (e.g.~$x_2$), and
\item\label{edgeBS} states in~$T$ with successors $B\setminus S$ but not in $S$ (e.g.~$x_3$).
\end{enumerate}

\begin{figure}[t] %
  \centering
  \begin{tikzpicture}[
    partitionBlock/.style={
            shape=rectangle,
            rounded corners=2.5mm,
            minimum height=5mm,
            minimum width=5mm,
            draw=black!10,
            line width=1pt,
            inner sep = 2mm,
    }
]
  \begin{scope}[
    every node/.append style={
      inner sep=0pt,
      outer sep=2pt,
      minimum width=1pt,
      minimum height=4pt,
      anchor=center,
    }
    ]
    \begin{scope}[yshift=1.5cm]
      \node (x1) at (-1.2,0) {$x_1$};
      \node (x2) at (0,0) {$x_2$};
      \node (x3) at (1.2,0) {$x_3$};
    \end{scope}
    \node (y1) at (-1.6,0) {$y_1$};
    \node (y2) at (-0.5,0) {$y_2$};
    \node (y3) at (0.5,0) {$y_3$};
    \node (y4) at (1.2,0) {$y_4$};
  \end{scope}
  \node[anchor=west,xshift=2mm] (y5) at (y4.east) {$\ldots$};
  \draw (y1) edge[draw=none] node {$\ldots$} (y2);
  \begin{scope}[
    every node/.append style={
      partitionBlock,
      inner sep=2pt,
    },
    every label/.append style={
      font=\footnotesize,
      anchor=south,
      outer sep=1pt,
      inner sep=1pt,
      minimum height=2mm,
      shape=rectangle,
    },
    ]
    \node[fit=(x1) (x2) (x3)] (x123){};
    \node[fit=(y1) (y5)] (C) {};
    \node[every label,overlay] at ([xshift=-1mm]x123.north east) {\normalsize $T$}; 
    \node[every label] at ([xshift=-1mm]C.north east) {\normalsize $B$};
    \begin{scope}[every node/.append style={minimum width=10mm}, 
      ]
      \newcommand{\nextblockdistance}{8mm} 
      \foreach \name/\nodesource/\anchor/\direction in
      {Peast/x123/east/,
        Qeast/C/east/,
        Pwest/x123/west/-,
        Qwest/C/west/-} {
        \begin{scope}
          \clip ([yshift=1mm]\nodesource.north \anchor)
          rectangle ([yshift=-1mm,xshift=\direction\nextblockdistance]
                     \nodesource.south \anchor);
          \node at ([xshift=\direction\nextblockdistance]\nodesource.\anchor) (\name) {};
        \end{scope}
      }
    \end{scope}
  \end{scope}
  \begin{scope}
  \draw[thick, decoration={brace,mirror},decorate]
  ([yshift=-2mm]y1.south west) -- node[anchor=north,yshift=-2pt]{\footnotesize $S$} ([yshift=-2mm]y2.south east) ;
  \draw[thick, decoration={brace,mirror},decorate]
  ([yshift=-2mm]y3.south west) -- node[anchor=north,yshift=-2pt]{\footnotesize $B\setminus S$} ([yshift=-2mm]y3.south west -| y5.south east) ;
  \foreach \name in {Peast,Pwest,Qeast,Qwest} {
    \node[draw=none,minimum width=1pt,minimum height=1pt,text height=1pt,inner xsep=1pt]
       at (\name.center) {$\ldots$};
  }
  \end{scope}
  \begin{scope}[text depth=2pt]
  \node[outer sep=5mm,anchor=east] at (Pwest.center -| Qwest.center) (P) {$C/P:$};
  \node[outer sep=5mm,anchor=east] at (Qwest.center) (Q) {$C/Q:$};
  \coordinate (mapXCoordinate) at ([xshift=10mm]Qeast.center);
  \node[outer sep=1pt,anchor=center] at (Peast.center -| mapXCoordinate) (X) {$C$};
  \node[outer sep=1pt,anchor=center] at (Qeast.center -| mapXCoordinate) (PY) {$\Pow C$};
  \end{scope}
  \draw[commutative diagrams/.cd, every arrow, every label] (X) edge node {$c$} (PY);
  \begin{scope}[
    bend angle=10,
    space/.style={
      draw=white,
      line width=4pt,
    },
    edge/.style={
      draw=black,
      -{>[length=2mm,width=2mm]},
      preaction={draw,-,line width=1mm,white},
      every node/.append style={
        fill=white,
        shape=rectangle,
        inner sep=1pt,
        anchor=base,
        pos=0.55,
      },
    },
    ]
    \draw[edge,bend right] (x1) to (y1);
    \draw[edge,bend left] (x1) to (y2);
    \draw[edge,bend right] (x2) to (y2);
    \draw[edge,bend left] (x2) to (y3);
    \draw[edge,bend left] (x2) to (y4);
    \draw[edge,bend right] (x3) to (y3);
    \draw[edge] (x3) to (y4);
    \draw[edge,bend left] (x3) to (y5);
  \end{scope}
  \begin{scope}[
    linestyle/.style={
      dashed,
      draw=black!50,
    },
    scissors/.style={
      text=black!50,
      font=\tiny,
      inner sep=0pt,
      overlay,
    },
    gapstyle/.style={
      draw=white,
      line width=1mm,
    }
    ]
    \foreach \leftnode/\rightnode/\northlen/\southlen in
    {x1/x2/4mm/4mm,x2/x3/4mm/4mm,y2/y3/4mm/4mm} {
      \coordinate (northend) at ($ (\leftnode) !.5! (\rightnode) + (0,\northlen)$);
      \coordinate (southend) at ($ (\leftnode) !.5! (\rightnode) - (0,\southlen)$);
      \node[anchor=east,rotate=-90,scissors] at (northend) {\scissors};
      \draw[gapstyle] (northend) -- (southend);
      \draw[linestyle] (northend) -- (southend);
    }
  \end{scope}
\end{tikzpicture}
  \caption{The refinement step as illustrated in \cite[Figure 6]{concurSpecialIssue}.}
  \label{fig:parttree}
\end{figure}

\noindent The partition $C/P_{i+1}$ arises from $C/P_i$ by splitting
all such predecessor blocks~$T$ of $B$ accordingly. If no such~$T$ is
properly split, then $P_{i+1}=Q_{i+1}$, and the algorithm terminates.
\twnote{}
It is straightforward to construct certificates for the blocks arising
during the execution:
\begin{itemize}
\item
  The certificate for the only block $C \in C/Q_0$ is $\top$, and the
  blocks for live states and deadlocks in $C/P_0$ have certificates
  $\Diamond\top$ and~$\neg\Diamond\top$, respectively.
  
\item %
  In the refinement step, suppose that $\delta, \beta$ are
  certificates of $S\in C/P_i$ and $B\in C/Q_i$, respectively, where
  $S\subsetneqq B$. For every predecessor block $T$ of $B$, the three
  blocks obtained by splitting $T$ are distinguished
  (see~\autoref{defDistinguish}) as follows:
  \begin{equation}
    \text{
    \caseref{edgeS}\quad $\neg\Diamond(\beta\wedge\neg\delta)$,
    \qquad\caseref{edgeBoth}\quad $\Diamond(\delta) \wedge \Diamond(\beta\wedge\neg \delta)$,
    \qquad\caseref{edgeBS}\quad $\neg\Diamond \delta$.
    }
    \label{edgeCasesFormula}
  \end{equation}
  Of course these formulae only distinguish the states in $T$ from each other
  (e.g.~there may be states in other blocks with transitions to both
  $S$ and~$B$). Hence, given a certificate~$\phi$ of $T$, one obtains
  certificates of the three resulting blocks in $C/P_{i+1}$ via conjunction: $\phi\wedge
  \neg\Diamond(\beta\wedge \neg \delta)$, etc.
\end{itemize}
Upon termination, every bisimilarity class $[x]_\sim$ in the
transition system is annotated with a certificate. A key step in the
generic development will be to come up with a coalgebraic
generalization of the formulae for \refAllEdgeCases.

\subsection{Generic Partition Refinement}
\label{sec:genPartRef}

The Paige-Tarjan algorithm has been adapted to other system types,
e.g.~weighted systems~\cite{ValmariF10}, and it has recently been
generalized to co\-al\-gebras~\cite{concurSpecialIssue, DorschEA17}. A
crucial step in this generalization is to rephrase the case
distinction \refAllEdgeCases in terms of the
functor $\Pow$: Given a predecessor block $T$ in $C/P_i$ for
$S\subsetneqq B\in C/Q_i$, the three cases distinguish between the
equivalence classes \([x]_{\Pow \chi_S^B\cdot c}\) for $x\in T$, where
the map $\chi_S^B\colon C \to 3$ in the composite
$\Pow \chi_S^B \cdot c\colon C\to \Pow 3$ is defined by
\begin{equation}
  \chi_S^B\colon C\to 3
  \qquad
\chi_S^B(x) = 
  \begin{cases}
    2 & \text{if $x \in S$}, \\
    1 & \text{if $x \in B \setminus S$}, \\
    0 & \text{if $x \in C \setminus B$},
  \end{cases}
  \qquad\text{for sets $S\subseteq B\subseteq C$}.
  \label{eqChi3}
\end{equation}
Every case is a possible value of $t:=\Pow\chi_S^B(c(x)) \in \Pow 3$:
\caseref{edgeS}~$2\in t \notni 1$,
\caseref{edgeBoth}~$2\in t \ni 1$, and
\caseref{edgeBS}~$2\notin t \ni 1$. Since $T$ is a predecessor block of
$B$, the `fourth case' $2\not\in t \not\mkern1mu\ni 1$ is not
possible. There is a transition from $x$ to some state outside of $B$
iff $0\in t$. However, because of the previous refinement steps
performed by the algorithm, either all or no states states of $T$ have
an edge to $C\setminus B$ (a property called
\emph{stability}~\cite{PaigeTarjan87}), hence no distinction on
$0\in t$ is necessary.

It is now easy to generalize from transition systems to coalgebras by
simply replacing the functor $\Pow$ with $F$ in the refinement step. We recall the algorithm:

\begin{notheorembrackets}
\begin{algorithm}[{\cite[Alg.~4.9, (5.1)]{concurSpecialIssue}}] \label{coalgPT}
  Given a coalgebra $c\colon C\to FC$, put
  \[
  C/Q_0 := \{C\} \qquad\text{and}\qquad
  P_0 := \ker(C\xrightarrow{c}{FC}\xrightarrow{F!} F1).
  \]
  Starting at iteration $i=0$, repeat the following while $P_i\neq Q_i$:
  \begin{enumerate}[({A}1)]
\makeatletter
\renewcommand{\labelenumi}{\theenumi}
\renewcommand{\theenumi}{\text{\sffamily\color{lipicsGray}(A\arabic{enumi})}}
\renewcommand{\p@enumi}{}
\makeatother
    \item\label{step1} Pick $S\in C/P_i$ and $B\in C/Q_i$ such that $S\subsetneqq B$ and $2\cdot
      |S| \le |B|$
    \item\label{defQi1} $C/Q_{i+1} := (C/Q_i \setminus \{B\}) \cup \{S, B\setminus S\}$
    \item\label{defPi1} $P_{i+1} := P_i \cap \ker (
      \begin{tikzcd}[every label/.append style={inner sep=1pt}]
        C
        \arrow{r}[description]{c}
        &[2mm] FC
        \arrow{r}[description]{F\chi_S^B}
        &[5mm] F3
      \end{tikzcd}
      )$
  \end{enumerate}
\end{algorithm}
\end{notheorembrackets}

This algorithm formalizes the intuitive steps from
\autoref{paigeTarjan}. Again, two sequences of partitions~$P_1$,~$Q_i$ are
constructed, and $P_i = Q_i$
upon termination. Initially, $Q_0$ identifies all states and $P_0$
distinguishes states by only their output behaviour; e.g.~for $F=\Pow$
and~$x\in C$, the value $\Pow!(c(x)) \in \Pow 1$ is~$\emptyset$ if~$x$
is a deadlock, and~$\{1\}$ if $x$ is a live state, and for
$FX=2\times X^A$, the value $F1(c(x))\in F1= 2\times 1^A \cong 2$
indicates whether $x$ is a final or non-final state.

In the main loop, blocks $S\in C/P_i$ and $B\in C/Q_i$ witnessing
$P_i\subsetneqq Q_i$ are picked, and $B$ is split into $S$ and
$B\setminus S$, like in the Paige-Tarjan algorithm. Note that step \ref{defQi1}
is equivalent to directly defining the equivalence relation $Q_{i+1}$ as
\[
  Q_{i+1} := Q_i \cap \ker{\chi_S^B}.
\]
A similar intersection of equivalence relations is performed in step
\ref{defPi1}. The intersection splits every block $T\in
C/P_i$ into smaller blocks such that $x,x'\in T$ end up in the same block iff
$F\chi_S^B(c(x)) = F\chi_S^B(c(x'))$, i.e.~$T$ is replaced by
$\{[x]_{F\chi_S^B(c(x))}\mid x\in T\}$.
Again, this corresponds to the distinction of the three cases~\refAllEdgeCases.
For example, for
$FX=2\times X^A$, there are $|F3| = 2\cdot 3^{|A|}$ cases to be distinguished,
and so every $T\in C/P_i$ is split into at most
that many blocks. %

The following property of $F$ is needed for correctness~\cite[Ex.~5.11]{concurSpecialIssue}.
\begin{notheorembrackets}
\begin{defn}[{\cite{concurSpecialIssue}}]
  A functor $F$ is \emph{zippable} if map
  \begin{equation*}
    \fpair{F(A+!), F(!+B)}\colon~ F(A+B) \longrightarrow F(A+1) \times F(1+B)
  \end{equation*}
  is injective for all sets~$A, B$.
\end{defn}
\end{notheorembrackets}
Intuitively, $t\in F(A+B)$ is a term in variables from $A$ and $B$. If
$F$ is zippable, then~$t$ is uniquely determined by the two elements
in $F(A+1)$ and $F(1+B)$ obtained by identifying all $B$- and all
$A$-variables with $0\in 1$, respectively. E.g.~$FX=X^2$ is zippable:
$t=(\inl(a),\inr(b)) \in (A+B)^2$ is uniquely determined by
$(\inl(a),\inr(0)) \in (A+1)^2$ and $(\inl(0),\inr(b)) \in (1+B)^2$,
and similarly for the three other cases of~$t$. In fact, all signature
functors as well as $\Pow$ and all monoid-valued functors are
zippable. Moreover, the class of zippable functors is closed under
products, coproducts, and subfunctors but not under composition,
e.g.~$\Pow\Pow$ is not zippable~\cite{concurSpecialIssue}.

\begin{rem} \label{multisort} To apply the algorithm to coalgebras for
  composites $FG$ of zippable functors, e.g.~$\Pow(A\times (-))$,
  there is a reduction~\cite[Section~8]{concurSpecialIssue} that
  embeds every $FG$-coalgebra into a coalgebra for the zippable
  functor $(F+G)(X) := FX + GX$. This reduction preserves and reflects
  behavioural equivalence, but introduces an intermediate state for
  every transition.
\end{rem}

\begin{notheorembrackets}
\begin{theorem}[{\cite[Thm 4.20, 5.20]{concurSpecialIssue}}]
  On a finite coalgebra $(C,c)$ for a zippable functor,
  \autoref{coalgPT} terminates after $i\le |C|$ loop iterations, and
  the resulting partition identifies precisely the behaviourally
  equivalent states ($P_i = \mathord{\sim}$).
\end{theorem}
\end{notheorembrackets}

\subsection{Generic Modal Operators}\label{genericModalOp} The extended
Paige-Tarjan algorithm (\autoref{paigeTarjan}) constructs a
distinguishing formula according to the three cases
\refAllEdgeCases. In the coalgebraic \autoref{coalgPT}, these cases
correspond to elements of~$F3$, which determine in which block an
element of a predecessor block~$T$ ends up. Indeed, the elements
of~$F3$ will also serve as generic modalities in characteristic
formulae for blocks of states, essentially by the known equivalence
between~\mbox{$n$-ary} predicate liftings and (in this case,
singleton) subsets of $F(2^n)$~\cite{Schroder08} (termed \emph{tests}
by Klin~\cite{Klin05}):
\begin{defn} \label{defF3Mod}
  The signature of \emph{$F3$-modalities} for a functor $F$ is
  \[
    \Lambda = \{ \arity{\fmod{t}}{2} \mid t\in F3 \};
  \]
  that is, we write $\fmod{t}$ for the syntactic representation of a binary
  modality for every $t\in F3$. The interpretation of $\fmod{t}$ for $F3$ is
  given by the predicate lifting
  \[
    \semantics{\fmod{t}}\colon
    (2^X)^2 \to 2^{FX},
    \qquad
    \semantics{\fmod{t}}(S,B) = \{t'\in FX\mid
    F\chi_{S\cap B}^B(t') = t
    \}.
  \]
\end{defn}
The intended use of $\fmod{t}$ is as follows: Suppose a block~$B$ is
split into subblocks $S\subseteq B$ and $B\setminus S$ with
certificates $\delta$ and $\beta$, respectively:
$\semantics{\delta} = S$ and $\semantics{\beta}= B$. As in
\autoref{fig:parttree}, we then split every predecessor block $T$ of
$B$ into smaller parts, each of which is uniquely characterized by the
formula $\fmod{t}(\delta,\beta)$ for some $t\in F3$. \smnote{}

\begin{example} \label{examplePowF3Mod}
  For $F=\Pow$, $\fmod{\set{0,2}}(\delta,\beta)$ is equivalent to
  $\smash{\overbrace{\Diamond \neg \beta}^{\text{\textqt{0}}}}
  \wedge
  \neg \smash{\overbrace{\Diamond (\beta\wedge\neg\delta)}^{\text{\textqt{1}}}}
  \wedge \smash{\overbrace{\Diamond(\delta\wedge\beta)}^{\text{\textqt{2}}}}$.
\end{example}

\begin{lemma} \label{lemF3CoalgSem} Given an $F$-coalgebra $(C,c)$,
  $x \in C$, and formulae $\delta$ and $\beta$ such that
  $\semantics{\delta}\subseteq \semantics{\beta}\subseteq C$,  we
  have
  $x \in \semantics{\fmod{t}(\delta,\beta)}$ if and only if
  $F\chi_{\semantics{\delta}}^{\semantics{\beta}}(c(x)) = t$.
\end{lemma}
In the initial partition $C/P_0$ on a transition system $(C,c)$, we
used the formulae $\Diamond\top$ and~$\neg\Diamond\top$ to distinguish
live states and deadlocks. In general, we can similarly describe the
initial partition using modalities induced by elements of~$F1$:

\begin{notation} \label{notationF1Mod}
  Define the injective map $j_1\colon 1\monoto 3$ by $j_1(0) = 2$. Then
  the injection $Fj_1\colon F1\monoto F3$ provides a way to interpret elements $t\in F1$ as 
  nullary modalities $\fmod{t}$:
  \[
    \fmod{t} := \fmod{Fj_1(t)}(\top,\top)
    \qquad\text{for $t\in F1$.}
  \]
  (Alternatively, we could introduce $\fmod{t}$ directly as a nullary
  modality.)
\end{notation}
\begin{lemma} \label{lemF1CoalgSem}
  For $x\in C$, $c\colon C\to FC$, and $t\in F1$, we have
  $x\in \semantics{\fmod{t}}$ if and only if $F!(c(x)) = t$.
\end{lemma}

\subsection{Algorithmic Construction of Certificates}
\label{certConstruct}
The $F3$-modalities introduced above (\autoref{defF3Mod}) induce an
instance of coalgebraic modal logic (\autoref{sec:coalgLogic}). We
refer to coalgebraic modal formulae employing the $F3$-modalities as
\emph{$F3$-modal formulae}, and write $\Formulae$ for the set of
$F3$-modal formulae. As in the extended Paige-Tarjan algorithm
(\autoref{paigeTarjan}), we annotate every block arising during the
execution of \autoref{coalgPT} with a certificate in the shape of an
$F3$-modal formula.  Annotating blocks with formulae means that we
construct maps
\[
  \beta_i\colon C/Q_i \to \Formulae
  \qquad\text{and}\qquad
  \delta_i\colon C/P_i \to \Formulae
  \qquad
  \text{for $i \in \N$}.
\]
As in \autoref{coalgPT}, $i$ indexes the loop iterations. For
blocks $B, S$ in the respective par\-ti\-tion,~$\beta_i(B)$,
$\delta_i(S)$ denote corresponding certificates: we will have
\twnote{}%
\takeout{}
\begin{equation}
  \forall B\in X/Q_i\colon
  \semantics{\beta_i(B)} = B
  \qquad\text{and}\qquad
  \forall S\in X/P_i\colon
  \semantics{\delta_i(S)}
  = S.
  \label{eqCertCorrect}
\end{equation}
We construct $\beta_i(B)$ and $\delta_i(S)$ iteratively, using certificates for the blocks
$S\subsetneqq B$ at every iteration:

\begin{algorithm} \label{algoCerts}
  We extend \autoref{coalgPT} by the following. Initially, put
  \[
    \beta_0(\{C\}) := \top
    \qquad\text{and}\qquad
    \delta_0([x]_{P_0}) := \fmod{F!(c(x))}\quad\text{for every $x \in C$.}
  \]
  In the $i$-th iteration, extend steps \ref{defQi1} and  \ref{defPi1}
  by the following assignments:\medskip
  \begin{enumerate}[({A$\!'$}1)]
\makeatletter
\renewcommand{\labelenumi}{\theenumi}
\renewcommand{\theenumi}{\text{\sffamily\color{lipicsGray}(A$\!$'\arabic{enumi})}}
\renewcommand{\p@enumi}{}
\makeatother
  \refstepcounter{enumi}
  \item\label{defBetai1} $\mathrlap{\beta_{i+1}(D)}\phantom{\delta_{i+1}([x]_{P_{i+1}})} = \begin{cases}
      \delta_{i}(S) & \text{if }D = S \\
      \beta_{i}(B)\wedge \neg \delta_{i}(S) & \text{if }D = B\setminus S \\
      \beta_{i}(D) & \text{if }D \in C/Q_i \\
      \end{cases}$\medskip
    \item\label{defDeltai1} $\delta_{i+1}([x]_{P_{i+1}}) = \begin{cases}
        \delta_{i}([x]_{P_{i}}) &\text{if }[x]_{P_{i+1}} = [x]_{P_i} \\
        \delta_{i}([x]_{P_{i}}) \wedge \fmod{F\chi_S^B(c(x))}(\delta_i(S),\beta_i(B))
        &\text{otherwise.}\\
        \end{cases}$\medskip
  \end{enumerate}
  Upon termination, return $\delta_i$.
\end{algorithm}
Like in \autoref{paigeTarjan}, the only block of $C/Q_0$ has
$\beta_0(\{C\}) = \top$ as a certificate. Since the partition $C/P_0$
distinguishes by the `output' (e.g.~final vs. non-final states), the
certificate of~$[x]_{P_0}$ specifies what $F!(c(x))\in F1$ is
(\autoref{lemF1CoalgSem}).

In the $i$-th iteration of the main loop, we have certificates
$\delta_{i}(S)$ and $\beta_i(B)$ for $S\subsetneqq B$ in
step~\ref{step1} satisfying \eqref{eqCertCorrect} available from the
previous iterations.  In~\ref{defBetai1}, the Boolean connectives
describe how $B$ is split into $S$ and $B\setminus S$. In
\ref{defDeltai1}, new certificates are constructed for every
predecessor block $T\in C/P_i$ that is refined. If $T$ does not
change, then neither does its certificate. Otherwise, the block
$T = [x]_{P_i}$ is split into the blocks $[x]_{F\chi_S^B(c(x))}$ for
$x \in T$ in step \ref{defPi1}, which is reflected by the $F3$ modality
$\fmod{F\chi_S^B(c(x))}$ as per \autoref{lemF3CoalgSem}.

\begin{rem} \label{cancelConjunct} In step~\ref{defBetai1},
  $\beta_{i+1}(D)$ can be simplified to be no larger than
  $\delta_i(S)$. To see this, note that for $S\subseteq B\subseteq C$,
  $S\in X/P_i$, and $B\in X/Q_i$, every conjunct of $\beta_i(B)$ is
  also a conjunct of $\delta_i(S)$.  In
  $\beta_i(B)\wedge \neg \delta_i(S)$, one can hence remove all
  conjuncts of $\beta_i(B)$ from $\delta_i(S)$, obtaining a formula
  $\delta'$, and then equivalently use $\beta_i(B)\wedge \neg \delta'$
  in the definition of $\beta_{i+1}(D)$.
\end{rem}

\begin{theorem} \label{algoCertsCorrect} For zippable~$F$,
  \autoref{algoCerts} is correct, i.e.~\eqref{eqCertCorrect} holds for
  all~$i$. Thus, upon termination $\delta_i$ assigns certificates to
  each block of $C/\mathord{\sim} = C/P_i$.
\end{theorem}
\begin{corollary}[Hennessy-Milner]
  \label{hennessyMilner}
  For zippable $F$, states $x,y$ in a finite $F$-coalgebra are
  behaviourally equivalent iff they agree on all $F3$-modal formulae.
\end{corollary}

\begin{rem} \label{extractDistinguish} A smaller formula
  distinguishing a state $x$ from a state $y$ can be extracted from
  the certificates in time $\CO(|C|)$. It is the leftmost conjunct
  that is different in the respective certificates of $x$ and
  $y$. This is the subformula starting at the modal operator
  introduced in $\delta_i$ for the least $i$ with $(x,y)\notin P_i$;
  hence, $x$ satisfies $\fmod{t}(\delta,\beta)$ but $y$ satisfies
  $\fmod{t'}(\delta,\beta)$ for some $t'\neq t$ in $F3$.
\end{rem}

\subsection{Complexity Analysis}
\label{complexityAnalysis}
The operations introduced by \autoref{algoCerts} can be implemented
with only constant run time overhead. To this end, one implements
$\beta$ and $\delta$ as arrays of formulae of length $|C|$ (note that
at any point, there are at most $|C|$-many blocks). In the
refinable-partition data structure~\cite{ValmariLehtinen08}, every
block has an index (a natural number) and there is an array of
length~$|C|$ mapping every state $x\in C$ to the block it is contained
in. Hence, for both partitions $C/P$ and $C/Q$, one can look up a
state's block and a block's certificate in constant time.

It is very likely that the certificates contain a particular
subformula multiple times and that certificates of different blocks
share common subformulae. For example, every certificate of a block
refined in the $i$-th iteration using $S\subsetneqq B$ contains the
subformulas $\delta_i(S)$ and~$\beta_i(B)$.  Therefore, it is
advantageous to represent all certificates constructed as one directed
acyclic graph (dag) with nodes labelled by modal operators and
conjunction and having precisely two outgoing edges. Moreover, edges
have a binary flag indicating whether they represent negation $\neg$.
Initially, there is only one node representing~$\top$, and the
operations of \autoref{algoCerts} allocate new nodes and update the
arrays for~$\beta$ and~$\delta$ to point to the right nodes.  For
example, if the predecessor block $T\in C/P_i$ is refined in step
\ref{defDeltai1}, yielding a new block $[x]_{P_{i+1}}$, then a new
node labelled $\wedge$ is allocated with edges to the nodes
$\delta_i(T)$ and to another new node labelled
$F\chi_S^B(c(x))$ with edges to the nodes $\delta_i(S)$ and
$\delta_i(B)$.

For purposes of estimating the size of formulae generated by the
algorithm, we use a notion of \emph{transition} in coalgebras,
inspired by the notion of canonical graph~\cite{Gumm2005}.

\begin{definition}\label{coalgebraEdge}
  For states $x,y$ in an $F$-coalgebra $(C,c)$, we say that there is a
  \emph{transition $x\to y$} if $c(x)\in FC$ is not in the image
  $Fi[F(C\setminus \{y\})]~(\subseteq FC)$, where
  $i\colon C\setminus \{y\}\monoto C$ is the inclusion map.
\end{definition}
\begin{theorem}\label{certSize} For a coalgebra with~$n$ states
  and~$m$ transitions, the formula dag constructed by
  \autoref{algoCerts} has size $\CO(m\cdot \log n + n)$ and height at
  most~${n+1}$.
\end{theorem}
\begin{theorem} \label{runTimePreserved} 
  \autoref{algoCerts} adds only constant run time overhead, thus it has the same
  run time as \autoref{coalgPT} (regardless of the optimization from
\autoref{cancelConjunct}).
\end{theorem}
\noindent For a tighter run time analysis of the underlying partition refinement
algorithm, one additionally requires that~$F$ is equipped with a \emph{refinement
  interface}~\cite[Def.~6.4]{concurSpecialIssue}, which is based on a
given encoding of $F$-coalgebras in terms of \emph{edges} between
states (encodings serve only as data structures and have no direct
semantic meaning, in particular do not entail a semantic reduction to
relational structures). This notion of edge yields the same numbers (in
$\CO$-notation) as \autoref{coalgebraEdge} for
all functors considered.
All zippable functors we consider here have
refinement interfaces~\cite{concurSpecialIssue,coparFM19}. In presence
of a refinement interface, step \ref{defPi1} can be implemented
efficiently, with resulting overall run time
$\CO((m+n) \cdot \log n\cdot p(c))$ where $n=|C|$, $m$ is the number
of edges in the encoding of the input coalgebra $(C,c)$, and the
\emph{run-time factor} $p(c)$ is associated with the refinement interface. In
most instances, e.g.~for~$\Pow$, $\R^{(-)}$, one has $p(c) = 1$; in particular,
the generic algorithm recovers the run time of the Paige-Tarjan
algorithm.

\begin{rem}
  The claimed run time relies on close attention to a number of
  implementation details. This includes use of an efficient data
  structure for the
  partition~$C/P_i$~\cite{Knuutila2001,ValmariLehtinen08}; the other
  partition $C/Q_i$ is only represented implicitly in terms of a queue
  of blocks $S\subsetneqq B$ witnessing $P_i\subsetneqq Q_i$,
  requiring additional care when splitting blocks in the queue
  \cite[Fig.~3]{ValmariF10}. Moreover, grouping the elements of a
  block by $F3$ involves the consideration of a \emph{possible
    majority candidate}~\cite{ValmariF10}.
\end{rem}

\begin{theorem} \label{thmLogRunTime} On a coalgebra with $n$ states
  and $m$ transitions for a zippable set functor with a refinement
  interface with factor~$p(c)$, \autoref{algoCerts} runs in time
  $\CO((m+n)\cdot \log n\cdot p(c))$.
\end{theorem}

\section{Cancellative Functors}
\label{sec:cancellative}
Our use of binary modalities relates to the fact that, as observed
already by Paige and Tarjan, when splitting a block according to an
existing partition of a block~$B$ into $S\subseteq B$ and~$B\setminus S$,
it is not in general sufficient to look only at the successors
in~$S$. However, this does suffice for some transition types;
e.g.~Hopcroft's algorithm for deterministic automata~\cite{Hopcroft71}
and Valmari and Franceschinis' algorithm for weighted systems
(e.g.~Markov chains)~\cite{ValmariF10} both split only with respect
to~$S$. In the following, we exhibit a criterion on the level of
functors that captures that splitting w.r.t.~only $S$ is sufficient:
\begin{samepage}
\begin{defn} \label{funCancellative}
  A functor $F$ is \emph{cancellative} if the map
  \[
    \fpair{F\chi_{\set{1,2}},F\chi_{\set{2}}}\colon F3\to F2\times F2
  \]
  is injective.
\end{defn}
\end{samepage}
\noindent To understand the role of the above map, recall the function
$\chi_S^B\colon C\to 3$ from~\eqref{eqChi3} and note that
$\chi_{\set{1,2}}\cdot\chi_S^B=\chi_B$ and
$\chi_{\set{2}}\cdot\chi_S^B=\chi_S$, so the composite
$\fpair{F\chi_{\set{1,2}},F\chi_{\set{2}}}\cdot F\chi_S^B$ yields
information about the accumulated transition weights into~$B$ and~$S$
but not about the one into $B\setminus S$; the injectivity condition
means that for cancellative functors, this information suffices in the
splitting step for $S\subseteq B\subseteq C$.  The term
\emph{cancellative} stems from the respective property on monoids;
recall that a monoid~$M$ is \emph{cancellative} if $s + b_1 = s+b_2$
implies $b_1 = b_2$ for all $s,b_1,b_2\in M$.
\begin{proposition}
  \label{monoidValCancellative}
  The monoid-valued functor $M^{(-)}$ for a commutative monoid $M$ is
  cancellative if and only if~$M$ is a cancellative monoid.
\end{proposition}
\noindent
Hence, $\R^{(-)}$ is cancellative, but $\Powf$ is not.
Moreover, all signature functors are cancellative:
\begin{proposition} \label{cancellativeClosure}
  The class of cancellative functors
  contains the all constant functors as well as the identity functor,
  and it is
  closed under subfunctors, products, and coproducts. 
  \twnote{}
\end{proposition}
\noindent
For example, $\Dist$ is cancellative, but $\Pow$ is not because of its subfunctor
$\Powf$.
\begin{rem} \label{cancellativeVsZippable}
  Cancellative functors are neither closed under quotients nor under composition.
  Zippability and cancellativity are independent properties.
  Zippability in conjunction with cancellativity implies $m$-zippability for all
  $m\in \N$, the
  $m$-ary variant~\cite{KoenigEA20} of zippability.
\end{rem}

\begin{theorem} \label{patchyTheorem}
  If $F$ is a cancellative functor, 
  $\fmod{F\chi_S^B(c(x))}(\delta_i(S),\beta_i(B))$
  in \autoref{algoCerts}
  can be replaced with
  \(
    \fmod{F\chi_S^C(c(x))}(\delta_i(S),\top)
  \). Then, the algorithm still correctly computes certificates in the given
  $F$-coalgebra $(C,c)$.
\end{theorem}
Note that in this optimized algorithm, the computation of $\beta$ can
be omitted because it is not used anymore. Hence, the resulting
formulae only involve $\wedge$, $\top$, and modalities from the
set~$F3$ (with the second parameter fixed to $\top$). These modalities
are equivalently unary modalities induced by elements of $F2$, which
we term \emph{$F2$-modalities}; hence, the corresponding
Hennessy-Milner Theorem (\autoref{hennessyMilner}) adapts to $F2$ for
cancellative functors, as follows:
\begin{corollary}
  For zippable and cancellative $F$, states in an $F$-coalgebra are
  behaviourally equivalent iff they agree on modal formulae built
  using $\top$, $\wedge$, and unary $F2$-modalities.
\end{corollary}

\section{Domain-Specific Certificates}
\label{domainSpecific}
On a given specific system type, one is typically interested in
certificates and distinguishing formulae expressed via modalities
whose use is established in the respective domain, e.g.~$\Box$
and~$\Diamond$ for transition systems. We next describe how the
generic $F3$ modalities can be rewritten to domain-specific ones in a
postprocessing step. The domain-specific modalities will not in
general be equivalent to $F3$-modalities, but still yield
certificates.  \twnote{}

\begin{defn}\label{remBoolCombination}
  The \emph{Boolean closure} $\bar\Lambda$ of a modal signature
  $\Lambda$ has as $n$-ary modalities
  propositional combinations of atoms of the form
  $\hearts(i_1,\dots,i_k)$, for $\arity{\hearts}{k}\in\Lambda$, where
  $i_1,\dots,i_k$ are propositional combinations of elements of
  $\{1,\ldots,n\}$. Such a modality~$\arity{\lambda}{n}$ is
  interpreted by predicate liftings
  $\semantics{\lambda}_X\colon (2^X)^n \to FX$ defined inductively in
  the obvious way.
\end{defn}
For example, the boolean closure of $\Lambda = \{\arity{\Diamond}{1}\}$ contains
the unary modality $\Box = \neg\Diamond\neg$.

\begin{defn}\label{domainCert}
  Given a modal signature $\Lambda$ for a functor $F$, a
  \emph{domain-specific interpretation} consists of functions
  $\tau\colon F1 \to \bar\Lambda$ and
  $\lambda\colon F3 \to \bar\Lambda$ assigning to each $o \in F1$ a
  nullary modality~$\tau_o$ and to each $t\in F3$ a binary
  modality~$\lambda_t$ such that the predicate liftings
  $\semantics{\tau_o}_X\in 2^{FX}$ and
  $\semantics{\lambda_t}_X\colon (2^X)^2 \to 2^{FX}$
  satisfy\smnote{}
  \[
    \semantics{\tau_o}_1 = \{o\}
    \quad\text{(in $2^{F1}$)}
    \quad\text{ and }\quad
    [t]_{F\chi_{\{1,2\}}} \cap \semantics{\lambda_t}_3(\{2\},\{1\})
    =
    \{t\}
    \quad\text{(in $2^{F3}$)}.
  \]
  (Recall that $\chi_{\{1,2\}}\colon 3\to 2$ is the characteristic
  function of $\{1,2\}\subseteq 3$, and
  $[t]_{F\chi_{\{1,2\}}}\subseteq F3$ denotes the equivalence class of
  $t$ w.r.t.~$F\chi_{\{1,2\}}\colon F3 \to F2$.)
\end{defn}
\noindent Thus,~$\tau_o$ holds precisely at states with output
behaviour $o\in F1$.  Intuitively, $\lambda_t(\delta,\rho)$ describes
the refinement step of a predecessor block $T$ when splitting
$B:=\semantics{\delta}\cup\semantics{\rho}$ into
$S:=\semantics{\delta}$ and $B\setminus S:=\semantics{\rho}$
(\autoref{fig:parttree}), which translates into the arguments $\{2\}$
and $\{1\}$ of $\semantics{\lambda_t}_3$. In the refinement step, we
know from previous iterations that all elements have the same
behaviour w.r.t.~$B$. This is reflected in the intersection with
$[t]_{F\chi_{\set{1,2}}}$. The axiom guarantees that $\lambda_t$
characterizes $t\in F3$ uniquely, but only within the equivalence
class representing a predecessor block.  Thus, $\lambda_t$ can be much
smaller than equivalents of $\fmod{t}$
(cf.~\autoref{examplePowF3Mod}):

\begin{expl}\label{exDomainSpecInt}
  \begin{enumerate}
  \item\label{F3Pow} For $F=\Pow$, we have a domain-specific interpretation over
    the modal signature $\Lambda = \set{ \Diamond/1}$.
    For $\emptyset, \{0\}\in\Pow 1$, take
    $\tau_{\set 0} = \Diamond\top$ and
    $\tau_\emptyset = \neg\Diamond\top$. For $t \in \Pow 3$, we put
  \[
    \begin{array}{rl@{~~~}l@{\qquad}rl@{~~~}l}
    \lambda_t(\delta,\rho) &=
        \neg\Diamond \rho &\text{if }2\in t\notni 1
    &
    \lambda_t(\delta,\rho) &=
        \Diamond \delta \wedge \Diamond\rho &\text{if }2\in t \ni 1
                          \\
    \lambda_t(\delta,\rho) &=
        \neg\Diamond \delta &\text{if }2\notin t\ni 1
                          &
    \lambda_t(\delta,\rho) &=
        \top &\text{if }2\not\in t \notni 1.
    \end{array}
  \]
  The certificates obtained via this translation are precisely the
  ones generated in the example using the Paige-Tarjan
  algorithm, cf.~\eqref{edgeCasesFormula}, with $\rho$ in lieu of
  $\beta\wedge\neg\delta$. %
\item\label{dsiSignature} For a signature (functor) $\Sigma$, take
  $\Lambda=\{\arity{\sigma}{0}\mid \arity{\sigma}{n}\in \Sigma\}\cup
  \{\arity{\gradI}{1}\mid I\in\Powf (\N)\}$. We
  interpret~$\Lambda$ by the predicate liftings\smnote{}
    \begin{align*}
      \semantics{\sigma}_X & = \{\sigma(x_1, \ldots, x_n) \mid x_1,
      \ldots, x_n \in X\} \subseteq \Sigma X, \\
      \semantics{\gradI}(S) &= \{\sigma(x_1,\ldots,x_n) \in \Sigma X \mid
                              \forall i \in \N\colon i\in I \leftrightarrow (1\le i\le n~\wedge~ x_i \in S)  \}.
    \end{align*}
    Intuitively, $\gradI\,\phi$ states that the $i$th successor satisfies $\phi$
    iff $i\in I$.
    We then have a domain-specific interpretation
    $(\tau,\lambda)$ given by $\tau_o := \sigma$ for
    $o = \sigma(0,\ldots,0) \in \Sigma 1$ and
    $\lambda_t(\delta,\rho) := \grad{I}\delta$ for $t = \sigma(x_1,\ldots,x_n) \in \Sigma 3$ and $I =
    \{i\in \{1,\ldots,n\}\mid x_i = 2\}$.

  \item\label{dsiMonoid} For a monoid-valued functor~$M^{(-)}$, take
    $\Lambda = \{ \arity{\monmod{m}}{1} \mid m \in M \}$, interpreted
    by the predicate liftings
    $\semantics{{\monmod{m}}}_X\colon 2^X\to 2^{M^{(X)}}$ given by
    \[
      \semantics{{\monmod{m}}}_X(S)= \{ \mu \in M^{(X)} \mid
      m =
      \textstyle\sum_{x \in S}\mu(x)\}.
    \]
    A formula $\monmod{m}\,\delta$ thus states that the accumulated weight
    of the successors satisfying~$\delta$ is exactly~$m$. A
    domain-specific interpretation $(\tau,\lambda)$ is then given by
    $\tau_o = \monmod{o(0)}\,\top$ for $o \in M^{(1)}$ and
    \( \lambda_t(\delta,\rho) = \grad{t(2)}\,\delta \wedge
    \grad{t(1)}\,\rho \) for $t\in M^{(3)}$. In case~$M$ is
    cancellative, we can also simply put $\lambda_t(\delta,\rho) = \grad{t(2)}\,\delta$.

  \item\label{dsiMarkov} For labelled Markov chains,
    i.e.~$FX=(\Dist X +1)^A$, let
    $\Lambda = \{\arity{\fpair{a}_p}{1}\mid a\in A, p\in [0,1]\}$,
    where $\fpair{a}_p\phi$ denotes that on input $a$, the next state
    will satisfy~$\phi$ with probability at least~$p$, as in cited work by
    Desharnais \etal~\cite{desharnaisEA02}. This gives rise to the
    interpretation:
    \[
      \tau_o =
      \bigwedge_{\substack{a\in A\\ o(a) \in \Dist 1}} \fpair{a}_{1}\top
      \wedge
      \bigwedge_{\substack{a\in A\\ o(a) \in 1}} \neg \fpair{a}_{1}\top
      \qquad
      \lambda_t(\delta,\rho)
      = \bigwedge_{\substack{a\in A\\t(a)\in \Dist 3}}(\fpair{a}_{t(a)(2)}\,\delta\wedge \fpair{a}_{t(a)(1)}\,\rho)
    \]
  \end{enumerate}
\end{expl}

\noindent Given a domain-specific interpretation $(\tau,\lambda)$ for
a modal signature~$\Lambda$ for the functor~$F$, we can postprocess
certificates~$\phi$ produced by \autoref{algoCerts} by replacing the
modalities $\fmod t$ for $t \in F3$ according to the translation~$T$
recursively defined by the following clauses for modalities and by
commutation with propositional operators:
\twnote{}
\[
  T\big(\fmod{t}(\top,\top)\big) = \tau_{F!(t)}
  \qquad
  T\big(\fmod{t}(\delta,\beta))
  =
  \lambda_t\big(T(\delta), T(\beta) \wedge \neg T(\delta)\big).
\]
Note that one can replace $T(\beta)\wedge\neg T(\delta)$ with $T(\beta)\wedge
\neg T(\delta')$ for the optimized $\delta'$ from
\autoref{cancelConjunct}; the latter conjunction has essentially the same size as $T(\delta)$.

\begin{proposition}\label{domainCertMainThm}
  \smnote{}
  For every certificate $\phi$ of a behavioural equivalence class of
  a given co\-al\-ge\-bra produced by either \autoref{algoCerts}
  or its optimization (\autoref{patchyTheorem}),~$T(\phi)$ is also a certificate
  for that class.%
\end{proposition}
Thus, the domain-specific modal signatures also inherit a Hennessy-Milner Theorem.

\begin{example}\label{exCertMarkov}
  For labelled Markov chains ($FX=(\Dist X +1)^A$) and the
  interpretation via the modalities $\fpair{a}_p$
  (\itemref{exDomainSpecInt}{dsiMarkov}), this yields certificates
  (thus in particular distinguishing formulae) in run time
  $\CO(|A|\cdot m\cdot \log n)$, with the same bound on formula
  size. Desharnais \etal~describe an
  algorithm~\cite[Fig.~4]{desharnaisEA02} that computes distinguishing
  formulae in the negation-free fragment of the same logic (they note
  also that this fragment does not suffice for certificates). They do
  not provide a run-time analysis, but the nested loop structure
  indicates that the asymptotic complexity is roughly $|A|\cdot n^4$.
\end{example}

\twnote[inline]{}

\section{Worst Case Tree Size of Certificates}
\label{worstcase}
In the complexity analysis (\autoref{complexityAnalysis}), we have
seen that certificates -- and thus also distinguishing formulae --
have dag size $\CO(m\cdot \log n + n)$ on input coalgebras with~$n$
states and~$m$ transitions. However, when formulae are written in the
usual linear way, multiple occurrences of the same subformula lead to
an exponential blow up of the formula size in this sense, which for
emphasis we refer to as the \emph{tree size}.

\takeout{}%

Figueira and Gor{\'{\i}}n~\cite{FigueiraG10} show that exponential
tree size is inevitable even for distinguishing formulae. The proof is
based on winning strategies in bisimulation games, a technique that is
also applied in other results on lower bounds on formula
size~\cite{FrenchEA13,AdlerImmermanLics01,AdlerImmermanTocl03}.

\begin{oprob}%
  Do states in $\R^{(-)}$-coalgebras generally have certificates of
  subexponential tree size in the number of states? If yes, can small
  certificates be computed efficiently?
\end{oprob}
We note that for another cancellative functor, the answer is
well-known: On deterministic automata, i.e.~coalgebras for
$FX = 2\times X^A$, the standard minimization algorithm constructs
distinguishing words of linear length.

\begin{rem}
\label{cleavelandMinimal}
Cleaveland~\cite[p.~368]{Cleaveland91} also mentions that minimal
distinguishing formulae may be exponential in size, however for a
slightly different notion of minimality: a formula~$\phi$
distinguishing $x$ from $y$ is \emph{minimal} if no $\phi$
obtained by replacing a non-trivial subformula of $\phi$ with the
formula $\top$ distinguishes $x$ from $y$.  This is weaker than
demanding that the formula size of $\phi$ is as small as
possible. For example, in the transition system \smnote{}
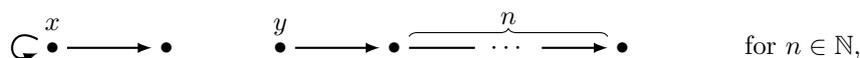
\begin{center}
  \hspace{8mm}
  \begin{tikzpicture}[coalgebra,x=1.5cm,baseline=(x.base)]
    \begin{scope}[every node/.append style={
        state,
        label distance=-.5mm,
        outer sep=3pt,
        inner sep=0pt,
        text depth=0pt,
        font=\normalsize,%
      }
      ]
      \node[label=above:$x$] (x)  at (0, 0) {$\bullet$};
      \node (x1)  at (1, 0) {$\bullet$};
      \begin{scope}[xshift=3cm]
        \node[label=above:$y$]
        (y 0) at (0,0) {$\bullet$};
        \foreach \n in {1,3} {
          \node (y \n) at (\n,0) {$\bullet$};
        }
        \node (y 2) at (2,0) {$\cdots$};
      \end{scope}
    \end{scope}
    \path[path with edges]
      (x) edge (x1)
      (x) edge[out=150,in=200,looseness=6.5,overlay] (x)
      (y 0) edge (y 1)
      (y 1) edge[shorten >=1mm,-] (y 2)
      (y 2) edge[shorten <=1mm] (y 3)
      ;
    \draw [decorate,decoration={brace,amplitude=3pt,raise=-1pt},yshift=0pt]
    (y 1.north east) -- node[font=\normalsize,yshift=6pt] {$n$}
    (y 3.north west);
  \end{tikzpicture}
  \qquad\qquad for $n\in \N$,
\end{center}
the formula $\phi = \Diamond^{n+2}\top$ distinguishes $x$ from $y$
and is minimal in the above sense. However,~$x$ can in fact be
distinguished from~$y$ in size $\CO(1)$,
by the formula~$\Diamond \neg\Diamond\top$.
\end{rem}

\section{Conclusions and Further Work}

We have presented a generic algorithm that computes distinguishing
formulae for behaviourally inequivalent states in state-based systems
of various types, cast as coalgebras for a functor capturing the
system type. Our algorithm is based on coalgebraic partition
refinement~\cite{concurSpecialIssue}, and like that algorithm runs in
time $\CO((m+n)\cdot \log n \cdot p(c))$, with a functor-specific
factor $p(c)$ that is $1$ in many cases of interest. Independently of
this factor, the distinguishing formulae constructed by the algorithm
have dag size $\CO(m\cdot \log n + n)$; they live in a dedicated
instance of coalgebraic modal logic~\cite{Pattinson04,Schroder08},
with binary modalities extracted from the type functor in a systematic
way.
We have shown that for \emph{cancellative} functors, the construction
of formulae and, more importantly, the logic can be simplified,
requiring only unary modalities and conjunction. We have also
discussed how distinguishing formulae can be translated into a more
familiar domain-specific syntax (e.g. Hennessy-Milner logic for
transition systems).

There is an open source implementation of the underlying partition refinement
algorithm~\cite{coparFM19}, which may serve as a basis for a future implementation.

In partition refinement, blocks are successively refined in a top-down
manner, and this is reflected by the use of conjunction in
distinguishing formulae. Alternatively, bisimilarity may be computed
bottom-up, as in a recent partition \emph{aggregation}
algorithm~\cite{BjorklundCleophas2020}. It is an interesting point for
future investigation whether this algorithm can be extended to compute
distinguishing formulae, which would likely be of a rather different
shape than those computed via partition refinement.

\label{maintextend} %

\smnote[inline]{}

\bibliographystyle{plainurl} %
\bibliography{refs}

\clearpage
\appendix
\section{Appendix: Omitted Proofs}
      \allowdisplaybreaks
\def\thesection{A} %
\appendixsect{preliminaries}{Preliminaries}
\begin{proofappendix}[Details for]{behEqDifferntCoalg}
  Given a pair of $F$-coalgebra $(C,c)$ and $(D,d)$, we
  have a canonical $F$-coalgebra structure on the
  the disjoint union $C+D$ of their carriers:
  \[
    C+D \xrightarrow{c+d} FC + FD \xrightarrow{[F\inl, F\inr]} F(C+D).
  \]
  The canonical inclusion maps $\inl\colon C\to C+D$ and $\inr\colon D\to C+D$
  are $F$-coalgebra morphisms. We say that states $x \in C$ and $y \in D$ are \emph{behavioural
    equivalent} if $\inl(x)\sim \inr(y)$.

  Note that this definition extends the original definition of $\sim$,
  in the sense that $x,y$ in the same coalgebra $(C,c)$ are
  behaviourally equivalent ($x\sim y$) iff $\inl(x) \sim \inr(y)$ in
  the canonical coalgebra on $C+C$.
\end{proofappendix}

\begin{proofappendix}[Details on Predicate Liftings in]{sec:coalgLogic}
  The naturality of $\semantics{\heartsuit}_X\colon (2^X)^n\to 2^{FX}$ in $X$
  for $\arity{\heartsuit}{n}$
  means that for every map $f\colon X\to Y$, the diagram
  \[
    \begin{tikzcd}
      (2^Y)^n
      \arrow{r}{\semantics{\heartsuit}_Y}
      \arrow{d}[swap]{(2^{f})^n}
      & 2^{FY}
      \arrow{d}{2^{Ff}}
      \\
      (2^X)^n
      \arrow{r}{\semantics{\heartsuit}_Y}
      & 2^{FX}
    \end{tikzcd}
  \]
  commutes. Since $2^{(-)}$ is \emph{contravariant}, the map $f\colon X\to Y$ is
  sent to $2^{f}\colon 2^Y\to 2^X$ which takes inverse images; writing down the
  commutativity element-wise yields \eqref{eq:naturality}. By the Yoneda
  lemma, one can define predicate liftings
  \begin{lemma} \label{predLiftYoneda}
    A predicate lifting $\semantics{\heartsuit}_X\colon (2^X)^n\to 2^{FX}$ for
    $\arity{\heartsuit}{n}$ is uniquely defined by a map $f\colon F(2^n)\to 2$.
    Then $\semantics{\heartsuit}_X$ is given by
    \[
      \semantics{\heartsuit}_X(P_1,\ldots,P_n)(\underbrace{t}_{\in\,FX}) =
      f(F(\underbrace{x\mapsto (P_1(x),\ldots,P_n(x))}_{X\to 2^n})(t))
    \]
    or written as sets (considering $f\subseteq F(2^n), P_i\subseteq X$):
    \[
      \semantics{\heartsuit}_X(P_1,\ldots,P_n) =
      \{ t\in FX \mid F\fpair{\chi_{P_1},\ldots,\chi_{P_n}}(t) \in f
      \}
    \]
  \end{lemma}
  \begin{proof}
    The following mathematical objects are in one-to-one correspondence
    \[
      \dfrac{
          F(2^n)\to 2
      }{
        \dfrac{
          (2^n)^X\to 2^{FX}\text{ natural in $X$}
        }{
          (2^X)^n\to 2^{FX}\text{ natural in $X$}
        }
      }
      \qquad
      \qquad
      \dfrac{
        f
      }{
        \dfrac{
          p\mapsto t\mapsto f(Fp(t))
        }{
          (P_1,\ldots,P_n)\mapsto t\mapsto f(F\fpair{P_1,\ldots,P_n}(t))
        }
      }
    \]
    The first correspondence is the
    Yoneda lemma and the second correspondence is a power law.
    On the right, the inhabitants of the sets are listed when starting with
    $f\colon F(2^n)\to 2$. By the definition of $\fpair{-,-}$ we have:
    \[
      \fpair{P_1,\ldots,P_n}\colon X\to 2^n
      \qquad
      x\mapsto (P_1(x),\ldots,P_n(x))
      \tag*{\qedhere}
    \]
  \end{proof}
\end{proofappendix}

\appendixsect{sec:main}{Constructing Distinguishing Formulae}
\begin{proofappendix}[Verification of]{defF3Mod}
  We verify that for every $t\in F3$
  \[
    \semantics{\fmod{t}}_X\colon
    (2^X)^2 \to 2^{FX},
    \qquad
    \semantics{\fmod{t}}_X(S,B) = \{t'\in FX\mid
    F\chi_{S\cap B}^B(t') = t
    \}
  \]
  defines a predicate lifting \eqref{eq:naturality}. For $f\colon X\to Y$ 
  and $S, B\in 2^X$, note that we have
  \[
    \chi_{S\cap B}^B\cdot f = \chi_{f^{-1}[S\cap B]}^{f^{-1}[B]}
    \tag*{$(*)$}
  \]
  because $f(x) \in X'$ iff $x\in f^{-1}[X']$ 
  for all $x\in X$ and $X'\subseteq X$. We verify:
  \begin{align*}
    Ff^{-1}\big[\semantics{\fmod{t}}_Y(S,B)\big]
    &= Ff^{-1}\big[\set{t'\in FY\mid F\chi_{S\cap B}^B(t') = t}\big]
      \tag{\autoref{defF3Mod}}
    \\
    &= \set{t'' \in FX\mid F\chi_{S\cap B}^B(Ff(t'')) = t}\big]
      \tag{def.~inv.~Image}
    \\
    &= \set{t'' \in FX\mid F(\chi_{S\cap B}^B\cdot f)(t'')) = t}\big]
      \tag{Functoriality}
    \\
    &= \set{t'' \in FX\mid F\chi_{f^{-1}[S\cap B]}^{f^{-1}[B]}(t'') = t}\big]
    \tag*{$(*)$}
    \\
    &= \set{t'' \in FX\mid F\chi_{f^{-1}[S]\cap f^{-1}[B]}^{f^{-1}[B]}(t'') = t}\big]
    \\
    &= \semantics{\fmod{t}}_X(f^{-1}[S],f^{-1}[B])
      \tag{\autoref{defF3Mod}}
  \end{align*}
  Hence, $\semantics{\fmod{t}}$ is a predicate lifting.
\end{proofappendix}

\begin{proofappendix}{lemF3CoalgSem}
  This follows directly from \autoref{defF3Mod} for $S:=\semantics{\phi_S}$ and
  $B:=\semantics{\phi_B}$, using that $S\cap B = S$:
  \begin{align*}
    \semantics{\fmod{t}(\phi_S,\phi_B)}
    &=
    c^{-1}[\semantics{\fmod{t}}_C(\semantics{\phi_S},\semantics{\phi_B})]
    \\
    &= c^{-1}[\semantics{\fmod{t}}_C(S,B)]
    \\
    & = c^{-1}[\{t'\in FC \mid
    F\chi_{S\cap B}^B(t') = t\}]
    \\
    &= \{x \in C \mid F\chi_{S}^B(c(x)) = t\}.
    \tag*{\qedhere}
  \end{align*}
\end{proofappendix}

\begin{proofappendix}{lemF1CoalgSem}
  Note that for $\chi_C^C\colon C\to 3$, we have
  $\chi_C^C = (C \xra{!} 1 \xra{j_1} 3)$ where $j_1(0) = 2$.
  \begin{align*}
    \semantics{\fmod{t}}
    &= \semantics{\fmod{Fj_1(t)}(\top,\top)}
    \tag{\autoref{notationF1Mod}}
    \\
    & = \{x \in C \mid F\chi_{C}^{C}(c(x)) = Fj_1(t)\}
    \tag{\autoref{lemF3CoalgSem}, $\semantics{\top}=C$}
    \\
    & = \{x \in C \mid Fj_1(F!(c(x))) = Fj_1(t)\}
    \tag{$\chi_C^C =j_1 \cdot !$}
    \\
    & = \{x \in C \mid F!(c(x)) = t\}
    \tag{$Fj_1$ injective}
  \end{align*}
  In the last step use that, w.l.o.g., $F$ preserves injective maps (\autoref{trnkovahull}).
  \qed
\end{proofappendix}

\begin{proofappendix}{algoCertsCorrect}
  \begin{enumerate}
  \item We first observe that given $x\in C$,
    $S\subseteq B \subseteq C$, and formulae $\phi_S$ and~$\phi_B$
    which characterize $S$ and $B$, respectively, we have:
  \begin{align}
    \semantics{\fmod{F\chi_{S}^B(c(x))}(\phi_S,\phi_B)}
    &=%
    \{x' \in C\mid F\chi_{\semantics{\phi_S}}^{\semantics{\phi_B}}(c(x')) =
    F\chi_{S}^B(c(x)) \}\label{eqF3ModBS:1}
      \\
    &= [x]_{F\chi_S^B(c(x))},
    \label{eqF3ModBS}
  \end{align}
  where~\eqref{eqF3ModBS:1} uses \autoref{lemF3CoalgSem}
  and~\eqref{eqF3ModBS}, holds since $\semantics{\phi_B} = B$ and
  $\semantics{\phi_S} = S$.
\item We proceed to\smnote{} the verification of \eqref{eqCertCorrect} by induction on $i$.
  \begin{itemize}
  \item In the base case $i = 0$, we have $\semantics{\beta_0(\{C\})} =
    \semantics{\top} = \{C\}$ for the only block in $X/Q_0$. Since $P_0 = \ker
    (F!\cdot c)$, $\delta_0$ is well-defined, and by \autoref{lemF1CoalgSem} we
    have 
    \[
      \semantics{\delta_0([x]_{P_0})}
      = \semantics{\fmod{F!(c(x))}}
      = \{y\in C \mid F!(c(x)) = F!(c(y))\}
      = [x]_{P_0}.
    \]
  \item The inductive hypothesis states that 
    \[
      \semantics{\delta_i(S)} = S
      \qquad\text{and}\qquad
      \semantics{\beta_i(B)} = B.\tag*{\text{(IH)}}
    \]
    We prove that $\beta_{i+1}$ is correct:
    \begin{align*}
      & \semantics{\beta_{i+1}([x]_{Q_{i+1}})}
      \\ &=
      \begin{cases}
        \semantics{\delta_{i}(S)} & \text{if }[x]_{Q_{i+1}} = S
        \text{, hence }S = [x]_{P_i} \\
        \semantics{\beta_{i}(B)} ~\cap ~C\setminus \semantics{\delta_{i}(S)} & \text{if }[x]_{Q_{i+1}} = B\setminus S
        \text{, hence }B = [x]_{Q_i}
 \\
        \semantics{\beta_{i}([x]_{Q_i})} & \text{if }[x]_{Q_{i+1}} \in C/Q_i \\
      \end{cases}
      \\
    &\overset{\mathclap{\text{(IH)}}}{=}
      \begin{cases}
        S & \text{if }[x]_{Q_{i+1}} = S\\
        B ~\cap ~C\setminus S & \text{if }[x]_{Q_{i+1}} = B\setminus S \\
        [x]_{Q_i} & \text{if }[x]_{Q_{i+1}} \in C/Q_i \\
      \end{cases}
      \\ &=
      \begin{cases}
        [x]_{Q_{i+1}} & \text{if }[x]_{Q_{i+1}} = S = [x]_{P_i}\\
        [x]_{Q_{i+1}}& \text{if }[x]_{Q_{i+1}} = B\setminus S
        \qquad\text{(since $B \cap C\setminus S = B \setminus S$)}   \\
        [x]_{Q_{i+1}} & \text{if }[x]_{Q_{i+1}} \in C/Q_i
        \qquad\text{(since $[x]_{Q_i}$ is not split)}\\
      \end{cases}
      \\ &=
      [x]_{Q_{i+1}}.
    \end{align*}
    For $\delta_{i+1}$, we compute as follows:
    \begin{align*}
      & \semantics{\delta_{i+1}([x]_{P_{i+1}})}
      \\
      &=
                                              \qquad
          \begin{cases}
          \semantics{\delta_{i}([x]_{P_{i}})} &\text{if }[x]_{P_{i+1}} = [x]_{P_i} \\
          \semantics{\delta_{i}([x]_{P_{i}})}
          \cap \semantics{\fmod{F\chi_{S}^B(c(x))}(\delta_i(S), \beta_i(B))}
          &\text{otherwise}
          \end{cases}
            \\
        &\overset{\mathclap{\text{(IH) \& \eqref{eqF3ModBS}}}}{=}
          \qquad
          \begin{cases}
          [x]_{P_{i}} &\text{if }[x]_{P_{i+1}} = [x]_{P_i} \\
          [x]_{P_{i}}
          \cap [x]_{F\chi_S^B(c(x))}
          &\text{otherwise}
          \end{cases}
            \\
        &\overset{\mathclap{\text{def. }P_{i+1}}}{=}
          \qquad
          \begin{cases}
          [x]_{P_{i+1}} &\text{if }[x]_{P_{i+1}} = [x]_{P_i} \\
          [x]_{P_{i+1}}
          &\text{otherwise}
          \end{cases}
            \\ &= \quad[x]_{P_{i+1}}
                 \tag*{\qedhere}
    \end{align*}
  \end{itemize}
\end{enumerate}
\end{proofappendix}

\begin{proofappendix}[Details for]{extractDistinguish}
  In order to verify that the first differing conjunct is a distinguishing
  formula, we perform a case distinction on the least $i$ with $(x,y)\notin P_i$:

  If $x$ and $y$ are already split by $P_0$, then the conjunct at index $0$ in
  the respective certificates of $[x]_{\sim}$ and $[y]_{\sim}$ differs, and we have
  $t = F!(c(x))$ and $t' = F!(c(y))$. By \autoref{lemF1CoalgSem}, $\fmod{t}$
  distinguishes $x$ from $y$ (and $\fmod{t'}$ distinguishes~$y$ from
  $x$)\smnote{}.

  If $x$ and $y$ are split by $P_{i+1}$ (but $(x,y)\in P_i$) in the
  $i$th iteration, then
  \[
    \underbrace{F\chi_S^B(c(x))}_{t~:=} \neq \underbrace{F\chi_S^B(c(y))}_{t'~:=}.
  \]
  Thus, the conjuncts that differs in the respective certificates for
  $[x]_{\sim}$ and $[y]_{\sim}$ are the following conjuncts at index
  $i+1$:
  \[
    \fmod{t}(\delta_i(S),\beta_i(B))
    \qquad
    \text{and}
    \qquad
    \fmod{t'}(\delta_i(S),\beta_i(B)).
  \]
  By \autoref{lemF3CoalgSem}, $\fmod{t}(\delta_i(S),\beta_i(B))$
  distinguishes $x$ from $y$ (and $\fmod{t}(\delta_i(S),\beta_i(B))$
  distinguishes $y$ from $x$).
\end{proofappendix}
\takeout{}%

\begin{proofappendix}{certSize}
  Before proving \autoref{certSize}, we need to establish a sequence
  of lemmas on the underlying partition refinement algorithm. We
  assume wlog that $F$ preserves finite intersections; that is
  pullbacks of pairs of injective maps. In fact, the functor $G$
  mentioned in \autoref{trnkovahull}, which coincides with $F$ on all
  nonempty sets and map and therefore has the same coalgebras, preserves
  finite intersections.\lsnote{}.

  Let $(C,c)$ be a coalgebra for $F$. As additional notation, we
  define for all sets $T\subseteq C$ and $S\subseteq C$:
  \[
    T\to S \qquad:\Longleftrightarrow\qquad
    \exists x\in T, y\in  S\colon x\to y.
  \]
  In other words, we write $T\to S$ if there is a transition from (some state
  of) $T$ to (some state of) $S$. Also we define
  the set of predecessor states of a set as:
  \[
    \pred(S) = \{x\in C\mid \{x\}\to S\}
    \qquad\text{for }S\subseteq C.
  \]
  \begin{lemma} \label{noEdgeNoChiS}
    For every $F$-coalgebra $(C,c)$, $x\in C$, and $S\subseteq B\subseteq C$
    with $S$ finite, we have
    \[
      \{x\}\not\to S
      \qquad\Longrightarrow\qquad
      F\chi_S^B(c(x)) = F\chi_{\emptyset}^B(c(x)).
    \]
  \end{lemma}
  \begin{proof}
    For every $y\in S$, we have that $x\not\to y$. Hence, for every $y\in S$,
    there exists $t_y\in F(C\setminus\{y\})$ such that
    \[
      c(x) = Fi(t_y)
      \qquad\text{for }i\colon C\setminus\{y\} \monoto C.
    \]
    The set $C\setminus S$ is the intersection of all sets
    $C\setminus\{y\}$ with $y \in S$:
    \[
      C\setminus S = \bigcap_{y\in S} (C\setminus\{y\}).
    \]
    Since $F$ preserves finite intersections and $S$ is
    finite, we have that
    \[
      F(C\setminus S) = \bigcap_{y\in S} F(C\setminus\{y\}).
    \]
    Since $c(x)\in FC$ is contained in every $F(C\setminus\{y\})$ (as witnessed by $t_y$)
    it is also contained in their intersection. That is, for $m\colon C\setminus
    S\monoto C$ being the inclusion map, there is $t'\in F(C\setminus S)$ with
    $Fm(t') = c(x)$. Now consider the following diagrams:
    \[
      \begin{tikzcd}[column sep=2mm]
        c(x)
        \descto{r}{\(\in\)}
        & FC
        \arrow{r}{F\chi_S^B}
        &[8mm] F3
        \\
        t'
        \arrow[mapsto]{u}
        \descto{r}{\(\in\)}
        & F(C\setminus S)
        \arrow{u}{Fm}
        \arrow{ur}[swap]{F\chi_\emptyset^B}
      \end{tikzcd}
      \qquad\text{and}\qquad
      \begin{tikzcd}[column sep=2mm]
        FC
        \arrow{r}{F\chi_\emptyset^B}
        &[8mm] F3
        \\
        F(C\setminus S)
        \arrow{u}{Fm}
        \arrow{ur}[swap]{F\chi_\emptyset^B}
      \end{tikzcd}
    \]
    Both triangles commute because $\chi_\emptyset^B = \chi_S^B\cdot m$ and
    $\chi_{\emptyset}^B = \chi_\emptyset^B\cdot m$. Thus, we conclude
    \[
      F\chi_S^B(c(x))
      = F\chi_S^B(Fm(t'))
      = F\chi_\emptyset^B(t')
      = F\chi_\emptyset^B(Fm(t'))
      = F\chi_\emptyset^B(c(x)).\tag*{\qedhere}
    \]
  \end{proof}
  \begin{lemma} \label{stableP}
    For all $(x,x') \in P_i$ and $B\in C/Q_i$ in \autoref{coalgPT}, we have
    \[
      F\chi_\emptyset^B(c(x)) = F\chi_\emptyset^B(c(x')).
    \]
  \end{lemma}
  \begin{proof}
    One can show\cite[Prop.~4.12]{concurSpecialIssue} that in every
    iteration there is a map $c_i\colon C/P_i\to F(C/Q_i)$ that
    satisfies $F[-]_{Q_i}\cdot c= c_i\cdot [-]_{P_i}$:
    \[
      \begin{tikzcd}
        C
        \arrow{r}{c}
        \arrow[->>]{d}[swap]{[-]_{P_i}}
        & FC
        \arrow[->>]{d}{F[-]_{Q_i}}
        \\
        C/P_i
        \arrow{r}{c_i}
        & F(C/Q_i)
      \end{tikzcd}
    \]
    where the maps $[-]_{P_i}, [-]_{Q_i}$ send elements of $C$ to their
    equivalence class (\autoref{eqEqClass}). The map $\chi_\emptyset^B\colon
    C\to 3$ for $B\in C/Q_i$ can be decomposes as:
    \[
      \begin{tikzcd}
        C
        \arrow[->>]{d}[swap]{[-]_{Q_i}}
        \arrow{r}{\chi_\emptyset^B}
        & 3
        \\
        C/Q_i
        \arrow{ur}[swap]{\chi_\emptyset^{\{B\}}}
      \end{tikzcd}
    \]
    Combining these two diagrams, we obtain:
    \begin{equation}\label{eqChiEq}
      F\chi_\emptyset^B\cdot c
      = F\chi_\emptyset^{\{B\}}\cdot F[-]_{Q_i}\cdot c
      = F\chi_\emptyset^{\{B\}}\cdot c_i\cdot [-]_{P_i}.
    \end{equation}
    For all $(x,x') \in P_i$, we have $[x]_{P_i} = [x']_{P_{i}}$, and
    thus we have
    \[
      F\chi_\emptyset^B(c(x))
      \overset{\eqref{eqChiEq}}= F\chi_\emptyset^{\{B\}}(c_i([x]_{P_i}))
      = F\chi_\emptyset^{\{B\}}(c_i([x']_{P_i}))
      \overset{\eqref{eqChiEq}}= F\chi_\emptyset^B(c(x')).
      \tag*{\qedhere}
    \]
  \end{proof}
  \begin{lemma} \label{noEdgeNoSplit}
    For $S\subsetneqq B\in C/Q_i$ in the $i$th iteration of \autoref{coalgPT},
    any block $T\in C/P_i$ with no edge to $S$ is not modified; in
    symbols:\smnote{}
    
    \[
      T\not\to S \quad\Longrightarrow\quad
      T\in C/P_{i+1}
      \tag*{for all \(T\in C/P_i\).}
    \]
  \end{lemma}
  \begin{proof}
    Since $T\not\to S$, we have
    $\{x\}\not\to S$ and $\{x'\}\not\to S$ for all $x,x'\in T$. Thus,
    \begin{align*}
      F\chi_S^B(c(x))
      &= F\chi_\emptyset^B(c(x))
        \tag{\autoref{noEdgeNoChiS}, $\{x\}\not\to S$}
        \\
      &= F\chi_\emptyset^B(c(x'))
        \tag{\autoref{stableP}, $(x,x')\in P_i$}
        \\
      &= F\chi_S^B(c(x'))
        \tag{\autoref{noEdgeNoChiS}, $\{x'\}\not\to S$}
    \end{align*}
    as desired. \qedhere
  \end{proof}

  \begin{lemma}
    \label{newBlockCount}
    For  $S\subseteq C$ and finite $C$ in the $i$th iteration of
    \autoref{coalgPT}, 
    \[
      |\{T'\in C/P_{i+1}\mid T'\not\in C/P_i\}| ~\le~ 2\cdot |\pred(S)|.
    \]
  \end{lemma}
  \begin{proof}
    Let $S\subsetneqq B\in C/Q_i$ be used for splitting in iteration
    $i$. By contraposition, \autoref{noEdgeNoSplit} implies that if
    $T'\in C/P_{i+1}$ and $T'\not\in C/P_i$, then (the unique)
    $T\in C/P_i$ with $T'\subseteq T$ satisfies $T\not\in C/P_{i+1}$
    and thefore has a transition to~$S$. By the finiteness of $C$, the
    block $T\in C/P_i$ is split into finitely many blocks
    $T_1,\ldots,T_k\in C/P_{i+1}$, representing the equivalence
    classes for $F\chi_S^B\cdot c\colon C\to F3$.  By
    \autoref{noEdgeNoChiS} we know that if $x\in T$ has no transition
    to $S$, then $F\chi_S^B(c(x)) =
    F\chi_\emptyset^B(c(x))$. Moreover, all elements of $T\in C/P_i$
    are sent to the same value by $F\chi_\emptyset^B\cdot c$
    (\autoref{stableP}). Hence, there is at most one block $T_j$ with no
    transition to $S$, and all other blocks $T_{j'}$, $j'\neq j$, have
    a transition to $S$.%
    Therefore the number blocks $T_j$ is bounded above as follows: $k
    \leq |T\cap \pred(S)| + 1$.
    Summing over all predecessor blocks $T$ we obtain:
    \begin{align*}
      &|\{T'\in C/P_{i+1}\mid T'\not\in C/P_i\}|
        \\
      \le~&
      |\{T'\in C/P_{i+1}\mid T'\subseteq T\in C/P_i\text{ and }T\to S\}|
            \tag{\autoref{noEdgeNoSplit}}
          \\
      =~&\sum_{\substack{T\in C/P_i\\ T\to S}} |\{T'\in C/P_{i+1}\mid T'\subseteq T\}|
          \\
      \leq~&\sum_{\substack{T\in C/P_i\\ T\to S}} (|T\cap \pred(S)| + 1)
      \tag{bound on $k$ above}
          \\
      \le~& 2\cdot \sum_{\substack{T\in C/P_i\\ T\to S}} |T\cap \pred(S)|
      \tag*{$|T\cap \pred(S)| \ge 1$}
      \\
      \le~& 2\cdot |\pred(S)| \tag{$T\in C/P_i$ are disjoint}
    \end{align*}
    This completes the proof \qedhere
  \end{proof}
  \begin{lemma} \label{totalBlockCount} Throughout the execution of
    \autoref{coalgPT} for an input coalgebra $(C,c)$ with $n=|C|$ states and $m$
    transitions, we have
    \[
      |\{T\subseteq C\mid T\in C/P_i\text{ for some }i\}|
      \le 2\cdot m \cdot \log_2 n + 2\cdot m + n.
    \]
  \end{lemma}
  \begin{rem*}
    Note that the proof is similar to arguments given in the complexity analysis
    of the Paige-Tarjan algorithm; for instance, compare to
    \cite[p.~980]{PaigeTarjan87} (or \cite[Lem.~7.15]{concurSpecialIssue}).
  \end{rem*}
  \begin{proof}
    Because $|S|\le \frac{1}{2}\cdot |B|$ holds in step \ref{step1} of
    \autoref{coalgPT}, one can show that every state $x\in C$ is
    contained in the set $S$ in at most $(\log_2(n)+1)$
    iterations~\cite[Lem.~7.15]{concurSpecialIssue}.
    More formally, let
    $S_i\subsetneq B_i\in C/Q_i$ be the blocks picked in the $i$th iteration 
    of \autoref{coalgPT}. Then we have
    \begin{equation}
    |\{S_i\mid x \in S_i\}| \le \log_2 n + 1
      \qquad\text{for all }x\in C.
      \label{timesInSubblock}
    \end{equation}
    Let the algorithm terminate after $\ell$ iterations
    returning~$C/P_\ell$. Then, the number of new blocks introduced by
    step \ref{defPi1} is bounded as follows (note that after the third step,
    $x \in S_i$ is a side condition enforcing that we have a summand
    $|\pred(\set{x})|$ provided that $x$ lies in $S_i$, whereas before
    we sum over all $x \in S_i$): \allowdisplaybreaks
    \begin{align*}
      &\sum_{0\le i< \ell}
      |\{T'\in C/P_{i+1}\mid T' \notin C/P_i\}|
      \\ \le~& \sum_{0\le i< \ell} 2\cdot |\pred(S_i)|
               \tag{\autoref{newBlockCount}}
      \\ \le~& 2\cdot \sum_{0\le i< \ell} \,\sum_{x\in S_i}|\pred(\{x\})|
      \\ =~& 2\cdot \sum_{x\in C}\,\sum_{\substack{0\le i< \ell\\x\in S_i}} |\pred(\{x\})|
      \\ =~& 2\cdot \sum_{x\in C}\, |\pred(\{x\})|\cdot \sum_{\substack{0\le i< \ell\\x\in S_i}} 1
      \\ =~& 2\cdot \sum_{x\in C}\, |\pred(\{x\})|\cdot (\log_2n + 1)
             \tag*{by~\eqref{timesInSubblock}}
      \\ =~& 2\cdot m \cdot (\log_2n + 1)
             = 2\cdot m\cdot \log_2 n + 2\cdot m
    \end{align*}
    The only blocks we have not counted so far are the blocks of $C/P_0$, since
    $|C/P_0|\le n$, we have at most $2\cdot m\cdot \log_2 n + 2\cdot m+ n$
    different blocks in $(C/P_i)_{0\le i <\ell}$.
    \qedhere
  \end{proof}
  We are now ready to prove the main theorem on the dag size of
  formulae created by \autoref{algoCerts}.
  \begin{proof}[Proof of \autoref{certSize}]
    Regarding the height of the dag, it is immediate that~$\delta_i$ and $\beta_i$
    have a height of at most $i+1$. Since $|C/Q_{i}| < |C/Q_{i+1}| \le |C|=n$ for
    all $i$, there are at most $n$ iterations, with the final partition being
    $C/P_{n+1} = C/Q_{n+1}$.

    In \autoref{algoCerts} we create a new modal operator formula whenever \autoref{coalgPT}
    creates a new block in $C/P_i$. By \autoref{totalBlockCount}, the number of
    modalities in the dag is thus bounded by
    \[
      2\cdot m \cdot \log_2 n + 2\cdot m + n
    \]
    In every iteration of the main loop, $\beta$ is extended by two new
    formulae, one for $S$ and one for $B\setminus S$. The formula
    $\beta_{i+1}(S)$ does not increase the size of the dag, because no new
    node needs to be allocated. For $\beta_{i+1}(B\setminus S)$, we need to
    allocate one new node for the conjunction, so there are at most $n$ new
    such nodes allocated throughout the execution of the whole algorithm.
    Even if
    the optimization in \autoref{cancelConjunct} is applied, the additional run
    time can be neglected under the $\CO$-notation.\twnote{}
    \qedhere
  \end{proof}
\end{proofappendix}
\begin{proofappendix}{runTimePreserved}
  We implement every operation of \autoref{algoCerts} in constant time. The
  arrays for $\beta$ and $\delta$ are re-used in every iteration. Hence
  the index $i$ is entirely neglected and only serves as an indicator for
  whether we refer to a value before or after the loop iteration. We
  proceed by case distinction as follows:
  \begin{enumerate}
  \item Initialization step: 
    \begin{itemize}
    \item The only block $\{C\}$ in $C/Q_0$ has index 0, and so we make $\beta(0)$
      point to the node $\top$.
    \item For every block $T$ in $C/P_0$, \autoref{coalgPT} has computed
      $F!(c(x)) \in F1$ for some (in fact every) $x\in T$. Since $F1$ canonically
      embeds into $F3$ (\autoref{notationF1Mod}), we create a new node labelled
      $\fmod{Fj_1(F!(c(x)))}$ with two edges to~$\top$.
      
      For every $T\in C/P_0$, this runs in constant time and can be performed whenever
      the original \autoref{coalgPT} creates a new such block $T$.
    \end{itemize}
    
  \item In the refinement step, we can look up the certificates
    $\delta_i(S)$ resp.~$\beta_i(B)$ for $S$ resp.~$B$ in constant
    time using the indices of the blocks $S$ and $B$. Whenever the
    original algorithm creates a new block, we also immediately
    construct the certificate of this new block by creating at most
    two new nodes in the dag (with at most four outgoing
    edges). However, if a block does not change (that is,
    $[x]_{Q_i} = [x]_{Q_{i+1}}$ or $[x]_{P_i} = [x]_{P_{i+1}}$,
    resp.), then the corresponding certificate is not changed either
    in steps \autoref{defBetai1} resp.~\autoref{defDeltai1}.

    \smallskip
    \noindent
    In the loop body we update the certificates as follows:
    \begin{enumerate}[(I)]
    
  \item[\ref{defBetai1}]
    The new block $S\in C/Q_{i+1}$ just points to the certificate $\delta_i(S)$
    constructed earlier. For the new block $(B\setminus S) \in C/Q_{i+1}$, we
    allocate a new node $\wedge$, with one edge to $\beta_i(B)$ and one
    negated edge to $\delta_i(S)$. (See also details for
    \autoref{cancelConjunct} on the run time for computing the optimized negation.)
    
  \item[\ref{defDeltai1}]

    Not all resulting blocks have a transition to $S$. There may be
    (at most) one new block $T'\in C/P_{i+1}$, $T'\subseteq T$ with no
    transition to $S$ (see the proof of \autoref{newBlockCount}). In
    the refinable partition structure, such a block will inherit the
    index from $T$ (i.e.~the index of $T$ in $C/P_i$ equals the index
    of $T'$ in $C/P_{i+1}$). Moreover, every $x\in T'$ fulfils
    $F\chi_S^B(c(x)) = F\chi_\emptyset^B(c(x))$ (by
    \autoref{noEdgeNoChiS}), %
    and $F\chi_\emptyset^B(c(x)) = F\chi_\emptyset^B(c(y))$ for every
    $y\in T$ (by \autoref{stableP}). %

    Now, one first saves the node of the certificate $\delta_i(T)$ in some
    variable $\delta'$, say. Then the array $\delta$ is updated at index $T$ by the formula
    \[
      \fmod{F\chi_\emptyset^B(c(y))}(\delta_i(S),\beta_i(B))
      \qquad\text{for an arbitrary $y\in T$.}
    \]
    Consequently, any block $T'$ inheriting the index of $T$
    automatically has the correct certificate.
    
    The allocation of nodes for this formula is completely analogous
    to the one for an ordinary block $[x]_{P_{i+1}} \subsetneqq T$ having edges to $S$: One
    allocates a new node labelled $\wedge$ with edges to the saved node $\delta'$
    (the original value of $\delta_i(T)$) and to another newly allocated node labelled
    $\fmod{F\chi_S^B(c(x))}$ with edges to the nodes~$\delta_i(S)$ and $\delta_i(B)$.
    \qed
    \end{enumerate}
\end{enumerate}
\end{proofappendix}

\begin{proofappendix}[Details for]{cancelConjunct}
  \smnote{}
  
  In order to keep the formula
  size smaller, one can implement the optimization of
  \autoref{cancelConjunct}, but one has to take care not to increase the run
    time. To this end, mark every modal operator node $\fmod{t}(\delta,\beta)$
    in the formula dag with a boolean flag expressing whether:
    \begin{center}
      $\fmod{t}(\delta,\beta)$ is a conjunct of some $\beta_i$-formula.
    \end{center}
    Thus, every new modal operator in \ref{defDeltai1} is flagged
    `false' initially. When splitting the block $B$ in $C/Q_i$ into
    $S$ and $B\setminus S$ in step \ref{defBetai1}, the formula for
    block $B\setminus S$ is a conjunction of $\beta_i(B)$ and the
    negation of all `false'-marked conjuncts
    of~$\delta_i(S)$. Afterwards these conjuncts are all marked `true', because
    they are inherited by $\beta_i(S)$. The `false'-marked conjuncts
    always form a prefix of all conjuncts of a formula in
    $\delta_i$. It therefore suffices to greedily take conjuncts from the
    root of a formula graph while they are marked `false'.

    As a consequence, step \ref{defDeltai1} does not run in constant
    time but instead takes as many steps as there are `false'-marked
    conjuncts in~$\delta_i(S)$. However, over the whole execution of
    the algorithm this eventually amortizes because every newly
    allocated modal operator allocated is initially marked `false' and
    later marked `true' precisely once.
\end{proofappendix}

\begin{proofappendix}{thmLogRunTime}
  The overall run time is immediate, because the underlying \autoref{coalgPT}
  has run time $\CO((m+n)\cdot \log n\cdot p(c))$ and \autoref{algoCerts}
  preserves this run time by \autoref{runTimePreserved}.
\end{proofappendix}
  
\appendixsect{sec:cancellative}{Cancellative Functors}\label{sec:A-cancellative}

\begin{proofappendix}{monoidValCancellative}
  First note that for $FX=M^{(X)}$ the maps in \autoref{funCancellative}
  are defined by:
  \[
    \begin{array}{r@{\ }l@{\qquad}l}
      M^{(\chi_{\set{1,2}})}\colon & M^{(3)}\to M^{(2)}, & t \mapsto
      (t(0), t(1)+t(2)), \\[5pt]
      M^{(\chi_{\set{2}})}\colon & M^{(3)}\to M^{(2)}, & t \mapsto (t(0)+t(1), t(2)),
    \end{array}
  \]
  where we write $s\in M^{(2)}$ as the pair $(s(0), s(1))$.
  
  For $(\Leftarrow)$, let $s,t\in M^{(3)}$ with
  \[
    \fpair{M^{(\chi_{\set{1,2}})},M^{(\chi_{\set{2}})}}(s)
    = \fpair{M^{(\chi_{\set{1,2}})},M^{(\chi_{\set{2}})}}(t),
  \]
  which is written point-wise as follows:
  \begin{align*}  
    (s(0), s(1)+s(2))
    &= (t(0), t(1)+t(2)) \\
    (s(0)+s(1), s(2))
    &= (t(0)+t(1), t(2)).
  \end{align*}
  Hence, $s(0) = t(0)$, $s(2) = t(2)$ and moreover
  \[
    s(1) + s(2) = t(1) + t(2) = t(1) + s(2).
  \]
  Since $M$ is cancellative, we have $s(1) = t(1)$, which proves that
  $s=t$. Thus, the map
  $\fpair{M^{(\chi_{\set{1,2}})},M^{(\chi_{\set{2}})}}$ is injective.

  For $(\Rightarrow)$, let $a,b,c\in M$ with $c+a = c+b$. Define
  $s,t\in M^{(3)}$ by
  \[
    s(0) = s(2) = c,\quad s(1) = a
    \qquad\text{and}\qquad
    t(0) = t(2) = c,\quad t(1) = b.
  \]
  Thus,
  \[
  \begin{aligned}
    M^{(\chi_{\set{1,2}})}(s) &= (s(0), s(1) + s(2)) \\
    &= (c,a+c) \\
    &= (c,b+c)
    \\
    &= (t(0), t(1) + t(2))
    \\
    &= M^{(\chi_{\set{1,2}})}(t),
  \end{aligned}
  \qquad\qquad
  \begin{aligned}
    M^{(\chi_{\set{2}})}(s)
    & = (s(0) + s(1), s(2))\\
    &= (c+a,c)\\
    &= (c+b,c) \\
    &= (t(0) + t(1), t(2)) \\
    &= M^{(\chi_{\set{2}})}(t).
  \end{aligned}
  \]
  Since $\fpair{M^{(\chi_{\set{1,2}})},M^{(\chi_{\set{2}})}}$ is
  injective, we see that $s =t $ holds. Thus, we have $a=s(1) = t(1) = b$,
  which proves that $M$ is cancellative.\qed
\end{proofappendix}

\begin{proofappendix}{cancellativeClosure}
  \begin{enumerate}
  \item For the constant functor $C_X$ with value $X$,  $C_X\chi_S$ is the
    identity map on $X$ for every set~$S$. Therefore $C_X$ is cancellative.

  \item The identity functor is cancellative because the map $\langle
    \chi_{\set{1,2}},\chi_{\set{2}}\rangle$ is clearly injective.

  \item Let $\alpha\colon F\monoto G$ a natural transformation with injective
    components and let $G$ be cancellative. Combining the naturality squares of
    $\alpha$ for $\chi_{\set{1,2}}$ and $\chi_{\set{2}}$, we obtain the
    commutative square:
    \[
      \begin{tikzcd}
        F3
        \arrow{r}{\fpair{F\chi_{\set{1,2}}, F{\chi_{\set{2}}}}}
        \arrow[>->]{d}[swap]{\alpha_{3}}
        &[15mm]
        F2\times F2
        \arrow{d}{\alpha_{2}\times \alpha_{2}}
        \\
        G3
        \arrow[>->]{r}{\fpair{G\chi_{\set{1,2}}, G{\chi_{\set{2}}}}}
        & G2\times G2
      \end{tikzcd}
    \]
    Every composition of injective maps is injective, and so by standard
    cancellation laws for injective maps,
    $\fpair{F\chi_{\set{1,2}},F\chi_{\set{2}}}$ is injective as well, showing
    that the subfunctor $F$ is cancellative.

  \item Let $(F_i)_{i\in I}$ be a family of cancellative
    functors, and suppose that we have elements $s,t\in
    (\prod_{i\in I}F_i)(3) = \prod_{i \in I}F_i3$ with
    \[
      \big(\prod_{i\in I}F_i\chi_{\set{1,2}}\big)(s)
      = \big(\prod_{i\in I}F_i\chi_{\set{1,2}}\big)(t)
      \quad\text{and}\quad
      \big(\prod_{i\in I}F_i\chi_{\set{2}}\big)(s)
      = \big(\prod_{i\in I}F_i\chi_{\set{2}}\big)(t).
    \]
    Write $\pr_i$ for the $i$th projection function from the product.
    For every $i \in I$ we have: 
    \[
      F_i\chi_{\set{1,2}}(\pr_i(s))
      = F_i\chi_{\set{1,2}}(\pr_i(t))
      \qquad\text{and}\qquad
      F_i\chi_{\set{2}}(\pr_i(s))
      = F_i\chi_{\set{2}}(\pr_i(t)).
    \]
    Since every $F_i$ is cancellative, we have $\pr_i(s)=\pr_i(t)$ for every
    $i\in I$. This implies~$s=t$ since the product projections $(\pr_i)_{i\in I}$
    are jointly injective.
    \smnote{}
    
  \item Again, let $(F_i)_{i\in I}$ be a family of cancellative
    functors. Suppose that we have elements $s,t\in
    (\coprod_{i\in I}F_i)(3) = \coprod_{i \in I}F_i3$ satisfying%
    \smnote{}
    \[%
      \big(\coprod_{i\in I}F_i\chi_{\set{1,2}} \big)(s)
      =  \big(\coprod_{i\in I}F_i\chi_{\set{1,2}} \big)(t)
      \quad\text{and}\quad
       \big(\coprod_{i\in I}F_i\chi_{\set{2}} \big)(s)
      =  \big(\coprod_{i\in I}F_i\chi_{\set{2}} \big)(t).
    \]
    This implies that there exists an $i\in I$ and $t',s'
    \in F_i$ with $s = \inj_i (s')$, $t=\inj_i(t')$, and 
    \[
      F_i\chi_{\set{1,2}}(s)
      = F_i\chi_{\set{1,2}}(t)
      \qquad\text{and}\qquad
      F_i\chi_{\set{2}}(s)
      = F_i\chi_{\set{2}}(t).
    \]
    Since $F_i$ is cancellative, we have $s=t$ as desired.
    \qed
  \end{enumerate}
\end{proofappendix}

\begin{proofappendix}[Details for]{cancellativeVsZippable}
  \begin{center}
    \begin{tabular}{@{}lll@{}}
      \toprule
      Operation & cancellative & non-cancellative \\
      \midrule
      Quotient & $X  \mapsto \coprod_{n \in \N}X^n$ & $\Powf$ \\
      Composition & $\Bag=\N^{(-)}$ & $\Bag\Bag$ \\
      \bottomrule
    \end{tabular}
    \hfill
    \begin{tabular}{@{}lll@{}}
      \toprule
      & cancellative & non-cancellative \\
      \midrule
      zippable &
              $X\mapsto X$
                 & $\Powf$
      \\
      non-zippable & see \eqref{size4sets}
                   & $\Powf\Powf$
      \\
      \bottomrule
    \end{tabular}
  \end{center}
  \begin{enumerate}
\item Cancellative functors are not closed under quotients: e.g.~the
  non-cancellative functor $\Powf$ is a quotient of the signature
  functor $X  \mapsto \coprod_{n \in \N}X^n$ (which is cancellative by \autoref{cancellativeClosure}).

\item Cancellative functors are not closed under composition. For the additive
  monoid $(\N,+,0)$ of natural numbers, the monoid-valued functor $\Bag =
  \N^{(-)}$ sends $X$ to the set of finite multisets on~$X$ (`bags'). Since $\N$
  is cancellative, $\Bag$ is a cancellative functor. However, $\Bag\Bag$ is not:
  \begin{align*}
    &\fpair{\Bag\Bag\chi_{\set{1,2}},\Bag\Bag\chi_{\set{2}}}\big(\bag[\big]{\bag{0,1},\bag{1,2}}\big)
    \\
    &= \big(\bag[\big]{\bag{0,1},\bag{1,1}}, \bag[\big]{\bag{0,0},\bag{0,1}}\big) \\
    &= \big(\bag[\big]{\bag{0,1},\bag{1,1}}, \bag[\big]{\bag{0,1},\bag{0,0}}\big) \\
    &=\fpair{\Bag\Bag\chi_{\set{1,2}},\Bag\Bag\chi_{\set{2}}}\big(\bag[\big]{\bag{0,2}, \bag{1,1}}\big)
  \end{align*}
  Here, we use $\bag{\cdots}$ to denote multisets,
  so~$\bag{0,1} = \bag{1,0}$ but $\bag{1}\neq \bag{1,1}$.

\item The identity functor $X\mapsto X$ is both
  zippable~\cite{concurSpecialIssue} and cancellative~(\autoref{cancellativeClosure}).
\item The monoid-valued functor $\Powf = \BoolMonoid^{(-)}$ is
  zippable~\cite{concurSpecialIssue}, but not
  cancellative~(\autoref{monoidValCancellative}), because $\BoolMonoid$ is a
  non-cancellative monoid.

  \item
  The functor $\Pow\Pow$ is neither
  zippable~\cite[Ex.~5.10]{concurSpecialIssue} nor cancellative because
  \begin{align*}
    \fpair{\Pow\Pow\chi_{\set{1,2}},\Pow\Pow\chi_{\set{2}}} (\big\{\set{0}, \set{2}\big\})
    &= (\set[\big]{\set{0}, \set{1}}, \set[\big]{\set{0}, \set{1}})
    \\ &= \fpair{\Pow\Pow\chi_{\set{1,2}},\Pow\Pow\chi_{\set{2}}} (\big\{\set{0},
    \set{1}, \set{2}\big\}).
  \end{align*}
  
  \item
  Every functor $F$ satisfying $|F(2+ 2)| > 1$ and $|F3| = 1$ is
  cancellative but not zippable:
  \begin{itemize}
  \item Indeed, every map with domain $1$ is injective, in particular
    the map
    \[
      \fpair{F\chi_{\set{1,2}},F\chi_{\set{2}}}\colon 1\cong F3 \longrightarrow F2\times F2,
    \]
    whence $F$ is cancellative.

  \item If $|F(2+2)|>1$ and $|F3|=1$ we have that the map
    \[
      \fpair{2+\mathord{!},\mathord{!}+2}
      \colon \underbrace{F(2+2)}_{|-|\mathrlap{\, > 1}} \to
      \underbrace{F(2+1)}_{\cong F3\cong 1} \times
      \underbrace{F(1+2)}_{\cong F3\cong 1}\cong 1
    \]
    is not injective, whence $F$ is not zippable.
  \end{itemize}
  An example for such a functor is given by\twnote{}
  \begin{equation}
    FX = \{ S \subseteq X : |S| = 0 \text{ or }|S| = 4\}
    \label{size4sets}
  \end{equation}
  which sends a map $f\colon X\to Y$ to the map $Ff\colon FX\to FY$ defined by
  \[
    Ff(S) = \begin{cases}
      f[S] & \text{if }|f[S]| = 4 \\
      \emptyset & \text{otherwise}.
    \end{cases}
  \]

\item For the proof of
  \[
    F\text{ zippable }\&~
    F\text{ cancellative }
    ~~\Longrightarrow~~
    F\text{ $m$-zippable }
    \qquad\text{for all }m\in \N
  \]
  recall from König et~al.~\cite{KoenigEA20} that a functor $F$
  is $m$-zippable
  if the canonical map
  \[
    \unzip_m\colon~~ F(A_1+A_2+\ldots + A_m)
    ~~\longrightarrow~~ F(A_1 + 1)
    \times F(A_2 + 1)
    \times\ldots
    \times F(A_m + 1)
  \]
  is injective. Formally, $\unzip_m$ is given by
  \[
    \fpair{F[\Delta_{i,j}]_{j\in \bar m}}_{i\in \bar m}\colon~~
    F\coprod_{j = 1}^{m} A_j
    \longrightarrow
    \prod_{i = 1}^{m} F(A_i + 1)
  \]
  where $\bar m$ is the set $\bar m = \set{1,\ldots,m}$ and the map $\Delta_{i,j}$ is defined by
  \[
    \Delta_{i,j}\colon A_j\to A_i + 1
    \qquad
    \Delta_{i,j} := \begin{cases}
      A_j \xrightarrow{\inl} A_i + 1
      &\text{if }i = j \\
      A_j \xrightarrow{!} 1 \xrightarrow{\inr} A_i + 1
      & \text{if }i \neq j.
    \end{cases}
  \]
  First, we show that for a zippable and cancellative functor $F$, the map
  \[
    g_{A,B} ~~:=~~
    F(A+1+B)
    \xrightarrow{\fpair{F(A+!), F(!+B)}}
    F(A+1)\times F(1+B)
  \]
  is injective for all sets $A,B$. Indeed, we have
  the following chain of injective maps, where the index at the $1$ is only
  notation to distinguish coproduct components more easily:
  \allowdisplaybreaks
  \begin{align*}
    & F(A+(1_M +B))
    \\
    & \quad\monodown \fpair{F(A+!), F(!+(1_M+B))} \tag{$F$ is zippable}
    \\
    & F(A+1) \times F(1_A+1_M+B)
    \\
    & \quad\monodown \id\times \fpair{F(!+B),F(1_A+1_M+!)}\tag{$F$ is zippable}
    \\
    & F(A+1)\times F(1+B) \times F(1_A+1_M+1_B)
    \\
    & \quad\monodown \id\times \id \times \fpair{F\chi_{1_M+1_B}, F\chi_{1_B}} \tag{$F$ is cancellative, $1_A+1_M+1_B \cong \{0,1,2\}$}
    \\
    & F(A+1)\times F(1+B) \times F2 \times F2
  \end{align*}
  Call this composition $f$.
  The injective map $f$ factors through $g_{A,B}$, because it matches with $g_{A,B}$ on the
  components $F(A+1)$ and $F(1+B)$, and for the other components, one has the
  map
  \[
    h ~:=~
    F(A+1)\times F(1+B)
    \xrightarrow{F\chi_{1}\times F\chi_B}
    F2 \times F2
  \]
  with $f = \fpair{\id_{F(A+1)\times F(1+B)}, h}\cdot g_{A,B}$. Since $f$ is
  injective, $g_{A,B}$ must be injective, too.
  
  Also note that a function $F$ is cancellative iff equivalently the map
  \[
    \fpair{F(1+!), F(!+1)}\colon 
    F(1+1+1)\longrightarrow F(1+1)\times F(1+1)
  \]
  is injective, for $!\colon 1+1 \to 1$ and $1+1+1\cong 3$ and $1+1 \cong 2$.

  We now proceed with the proof of the desired implication by induction on $m$.
  In the base cases $m=0$ and $m=1$ there is nothing to show because every functor is
  $0$- and $1$-zippable, and for $m=2$, the implication is trivial (zippability
  and $2$-zippability are identical properties). In the
  inductive step, given that $F$ is $2$-zippable, $m$-zippable ($m\ge 2$), and
  cancellative, we show that $F$ is $(m+1)$-zippable.

  We have the following chain of injective maps, where we again annotate some of
  the singleton sets $1$ with indices to indicate from which coproduct
  components they come:
  \begin{align*}
    & F(A_1+\ldots+A_{m-1}+(A_m+A_{m+1}))
    \\
    & \qquad\monodown \unzip_m 
      \tag{$F$ is $m$-zippable}
    \\
    &
    \prod_{i=1}^{m-1} F(A_i + 1) 
    \times F(A_m + A_{m+1} + 1_{1..(m-1)})
    \\
    \cong~
    &
    \prod_{i=1}^{m-1} F(A_i + 1) 
    \times F(A_m + 1_{1..(m-1)} + A_{m+1})
    \\
    &\quad \monodown {\id\times g_{A_m,A_{m+1}}}
      \tag{the above injective helper map $g$}
      \\
    &
    \prod_{i=1}^{m-1} F(A_i + 1) 
    \times F(A_m + 1)
    \times F(1+A_{m+1})
    \\
    \cong~
    &
    \prod_{i=1}^{m-1} F(A_i + 1) 
    \times F(A_m + 1)
    \times F(A_{m+1} + 1)
  \end{align*}
  This composition thus is injective as well, and in fact the composition is
  precisely $\unzip_{m+1}$, showing that $F$ is $(m+1)$-zippable.
  \qed
\end{enumerate}
\end{proofappendix}

The optimization present in the algorithms for
Markov chains~\cite{ValmariF10} and automata~\cite{Hopcroft71} can now
be adapted to coalgebras for cancellative functors, where it suffices
to split only according to transitions into $S$, neglecting
transitions into $B\setminus S$. More formally, this means that we
replace the three-valued $\chi_S^B\colon C \to 3$ with
$\chi_S\colon C\to 2$ in the refinement step \ref{defPi1}: %

\begin{proposition}\label{optimizedPcancel}
  Let $F$ be a cancellative set functor. For $S\in C/P_i$ in the $i$-th
  iteration of \autoref{coalgPT}, we have 
  \(
  P_{i+1} = P_i \cap \ker (C\xrightarrow{c}{FC}\xrightarrow{F\chi_S} F2).
  \)
\end{proposition}

\begin{proofappendix}{optimizedPcancel}
  From the definition~\eqref{eqKer} of the kernel, we immediately
  obtain the following properties for all maps $f,g\colon Y\to Z$,
  $h \colon X\to Y$:
  \begin{align}
    f\text{ injective} ~~&\Longrightarrow~~ \ker(f\cdot h) = \ker(h)
                         \label{kerInjective}
                         \\
    \ker(f) = \ker(g) ~~&\Longrightarrow~~ \ker(f\cdot h) = \ker(g\cdot h)
                         \label{kerPreCompose}
                          \\
    \ker(\fpair{f,g}) &= \ker(f)\cap\ker(g).
                         \label{kerIntersect}
  \end{align}
  \twnote{}
  For every coalgebra $c\colon C\to FC$ and $S\subseteq B\subseteq C$
  we have
  \[
    \fpair{F\chi_B,F\chi_S} = \fpair{F\chi_{\set{1,2}}, F\chi_{\set{2}}}\cdot F\chi_S^B.
  \]
  Since $F$ is cancellative, $\fpair{F\chi_{\set{1,2}}, F\chi_{\set{2}}}$ is
  injective, and we thus obtain
  \begin{equation}
    \ker(\fpair{F\chi_B,F\chi_S})
    = \ker(\fpair{F\chi_{\set{1,2}}, F\chi_{\set{2}}}\cdot F\chi_S^B)
    \overset{\text{\eqref{kerInjective}}}{=} \ker(F\chi_S^B).
  \end{equation}
  By \eqref{kerPreCompose}, this implies that
  \begin{equation}
    \ker(\fpair{F\chi_B,F\chi_S}\cdot c)
    = \ker(F\chi_S^B\cdot c).
    \label{cancelBinaryChi}
  \end{equation}
  
  Let $B\in C/Q_i$ be the block that is split into $S$ and $B\setminus S$ in
  iteration $i$. Since $P_i$ is finer than $Q_i$ and $B\in C/Q_i$, we have 
  $P_i\subseteq Q_i\subseteq \ker(F\chi_B\cdot c)$; thus:
  \begin{equation}
    P_i = P_i \cap \ker(C\xrightarrow{c}{FC} \xrightarrow{F\chi_B} F2).
    \label{PiKerB}
  \end{equation}
  Now we verify the desired property:\smnote{}
  \begin{align*}
    P_{i+1} &= ~P_i\cap \ker(C\xrightarrow{c}{FC}\xrightarrow{F\chi_S^B} F2)
    & \text{(by~\ref{defPi1})}
    \\
    &\overset{\mathclap{\text{}}}{=} 
    ~P_i\cap \ker(\fpair{F\chi_B,F\chi_S}\cdot c)
    & \text{(by~\eqref{cancelBinaryChi})}
    \\
    &= P_i\cap \ker(\fpair{F\chi_B\cdot c,F\chi_S\cdot c})
    &\text{(def.~$\fpair{-,-}$)}
    \\ &
    = P_i\cap \ker(F\chi_B\cdot c)\cap\ker(F\chi_S\cdot c)
    &\text{(by~\eqref{kerIntersect})}
      \\ &
    = P_i\cap\ker(F\chi_S\cdot c)
    &\text{(by~\eqref{PiKerB})}
  \end{align*}
  This completes the proof.
  \qed
\end{proofappendix}
\begin{expl}
  For coalgebras for a signature functor $\Sigma$ or a monoid-valued
  functor $M^{(-)}$ for cancellative~$M$, the refinement step
  \ref{defPi1} of \autoref{coalgPT} can be optimized to compute
  $P_{i+1}$ according to \autoref{optimizedPcancel}.
\end{expl}

Observe that, in the optimized step~\ref{defPi1}, $B$ is no
longer mentioned. It is therefore unsurprising that we do not need a
certificate for it when constructing certificates for the blocks of
$P_{i+1}$. Instead, we can reflect the map
$F\chi_S\cdot c\colon C\to F2$ in the coalgebraic modal formula and
take $F2$ as the (unary) modal operators.  Just like $F1$ in
\autoref{notationF1Mod}, the set $F2$ canonically embeds into $F3$:

\begin{proofappendix}{patchyTheorem}
  Before proving \autoref{patchyTheorem}, we define a new set of (unary)
  modalities (\autoref{notationF2Mod}), establish a lemma about its semantics
  (\autoref{lemF2CoalgSem}), fully phrase the entire optimized algorithm
  (\autoref{algoCertsCancel}), and then show its correctness (\autoref{algoPatchCorrect}).
\end{proofappendix}
\begin{notation} \label{notationF2Mod} Define the injective map
  $j_2\colon 2\monoto 3$ by $j_2(0) = 1$ and $j_2(1)=2$. Then the
  injection $Fj_2\colon F2\monoto F3$ provides a way to interpret
  elements $t\in F2$ as unary modalities $\fmod{t}$:
  \[
    \fmod{t}(\delta) := \fmod{Fj_2(t)}(\delta,\top).
  \]
\end{notation}

\begin{proofappendix}[Remark to]{notationF2Mod}
  There are several different ways to define $\fmod{t}(\delta)$ for $t\in
  F2$, depending on the definition of the inclusion $j_2$.
  \begin{center}
    \def\arraystretch{1.2}
    \begin{tabular}{l@{\qquad}l@{\qquad}l}
      \toprule
      $j_2\colon 2\monoto 3$ & $j_2\cdot \chi_S$ for $S\subseteq C$
      & Definition for $t\in F2$ \\
      \midrule
      $0\mapsto 0, 1\mapsto 1$
      & $j_2\cdot \chi_S = \chi_{\emptyset}^S$
      & $\fmod{t}(\delta) := \fmod{Fj_2(t)}(\bot,\delta)$
        \\
      $0\mapsto 0, 1\mapsto 2$
      & $j_2\cdot \chi_S = \chi_{S}^S$
      & $\fmod{t}(\delta) := \fmod{Fj_2(t)}(\delta,\delta)$
        \\
      $0\mapsto 1, 1\mapsto 2$
      & $j_2\cdot \chi_S = \chi_{S}^C$
      & $\fmod{t}(\delta) := \fmod{Fj_2(t)}(\delta,\top)$
        \\
      \bottomrule
    \end{tabular}
  \end{center}
  All these variants make the following \autoref{lemF2CoalgSem} true because in any case:
  \begin{equation}
    \fmod{t}(\delta) = \fmod{Fj_2(t)}(\phi,\psi)
    \qquad\text{implies}\qquad
    j_2\cdot \chi_{\semantics{\delta}} = \chi_{\semantics{\phi}}^{\semantics{\psi}}.
    \label{eqF2ModProp}
  \end{equation}
\end{proofappendix}
Analogously to \autoref{lemF3CoalgSem} we can show:
\begin{lemma} \label{lemF2CoalgSem} Given a cancellative functor~$F$,
  an $F$-coalgebra $(C,c)$, $t\in F2$, a formula~$\delta$, and
  $x \in C$, we have $x\in \semantics{\fmod{t}(\delta)}$ if and only if
  $F\chi_{\semantics{\delta}}(c(x)) = t$.
\end{lemma}

\noindent 
In \autoref{algoCerts}, the family $\beta$ is only used in the
definition of $\delta_{i+1}$ to characterize the larger block $B$ that
has been split into the smaller blocks $S\subseteq B$ and
$B\setminus S$. For a cancellative functor, we can replace
\[
    \fmod{F\chi_S^B(c(x))}(\delta_i(S),\beta_i(B))
    \quad\text{ with }\quad
    \fmod{F\chi_S(c(x))}(\delta_i(S))
\]
in the definition of $\delta_{i+1}$. Hence, we can omit $\beta_i$ from
\autoref{algoCerts} altogether, obtaining the following algorithm,
which is again based on coalgebraic partition refinement
(\autoref{coalgPT}).

\begin{proof}[Proof of \autoref{lemF2CoalgSem}]
  Since we put $j_2\colon 2\monoto 3$ with $j_2(0) = 1$ and $j_2(1)=2$,
  we have
  $j_2\cdot \chi_{S} = \chi_S^C$ for all $S\subseteq C$.
  \begin{align*}
    \semantics{\fmod{t}(\delta)}
    &= \semantics{\fmod{Fj_2(t)}(\delta,\top)}
      \tag{\autoref{notationF2Mod}}
      \\
    & = \{x \in C \mid F\chi_{\semantics{\delta}}^{C}(c(x)) = Fj_2(t)\}
      \tag{\autoref{lemF3CoalgSem}, $\semantics{\top} = C$}
      \\
    & = \{x \in C \mid Fj_2(F\chi_{\semantics{\delta}}(c(x))) = Fj_2(t)\}
      \tag{$\chi_{\semantics{\delta}}^C = j_2\cdot \chi_{\semantics{\delta}}$}
      \\
    & = \{x \in C \mid F\chi_{\semantics{\delta}}(c(x)) = t\}
      \tag{$Fj_2$ injective}
  \end{align*}
  In the last step, we use that $F$ preserves injective maps (\autoref{trnkovahull})
\end{proof}

\begin{algorithm}\label{algoCertsCancel}
  We extend \autoref{coalgPT} as follows. Initially, define
  \[
    \delta_0([x]_{P_0}) = \fmod{F!(c(x))}.
  \]
  In the $i$-th iteration, extend step \ref{defPi1} by the additional assignment
  \begin{enumerate}[({A$\!'$}1)]
    \setcounter{enumi}{2}
      \item\label{defDeltai1Optimized} $
      \delta_{i+1}([x]_{P_{i+1}}) = \begin{cases}
           \delta_{i}([x]_{P_{i}}) &\text{if }[x]_{P_{i+1}} = [x]_{P_i} \\
          \delta_{i}([x]_{P_{i}}) \wedge \fmod{F\chi_S(c(x))}(\delta_i(S))
          &\text{otherwise.}\\
          \end{cases}
        $
  \end{enumerate}
\end{algorithm}
The certificates thus computed are reduced to roughly half the size
compared to \autoref{algoCerts}; the asymptotic run time and formula
size (\autoref{complexityAnalysis}) remain unchanged. More importantly:
\begin{rem} \label{cancelNegFree} The certificates constructed by
  \autoref{algoCertsCancel} do not contain negation (or disjunction); they are
  built from $\top$, conjunction $\wedge$, and unary modal operators $\fmod{t}$
  for $t \in F2$ (the nullary operators $\fmod{t}$ for $t\in F1$ embed into $F2$).
\end{rem}
\begin{proof}[Details on \autoref{cancelNegFree}]
  Define the injective map $j_{12}\colon 1\monoto 2$ by $j_{12}(0) = 1$. Hence,
  we can also embed the nullary $t\in F1$ into $F2$:
  \[
    \fmod{t} = \fmod{Fj_{12}(t)}(\top)
    \qquad \text{(cf.~\autoref{notationF2Mod})}.
  \]
  This is compatible with the notations established so far because we have
  $j_2\cdot j_{12} = j_1\colon 1 \monoto 3$ for the inclusions defined in \autoref{notationF1Mod}
  and \autoref{notationF2Mod}. Thus, we obtain the same modal operator
  regardless of whether we embed $t\in F1$ first into $F2$ and from there into
  $F3$ ($j_2$, \autoref{notationF2Mod}) or directly into $F3$
  ($j_1$, \autoref{notationF1Mod}):
  \(
    \fmod{t} = \fmod{Fj_{12}(t)}(\top) = \fmod{Fj_2(Fj_{12}(t))}(\top,\top)
    =\fmod{Fj_{1}(t)}(\top,\top).
  \)
\end{proof}

\begin{theorem}\label{algoPatchCorrect}
  For cancellative functors, \autoref{algoCertsCancel} is correct;
  that is, for all $i\in \N$ we have:
  \[
    \forall S\in X/P_i\colon
    \semantics{\delta_i(S)}
    = S.
  \]
\end{theorem}
Note that the optimized \autoref{algoCertsCancel} can also be
implemented by directly constructing certificates for the unary modal
operators $F2$. That is, one can treat the modal operators $F2$ as
first class citizens, in lieu of embedding them into the set $F3$ as
we did in \autoref{notationF2Mod}.  The only difference between the
two implementation approaches w.r.t.\ the size of the formula dag is
one edge per modality, namely the edge to the node $\top$ from the
node $\fmod{Fj_2(F\chi_S(c(x)))}(\delta_i(\delta_i),\top)$, which
arises when step~\ref{defDeltai1Optimized} is expanded according to
\autoref{notationF2Mod}.

\begin{proof}[Proof of \autoref{algoPatchCorrect}]
  We prove the desired correctness by induction over $i$, the index of loop
  iterations.

  The definition of $\delta_0$ is identical to the definition in
  \autoref{algoCerts} whence
  \[
    \semantics{\delta_0(S)} = S
    \qquad
    \text{for all }S\in C/P_0,
  \]
  proved completely analogously as in the proof of \autoref{algoCertsCorrect}.

  In the $i$-th iteration with chosen block $S\in C/P_i$, we
  distinguish two cases, whether a block $[x]_{P_{i+1}} \in C/P_{i+1}$
  remains the same or is split into other blocks: 
  \begin{itemize}
  \item If $[x]_{P_{i+1}} = [x]_{P_{i}}$, then we have
    \[
      \semantics{\delta_{i+1}([x]_{P_{i}})}
      \overset{\text{\ref{defDeltai1Optimized}}}{=}
      \semantics{\delta_{i}([x]_{P_{i}})}
      \overset{\text{I.H.}}{=}
      [x]_{P_i} = [x]_{P_{i+1}}.
    \]
  \item If $[x]_{P_{i+1}} \neq [x]_{P_{i}}$, we compute as follows:
    \begin{align*}
      & \semantics{\delta_{i+1}([x]_{P_{i+1}})}
      \\
      =&
      \semantics{\delta_{i}([x]_{P_{i}})\wedge \fmod{F\chi_S(c(x))}(\delta_i(S))}
      \tag*{\ref{defDeltai1Optimized}}
        \\
      =&
      \semantics{\delta_{i}([x]_{P_{i}})}\cap\semantics{ \fmod{F\chi_S(c(x))}(\delta_i(S))}
      \\
      =&
      [x]_{P_{i}} \cap\semantics{ \fmod{F\chi_S(c(x))}(\delta_i(S))}
        \tag{I.H.}
      \\
      =&
        [x]_{P_{i}} \cap
        \{x'\in C\mid F\chi_{\semantics{\delta_i(S)}}(c(x'))
                       = F\chi_S(c(x)) \}
        \tag{\autoref{lemF2CoalgSem}}
      \\ =&
       [x]_{P_i} \cap 
        \{x'\in C\mid F\chi_{S}(c(x'))
                       = F\chi_S(c(x)) \}
            \tag{I.H.}
      \\ =&
       [x]_{P_i} \cap 
        \{x'\in C\mid (x,x') \in \ker(F\chi_{S}\cdot c) \}
            \tag{def.~$\ker$}
      \\ =&
       [x]_{P_i} \cap 
            [x]_{F\chi_{S}\cdot c}
            \tag{def.~$[x]_{R}$}
      \\ =& [x]_{P_{i+1}} %
    \end{align*}
    The last step is the block-wise definition of 
    $P_{i+1} = P_i\cap \ker(F\chi_S\cdot c)$ (see \autoref{optimizedPcancel}).
    \qedhere
  \end{itemize}
\end{proof}

\appendixsect{domainSpecific}{Domain-Specific Certificates}
\begin{proofappendix}[Details for]{remBoolCombination}
  For every set $X$, define the set $\BoolComb(X)$ as terms $K$ over the grammar
  \begin{equation}\label{eq:gramm}
    K ::= X ~\vert~ \neg K ~\vert~ K\wedge K.
  \end{equation}
  There is an obvious way to evaluate boolean combinations of
  predicates using the maps
  \[
    e_X\colon \BoolComb(2^X)\to 2^X
  \]
  defined inductively as follows:
  \[
    e_X(S\subseteq X) = X,
    \qquad
    e_X(\neg K) = X\setminus K,
    \qquad
    e_X(K_1\wedge K_2) = K_1\cap K_2.
  \]
  Given a signature $\Lambda$ of modal operators $\lambda$ and
  corresponding predicate liftings~$\semantics{\lambda}$, we can
  combine all of them. To this end, write $\Lambda$ for the
  corresponding signature functor
  (cf.~\itemref{ex:powerset}{exSignature}); we define a family of maps
  $\semantics{\Lambda}_X$ as follows:
  \[
    \semantics{\Lambda}_X\colon
    \Lambda (2^X) = \coprod_{\arity[\scriptstyle]{\lambda}{n}\in \Lambda} (2^X)^n
    \xra{\big[\semantics{\lambda}\big]_{\lambda\in \Lambda}}
    2^{FX}.
  \]
  Since every $\semantics{\lambda}_X\colon (2^X)^n\to 2^{FX}$ is natural in $X$,
  so is $\semantics{\Lambda}_X$.
  We can replace $\Lambda$ with the signature
  \[
    \Lambda' := \coprod_{n\in \N} \BoolComb(\Lambda(\BoolComb(n))),
  \]
  where $\inj_n(K) \in \Lambda'$,
  $K \in \BoolComb(\Lambda(\BoolComb(n)))$ has the arity $n$. Observe
  that $\BoolComb$ is functorial; in fact, it is the (free or term) monad for
  the signature functor $\Sigma X = X + X \times X$ associated to the
  grammar in~\eqref{eq:gramm}. Thus $\BoolComb \cdot \Lambda \cdot
  \BoolComb$ is a functor, too. Applying the
  Yoneda-Lemma to this functor, we have for every $\arity{t}{n}\in \Lambda'$ the
  (natural) family of maps $\alpha^t$:
  \[
    \alpha^t_X\colon X^n \to \BoolComb(\Lambda(\BoolComb(X)))
    \qquad\text{for every set }X.
  \]
  Hence, we obtain a predicate lifting for $t$ by defining:
  \[
    \semantics{t}_X\colon
    \!\!
    \begin{tikzcd}[column sep=7mm]
     \big((2^X)^n
      \arrow{r}{\alpha^t_{2^X}}& 
    \BoolComb(\Lambda(\BoolComb(2^X)))
    \arrow{r}[yshift=2mm]{\BoolComb(\Lambda(e_{X}))}& 
    \BoolComb(\Lambda(2^X))
    \arrow{r}[yshift=2mm]{\BoolComb(\semantics{\Lambda}_X)}& 
    \BoolComb(2^{FX})
    \arrow{r}[yshift=2mm]{\BoolComb(e_{FX})}& 
    2^{FX}\big).
    \end{tikzcd}
  \]
  It is a composition of natural transformations and so is itself natural in $X$.
\end{proofappendix}

\begin{defn}
  \label{domainCertSimple}
  Given a modal signature~$\Lambda$ for a functor~$F$, a \emph{simple
    domain-specific interpretation} consists of functions
  $\tau\colon F1 \to \bar\Lambda$ and
  $\kappa\colon F2 \to \bar\Lambda$ assigning a nullary modality
  $\tau_o$ to each $o \in F1$ and a unary modality $\kappa_s$ to
  each $s \in F2$ such that the predicate liftings
  $\semantics{\tau_o}_X\in 2^{FX}$ and
  $\semantics{\kappa_s}\colon 2^X \to 2^{FX}$ satisfy
  \[
    \semantics{\tau_o}_1 = \{f\}
    \quad\text{(in $2^{F1}$)}
    \qquad\text{and}\qquad
    [s]_{F!} \cap \semantics{\kappa_s}_2(\{1\}) = \{s\}
    \qquad\text{(in $2^{F2}$).}
  \]
\end{defn}

\begin{proposition}\label{domainCertCancellative}
  Let $\Lambda$ be a modal signature for a cancellative functor~$F$,
  and $(\tau,\kappa)$ a simple domain-specific interpretation. Define
  $\lambda\colon F3 \to \bar\Lambda$ by
  $\lambda_{t}(\delta,\rho) =
  \kappa_{F\chi_{\set{2}}(t)}(\delta)$. Then $(\tau, \lambda)$ is a
  domain-specific interpretation.
\end{proposition}

\begin{proof}
  We verify that $(\tau,\lambda)$ is a domain-specific interpretation
  (\autoref{domainCert}) by verifying that for every $t\in F3$, defining
  \[
    \tau_{t}(\delta,\rho) = \kappa_{F\chi_{\set{2}}(t)}(\delta)
  \]
  satisfies
  \[
    [t]_{F\chi_{\{1,2\}}} \cap \semantics{\tau_t}_3(\{2\},\{1\}) = \{t\}
    \qquad\text{in $2^{F3}$}.
  \]
  In the following, we put $s := F\chi_{\set{2}}(t) \in F2$.
  By the naturality of the predicate lifting of $\kappa_s$, the
  following square commutes (recall that $2^{(-)}$ is contravariant):
  \begin{equation}\label{diagNat}
    \begin{tikzcd}
      2^2
      \arrow{r}{\semantics{\kappa_s}_2}
      \arrow{d}[swap]{2^{\chi_{\set{2}}}}
      & 2^{F2}
      \arrow{d}{2^{F\chi_{\set{2}}}}
      \\
      2^3
      \arrow{r}{\semantics{\kappa_s}_3}
      & 2^{F3}
    \end{tikzcd}
  \end{equation}
  We thus have:
  \begin{align*}
    \semantics{\tau_t}_3(\set{2},\set{1})
    &=
    \semantics{\kappa_s}_3(\set{2})
      \tag{def.~$\tau_t$}
    \\
    &= \semantics{\kappa_s}_3(\chi_{\set{2}}^{-1}[\set{1}])
    \tag{def.~$\chi_{\set{2}}$}
    \\
    &= \semantics{\kappa_s}_3(2^{\chi_{\set{2}}}(\set{1}))
      \tag{def.~$2^{(-)}$}
    \\
    &= 2^{F\chi_{\set{2}}}(\semantics{\kappa_s}_2(\set{1}))
      \tag{by~\eqref{diagNat}}
    \\
    &= \{t'\in F3\mid F\chi_{\set{2}}(t')\in \semantics{\kappa_s}_2(\set{1})\}
  \end{align*}
  For every $t'\in F3$, we have%
  \smnote{}
  \begin{align*}
    &t' \in [t]_{F\chi_{\{1,2\}}} \cap \semantics{\tau_t}_3(\{2\},\{1\})
      \\
    \Leftrightarrow~&  t'\in [t]_{F\chi_{\{1,2\}}}\text{ and }t'\in\semantics{\tau_t}_3(\{2\},\{1\})
    \\
    \Leftrightarrow~&  t'\in [t]_{F\chi_{\{1,2\}}}\text{ and }F\chi_{\set{2}}(t')\in \semantics{\kappa_s}_2(\set{1})
    \tag{by the above calculation}
    \\
    \Leftrightarrow~&  t'\in [t]_{F\chi_{\{1,2\}}}\text{ and } F\chi_{\set{2}}(t')\in [F\chi_{\set{2}}(t')]_{F!}\cap\semantics{\kappa_s}_2(\set{1})
    \tag{$\ker$ reflexive}
    \\
    \Leftrightarrow~&  t'\in [t]_{F\chi_{\{1,2\}}}\text{ and } F\chi_{\set{2}}(t')\in [F\chi_{\set{2}}(t)]_{F!}\cap\semantics{\kappa_s}_2(\set{1})
    \tag{$t'\in [t]_{F\chi_{\set{1,2}}}$}
    \\
    \Leftrightarrow~&  t'\in [t]_{F\chi_{\{1,2\}}}\text{ and } F\chi_{\set{2}}(t')\in [s]_{F!}\cap\semantics{\kappa_s}_2(\set{1})
        \tag{def.~$s$}
    \\
    \Leftrightarrow~&  t'\in [t]_{F\chi_{\{1,2\}}}\text{ and } F\chi_{\set{2}}(t')\in \{s\}
        \tag{assumption on $\kappa_s$}
    \\
    \Leftrightarrow~& t'\in [t]_{F\chi_{\{1,2\}}}\text{ and } F\chi_{\set{2}}(t')\in \{F\chi_{\set{2}}(t)\}
        \tag{def.~$s$}
    \\
    \Leftrightarrow~& F\chi_{\{1,2\}}(t') = F\chi_{\{1,2\}}(t)\text{ and } F\chi_{\set{2}}(t') = F\chi_{\set{2}}(t)
    \\
    \Leftrightarrow~& \fpair{F\chi_{\{1,2\}},F\chi_{\{2\}}}(t') = \fpair{F\chi_{\{1,2\}},F\chi_{\{2\}}}(t)
    \tag{def.~$\fpair{-,-}$}
                    \\
    \Leftrightarrow~& t' = t
                      \tag{$F$ cancellative}
  \end{align*}
  Note that $\fpair{F\chi_{\{1,2\}},F\chi_{\{2\}}}$ is injective because $F$ is cancellative.
\end{proof}

\begin{proofappendix}[Details for]{exDomainSpecInt}
  \begin{enumerate}
    \item
  We verify that \itemref{exDomainSpecInt}{F3Pow} indeed provides domain-specific certificate (\autoref{domainCert}).
  For $t\in \Pow 3$, we have
  \[
    \lambda_t(\delta,\rho) ~=~
    \begin{cases}
      \neg\Diamond \rho &\text{if }2\in t\notni 1 \\
      \Diamond \delta \wedge \Diamond\rho &\text{if }2\in t \ni 1 \\
      \neg\Diamond \delta &\text{if }2\not\in t\ni 1  \\
      \top &\text{if }2\not\in t \notni 1 \\
    \end{cases}
  \]
  We proceed by the following case distinction:
  \begin{itemize}
  \item If $1\notin \Pow \chi_{\set{1,2}}(t)$, then $t= \emptyset$ or
    $t= \{0\}$. In both cases we have $[t]_{\Pow \chi_{1,2}} = \{t\}$. Since~$\lambda_t(\delta,\rho) =\top$,
    \[
      [t]_{\Pow \chi_{1,2}} \cap\semantics{\lambda_t}_3(\set{2},\set{1})= \{t\}
    \]
    as desired.
  \item If $1\in\Pow \chi_{\set{1,2}}(t)$, then
    $2\in t$ or $1\in t$. This yields 
    \begin{align*}
      2\in t \notni 1&\Longrightarrow~ \semantics{\overbrace{(\delta,\rho)\mapsto
                       \mathrlap{\neg\Diamond \rho}\phantom{\Diamond\delta\wedge\Diamond\rho}}^{\textstyle\lambda_t}}_3(\set{2}, \set{1}) = \{t'\in F3 \mid 1\notin t' \}
                      \\
      2\in t \ni 1 &\Longrightarrow~ \semantics{(\delta,\rho)\mapsto \Diamond\delta\wedge\Diamond\rho}_3(\set{2}, \set{1}) = \{t'\in F3 \mid 2\in t'\text{ and }1\in t' \}
                     \\
      2\notin t \ni 1&\Longrightarrow~ \semantics{(\delta,\rho)\mapsto
                       \mathrlap{\neg\Diamond \delta}\phantom{\Diamond\delta\wedge\Diamond\rho}}_3(\set{2}, \set{1}) = \{t'\in F3 \mid 2\notin t' \}
    \end{align*}
    Consequently, we have for every $t'\in
    \semantics{\lambda_t}_3(\set{2},\set{1})$ that 
    \[
      \text{$1\in t$ iff $1\in t'$} 
      \qquad\text{and}\qquad
      \text{$2\in t$ iff $2\in t'$}.
    \]
    To conclude, note that if $t'\in [t]_{\Pow\chi_{\set{1,2}}}$ then $0\in
    t'$ iff $0\in t$. Thus
    \[
      [t]_{\Pow\chi_{\{1,2\}}} \cap \semantics{\lambda_t}_3(\{2\},\{1\}) = \{t\}.
      \tag*{\qed}
    \]
  \end{itemize}

\item \textbf{\itemref{exDomainSpecInt}{dsiSignature}:}
  For the verification for signature functors,
 define a helper map
  $v\colon \Sigma 2\to \Powf \N$ by
  $v(\sigma(x_1,\ldots,x_n))$ $= \{i\in \N\mid x_i = 1\}$.
  The predicate lifting for the (unary) modal operator
    $\gradI$, for $I\subseteq \N$, is obtained from
    \autoref{predLiftYoneda} by the predicate $f_I\colon \Sigma 2\to 2$
    corresponding to the set
    \[
      f_I = \{t\in \Sigma 2\mid v(t) = I\}.
    \]
    This gives rise to the predicate lifting
    \begin{align*}
      \semantics{\gradI}_X(P) &=
      \{t\in \Sigma X\mid F\chi_P(t) \in f_I \}
      \tag{\autoref{predLiftYoneda}}
      \\ &=
      \{t\in \Sigma X\mid v(F\chi_P(t)) = I\}.
      \tag{def.~$f_I$}
    \end{align*}
    Similarly, for the nullary modal operator $\sigma$ (for the $n$-ary operation
    symbol $\arity{\sigma}{n}\in \Sigma$), take  $\Sigma 1 \to 2$
    given by the set 
    \[
      g_\sigma = \{\sigma(0,\ldots,0)\}
    \]
    (noting that $2^0 = 1$). This gives rise to the predicate lifting
    \begin{align*}
      \semantics{\sigma}_X &=
      \{t\in \Sigma X\mid F!(t) \in g_\sigma \}
      \tag{\autoref{predLiftYoneda}}
      \\ &=
      \big\{t\in \Sigma X\mid F\chi_P(t) \in \{\sigma(0,\ldots,0)\}\big\}
      \tag{def.~$g_\sigma$}
      \\ &= \{\sigma(x_1,\ldots,x_n)\mid x_1,\ldots,x_n \in X\}.
    \end{align*}
    \sloppypar
    \noindent
    For the verification of the (simple) domain-specific
    interpretation
    (\autoref{domainCertSimple}),
    we put
    \[
      \kappa_s(\delta) := \grad{v(s)}\delta
      \qquad\text{for }s\in \Sigma 2
    \]
    with then induces the claimed $\lambda_t$ via \autoref{domainCertCancellative}:
    \[
      \lambda_{\sigma(x_1,\ldots,x_n)}(\delta,\rho) = \grad{\{i\in \N\mid x_i = 2\}}\delta
      \qquad\text{for }\sigma(x_1,\ldots,x_n) \in \Sigma 3
    \]
    There is nothing to show for $\tau_o := \sigma$ since it has the correct
    semantics by the definition of $\semantics{\sigma}_1$.
    Note that
    $\fpair{F!,v}\colon \Sigma 2\to \Sigma 1\times \Powf\N$ is
    injective because for every $s\in \Sigma 2$ the operation symbol
    and all its parameters (from $2$) are uniquely determined by
    $F!(s)$ and $v(s)$. For $\kappa_s:=\grad{v(s)}$, $s\in F2$, we
    have
    \[
      [s]_{F!} = \{s' \in F2\mid F!(s) = F!(s')\}.
    \]
    Thus, we compute%
    \smnote{}
    \begin{align*}
      & [s]_{F!} \cap \semantics{\kappa_s}_2(\set{1})
      \\=~& \{ s'\in \Sigma 2\mid s'\in [s]_{F!}\text{ and }s'\in \semantics{\kappa_s}_2(\set{1})\}
      \\=~& \{ s'\in \Sigma 2\mid F!(s) = F!(s')\text{ and }s'\in \semantics{\grad{v(s)}}_2(\set{1})\}
      \\=~& \{ s'\in \Sigma 2\mid F!(s) = F!(s')\text{ and }
      v(F\chi_{\set{1}}(s')) = v(s)\}
      \tag{def.~$\semantics{\grad{v(s)}}_2$}
      \\=~& \{ s'\in \Sigma 2\mid F!(s) = F!(s')\text{ and }
      v(s') = v(s)\}
      \tag{$\id_2=\chi_{\set{1}}\colon 2\to 2$}
      \\=~& \{ s'\in \Sigma 2\mid \fpair{F!,v}(s) = \fpair{F!,v}(s') \}
      \tag{def.~$\fpair{-,-}$}
      \\=~& \{ s\}
      \tag{$\fpair{F!,v}$ injective}
    \end{align*}

  \item \textbf{\itemref{exDomainSpecInt}{dsiMonoid}:}
    For every $m\in M$, define the map
    \[
      f_m\colon M^{(2)}\to 2
      \qquad\text{with}\qquad \{\mu \in M^{(2)}\mid \mu(1) = m\}.
    \]
    which gives rise to the predicate lifting of the unary modal operator
    $\grad{m}$:
    \begin{align*}
      \semantics{\grad{m}}_X(P) &= \{\mu \in M^{(X)} \mid M^{(P)}(\mu) \in f_m\}
      \tag{\autoref{predLiftYoneda}} \\
      &=\{\mu \in M^{(X)} \mid M^{(P)}(\mu)(1) = m\}
      \tag{def.~$f_m$}
    \end{align*}
    For the verification of the axioms of the domain-specific interpretation
    (\autoref{domainCert}), we have that $\tau$ satisfies the axiom:
    \begin{align*}
      \semantics{\tau_o}_1 = \semantics{\grad{o(0)}\top}_1
      &= \{\mu\in M^{(1)}\mid \sum_{x\in \semantics{\top}_1} \mu(x) = o(0)\}
      \\
      &= \{\mu\in M^{(1)}\mid \mu(0) = o(0)\}
      = \{o\}
      \tag{$\semantics{\top}_1 = 1 = \{0\}$}
    \end{align*}
    For the other component of the domain-specific interpretation, we
    proceed by case distinction:
    \begin{itemize}
    \item If $M$ is non-cancellative, we have
      $\lambda_t(\delta,\rho)=\grad{t(2)}\delta\wedge \grad{t(1)}\rho$
      for $t\in M^{(3)}$ and thus we have for every $t'\in M^{(3)}$:
      \begin{align*}
        & t'\in ([t]_{F\chi_{\set{1,2}}}
        \cap \semantics{\lambda_t}_3(\set{2},\set{1}))
        \\ \Leftrightarrow~&
        t'\in [t]_{F\chi_{\set{1,2}}}
        \text{ and } t'\in \semantics{\lambda_t}_3(\set{2},\set{1})
        \\ \Leftrightarrow~&t'\in [t]_{F\chi_{\set{1,2}}}
        \text{ and } t'\in \semantics{(\delta,\rho)\mapsto \grad{t(2)}\delta\wedge \grad{t(1)}\rho}_3(\set{2},\set{1})
        \tag{def.~$\lambda_t$}
        \\ \Leftrightarrow~&t'\in [t]_{F\chi_{\set{1,2}}}
        \text{ and } t'\in \semantics{\grad{t(2)}}_3(\set{2}) \cap \semantics{\grad{t(1)}}_3(\set{1})
        \\ \Leftrightarrow~&t'\in [t]_{F\chi_{\set{1,2}}}
        \text{ and } t'\in \semantics{\grad{t(2)}}_3(\set{2})
        \text{ and } t'\in \semantics{\grad{t(1)}}_3(\set{1})
        \\ \Leftrightarrow~&t'\in [t]_{F\chi_{\set{1,2}}}
        \text{ and } t'(2) = t(2)
        \text{ and } t'(1) = t(1)
        \tag{def.~$\semantics{\grad{m}}$}
        \\ \Leftrightarrow~&t'(0) = t(0)\text{ and }t'(1)+t'(2) = t(1)+t(2)
        \\ &~~~~~
        \text{ and } t'(2) = t(2)
        \text{ and } t'(1) = t(1)
        \\ \Leftrightarrow~&t'(0) = t(0)
        \text{ and } t'(2) = t(2)
        \text{ and } t'(1) = t(1)
        \\ \Leftrightarrow~&t' = t
        \\ \Leftrightarrow~&t' \in \{t\}
      \end{align*}

    \item If $M$ is cancellative, we put \( \kappa_s(\delta) =
      \grad{s(1)}\,\delta \) for $s \in M^{(2)}$, which then induces
      $\lambda_t(\delta,\rho) = \grad{s(2)}\,\delta$ via
      \autoref{domainCertCancellative}. We verify \autoref{domainCertSimple} for
      all $s'\in M^{(2)}$:
      \begin{align*}
        &s' \in ([s]_{F!} \cap \semantics{\kappa_s}_2(\{1\}))
        \\ \Leftrightarrow~&s' \in [s]_{F!} \text{ and } s'\in \semantics{\kappa_s}_2(\{1\})
        \\ \Leftrightarrow~&F!(s') = F!(s) \text{ and } s'\in \semantics{\grad{s(1)}}_2(\{1\})
        \tag{def.~$\kappa_s$}
        \\ \Leftrightarrow~&F!(s') = F!(s) \text{ and } \sum_{x\in \{1\}}s'(x) = s(1)
        \tag{def.~$\grad{s(1)}$}
        \\ \Leftrightarrow~&F!(s') = F!(s) \text{ and } s'(1) = s(1)
        \\ \Leftrightarrow~&s'(0) + s'(1) = s(0) + s(1) \text{ and } s'(1) = s(1)
        \\ \Leftrightarrow~&s'(0) = s(0) \text{ and } s'(1) = s(1)
        \tag{$M$ cancellative}
        \\ \Leftrightarrow~&s' = s
        \\ \Leftrightarrow~&s' \in \{s\}
      \end{align*}
    \end{itemize}

  \item \textbf{\itemref{exDomainSpecInt}{dsiMarkov}:} For $FX=(\Dist X +1)^A$,
    the predicate lifting of $\fpair{a}_p$, $a\in A$, $p\in[0,1]$ is:
    \[
      \semantics{\fpair{a}_p}_X(S) := \{
      \text{if }t\in FX\mid p > 0 \text{ then }t(a) \in \Dist X\text{ and }
      \sum_{x\in S}t(a)(x) \ge p
      \}
    \]
    first note that
    \[
    \semantics{\fpair{a}_{1}\top}_1 = \{o\in F1\mid o(a) \in \Dist 1\}
    \qquad\text{ and }\qquad
    \semantics{\neg\fpair{a}_{1}\top}_1 = \{o\in F1\mid o(a) \in 1\}.
    \]
    Thus, we have:
    \begin{align*}
      \semantics{\tau_o}_1 &=
      \semantics[\big]{
      \bigwedge_{\substack{a\in A\\ o(a) \in \Dist 1}} \fpair{a}_{1}\top
      \wedge
      \bigwedge_{\substack{a\in A\\ o(a) \in 1}} \neg \fpair{a}_{1}\top
      }_1
      \\ & =
      \bigcap_{\substack{a\in A\\ o(a) \in \Dist 1}} \{o'\in F1\mid o'(a)\in \Dist 1\}
      \cap
      \bigcap_{\substack{a\in A\\ o(a) \in 1}} \{o'\in F1\mid o'(a)\in 1\}
      = \{o\}
    \end{align*}
    For the axiom of $\lambda_t$, $t\in F3 = (\Dist 3+1)^A$, we verify for all
    $t'\in F3$, where the crucial step is the arithmetic argument for replacing
    the inequalities by equalities:
    \allowdisplaybreaks
    \begin{align*}
      &t' \in ([t]_{F\chi_{\set{1,2}}}\cap \semantics{\lambda_t}_3(\set{2},\set{1}))
      \\ \Leftrightarrow~&
      t' \in [t]_{F\chi_{\set{1,2}}}\text{ and } t'\semantics{\lambda_t}_3(\set{2},\set{1})
      \\ \Leftrightarrow~&
      t' \in [t]_{F\chi_{\set{1,2}}} \text{ and }
                           t' \in
      \semantics[\big]{(\delta,\rho)\mapsto
      \bigwedge_{\substack{a\in A\\t(a) \in \Dist 3}}(\fpair{a}_{t(a)(2)}\,\delta\wedge \fpair{a}_{t(a)(1)}\,\rho)
      }_3(\set{2},\set{1})
      \\ \Leftrightarrow~&
      t' \in [t]_{F\chi_{\set{1,2}}} \text{ and }
                           t' \in
                           \bigcap_{\substack{a\in A\\t(a) \in \Dist 3}}
      \semantics{(\delta,\rho)\mapsto
      \fpair{a}_{t(a)(2)}\,\delta\wedge \fpair{a}_{t(a)(1)}\,\rho
      }_3(\set{2},\set{1})
      \\ \Leftrightarrow~&
      t' \in [t]_{F\chi_{\set{1,2}}} \text{ and }
                           t' \in
                           \bigcap_{\substack{a\in A\\t(a) \in \Dist 3}}
      \semantics{\fpair{a}_{t(a)(2)}}_3(\set{2})
      \cap\semantics{\fpair{a}_{t(a)(1)}}_3(\set{1})
      \\ \Leftrightarrow~&
      t' \in [t]_{F\chi_{\set{1,2}}} \text{ and }
                           \forall a\in A, t(a)\in \Dist 3:
                           t' \in
      \semantics{\fpair{a}_{t(a)(2)}}_3(\set{2})
      \cap\semantics{\fpair{a}_{t(a)(1)}}_3(\set{1})
      \\ \Leftrightarrow~&
      t' \in [t]_{F\chi_{\set{1,2}}} \text{ and }
                           \forall a\in A, t(a)\in \Dist 3:
                           t'(a)(2) \ge t(a)(2)
                           \wedge
                           t'(a)(1) \ge t(a)(1)
                           \tag{Def.~$\semantics{\fpair{a}p}$}
      \\ \Leftrightarrow~&
              \forall a \in A\colon (t'(a)\in 1\leftrightarrow t(a)\in 1)\text{ and if }a\in \Dist 3\text{ then:}
                           \\ &\qquad
                           t'(a)(0) = t(a)(0),~~ t'(a)(1) + t'(a)(2) = t(a)(1)+t(a)(2),
                           \\ &\qquad
                           t'(a)(2) \ge t(a)(2),~~
                           t'(a)(1) \ge t(a)(1)
      \\ \Leftrightarrow~&
              \forall a \in A\colon (t'(a)\in 1\leftrightarrow t(a)\in 1)\text{ and if }a\in \Dist 3\text{ then:}
                           \\ &\qquad
                           t'(a)(0) = t(a)(0),~~
                           t'(a)(1) = t(a)(2),~~
                           t'(a)(2) = t(a)(2)
                                \tag{arithmetic}
      \\ \Leftrightarrow~&
              \forall a \in A\colon (t'(a)\in 1\leftrightarrow t(a)\in 1)\text{ and if }a\in \Dist 3\text{ then }
                           t'(a) = t(a)
      \\ \Leftrightarrow~& t' \in \set{t}
    \end{align*}
  \end{enumerate}
\end{proofappendix}

\begin{proofappendix}{domainCertMainThm}
  \begin{lemma}
    \label{domainSpecificSetIndexed}
    Let $(\tau,\lambda)$ be a domain-specific interpretation for $F$. 
    For all $t\in FC$ and $S\subseteq B\subseteq C$ we have:
    \[
      \big([t]_{F\chi_B} ~\cap~
      \semantics{\lambda_{F\chi_S^B(t)}}_C(S,B\setminus S)\big)
      = [t]_{F\chi_S^B}
      \qquad\text{in $2^{FC}$.}
    \]
  \end{lemma}
  \begin{proof}
    Put $d:= F\chi_S^B(t)$;
    the naturality square of $\semantics{\lambda_d}$ for $\chi_S^B\colon C\to 3$
    is 
    \[
      \begin{tikzcd}
        2^3\times 2^3
        \arrow{r}{\semantics{\lambda_d}_3}
        \arrow{d}[swap]{2^{\chi_S^B}\times 2^{\chi_S^B}}
        & 2^{F3}
        \arrow{d}{2^{F\chi_S^B}}
        \\
        2^C\times 2^C
        \arrow{r}{\semantics{\lambda_d}_C}
        & 2^{FC}
      \end{tikzcd}
    \]
    Hence:
    \begin{align*}
      (F\chi_S^B)^{-1}\big[\semantics{\lambda_d}_3(\set{2},\set{1})\big]
      &=
      \semantics{\lambda_d}_C((\chi_S^B)^{-1}[\set{2}],(\chi_S^B)^{-1}[\set{1}])
      \\
      &= \semantics{\lambda_d}_C(B,B\setminus S).
      \tag{$*$}
    \end{align*}
    Now we verify for every $t'\in FC$ that
    \allowdisplaybreaks
    \begin{align*}
      &t' \in \big([t]_{F\chi_B} ~\cap~ \semantics{\lambda_{F\chi_S^B(t)}}_C(S,B\setminus S)\big)
      \\
  \Leftrightarrow~& t' \in [t]_{F\chi_B}\text{ and }
                    t'\in \semantics{\lambda_{F\chi_S^B(t)}}_C(S,B\setminus S)
      \\
  \Leftrightarrow~& t' \in [t]_{F\chi_B}\text{ and }
                    t'\in (F\chi_S^B)^{-1}\big[\semantics{\lambda_{F\chi_S^B(t)}}_3(\set{2},\set{1})\big]
               \tag{$*$}
      \\
  \Leftrightarrow~& F\chi_S^B(t') \in [F\chi_S^B(t)]_{F\chi_{\set{1,2}}}\text{ and }
                    F\chi_S^B(t')\in \semantics{\lambda_{F\chi_S^B(t)}}_3(\set{2},\set{1})
                    \tag{$\chi_{\set{1,2}}\cdot \chi_S^B = \chi_B$}
      \\
  \Leftrightarrow~& F\chi_S^B(t') \in [F\chi_S^B(t)]_{F\chi_{\set{1,2}}} \cap
                    \semantics{\lambda_{F\chi_S^B(t)}}_3(\set{2},\set{1})
      \\
  \Leftrightarrow~& F\chi_S^B(t') \in \{F\chi_S^B(t)\}
                    \tag{\autoref{domainCert}}
      \\
  \Leftrightarrow~& F\chi_S^B(t') = F\chi_S^B(t)
      \\
  \Leftrightarrow~& t' \in [t]_{F\chi_S^B}.
                    \tag*{\qedhere}
    \end{align*}
  \end{proof}
  \begin{proof}[Proof of \autoref{domainCertMainThm}]
    We prove by induction over the index $i$ of main loop iterations that
    $T(\delta_i([x]_{P_i}))$ and $T(\beta_i([x]_{Q_i}))$ are a certificates for $[x]_{P_i}$
    and $[x]_{Q_i}$, respectively. (In the
    cancellative case, $Q_i$ and $\beta_i$ are not defined; so just
    put~$C/Q_i=\{C\}$, $\beta_i(C) = \top$ for convenience.)
    \begin{enumerate}
    \item For $i=0$, we trivially have
      \[
        \semantics{T(\beta_0([x]_{P_i}))} =  \semantics{T(\top)} = \semantics{\top} = C.
      \]
      Furthermore, unravelling \autoref{notationF1Mod},
      \[
        \delta_0([x]_{P_0}) = \fmod{F!(c(x))} = \fmod{Fj_1(F!(c(x)))}(\top,\top).
      \]
      Consequently,
      \[
        T(\delta_0([x]_{P_0})) = \tau_{F!(Fj_1(F!(c(x))))} = \tau_{F!(c(x))}
      \]
      using $!\cdot j_1\cdot \mathord{!} = \mathord{!}\colon C\to 1$. The naturality of
      $\semantics{\tau_o}$, $o\in F1$, implies that
      $\semantics{\tau_o}_X = \{t\in FX\mid F!(t) = o\}$. Hence,
      \[
        \semantics{T(\delta_0([x]_{P_0}))} =
        c^{-1}[\semantics{\tau_{F!(c(x))}}_C] = \{x'\in C\mid F!(c(x')) = F!(c(x))\}=[x]_{P_0}.
      \]

    \item In the inductive step, there is nothing to show for $\beta_{i+1}$
      because it is only a boolean combination of $\beta_{i}$ and $\delta_i$.
      For $\delta_{i+1}$, we distinguish two cases: whether the class
      $[x]_{P_i}$ is refined or not. If $[x]_{P_{i+1}} = [x]_{P_i}$, then 
      \[
        \semantics{T(\delta_{i+1}([x]_{P_{i+1}}))}
        = \semantics{T(\delta_{i}([x]_{P_i}))}
        = [x]_{P_i},
      \]
      and we are done. Now suppose that $[x]_{P_{i+1}}\neq [x]_{P_i}$ in the $i$-th
      iteration with chosen $S\subsetneqq B\subseteq C$. By
      \ref{defDeltai1} resp.~\ref{defDeltai1Optimized} we have:
      \[
        \delta_{i+1}([x]_{P_{i+1}}) = \delta_i([x]_{P_i}) \wedge
        \fmod{t}(\delta_i(S), \beta')
      \]
      where $\beta'$ is $\beta_i(B)$ or $\top$; in any case
      $\semantics{\delta_i(S)}= S\subseteq \semantics{\beta'}$.
      Note that $t$ here is either~$F\chi_S^B(c(x))$ (\autoref{algoCerts})
      or $Fj_2(F\chi_S(c(x)))$ (\autoref{algoCertsCancel}). Put $B'=B$ in
      the first case and $B'=C$ else. Using $\chi_S^{C} =
      j_2\cdot\chi_S$, we see that 
      \[
        t = F\chi_S^{B'}(c(x))
        \qquad
        \semantics{\beta'} = B',
        \qquad\text{and}\qquad
        \semantics{T(\beta')} = B',
      \]
      where the last equation follows from the inductive hypothesis.
      Thus, we have
      \[
        \delta_{i+1}([x]_{P_{i+1}}) = \delta_i([x]_{P_i}) \wedge
        \fmod{F\chi_S^{B'}(c(x))}(\delta_i(S), \beta'),
      \]
      and therefore
      \[
        T(\delta_{i+1}([x]_{P_{i+1}})) = T(\delta_i([x]_{P_i})) \wedge
        \lambda_{F\chi_S^{B'}(c(x))}\big(T(\delta_i(S)), T(\beta')
        \wedge \neg T(\delta_i(S))\big).
      \]
      Moreover, we have
      \[
        P_{i+1} = P_i\cap \ker(F\chi_S^{B'}\cdot c),
      \]
      in the first case by \autoref{defPi1}, in the second case by
      \autoref{optimizedPcancel}, recalling that $\chi_S = \chi_S^C$.
      
      We are now prepared for our final computation:
      \begin{align*}
        & \semantics{T(\delta_{i+1}([x]_{P_{i+1}})) }
          \\
        =~& \semantics{T(\delta_i([x]_{P_i})) \wedge
        \lambda_{F\chi_S^{B'}(c(x))}(T(\delta_i(S)), T(\beta')\wedge \neg T(\delta_i(S)))}
            \\
        =~& \semantics{T(\delta_i([x]_{P_i}))} \cap
        \semantics{\lambda_{F\chi_S^{B'}(c(x))}(T(\delta_i(S)), T(\beta')\wedge \neg T(\delta_i(S)))}
            \\
        =~& \semantics{T(\delta_i([x]_{P_i}))} \cap
        c^{-1}\big[\semantics{\lambda_{F\chi_S^{B'}(c(x))}}_C(\semantics{T(\delta_i(S))},
        \semantics{T(\beta')}\cap C\setminus \semantics{T(\delta_i(S))})\big]
            \tag{Semantics of $\heartsuit$}
            \\
        =~& [x]_{P_i} \cap
        c^{-1}\big[\semantics{\lambda_{F\chi_S^{B'}(c(x))}}_C(S, B'
        \cap C\setminus S )\big]
            \tag{I.H.}
            \\
        =~& [x]_{P_i} \cap
        c^{-1}\big[\semantics{\lambda_{F\chi_S^{B'}(c(x))}}_C(S, B'\setminus S)\big]
        \tag{$B' \cap C\setminus S = B'\setminus S$}
        \\   
        =~& [x]_{P_i} \cap
            [x]_{F\chi_{B'}\cdot c} \cap
        c^{-1}\big[\semantics{\lambda_{F\chi_S^{B'}(c(x))}}_C(S, B'\setminus S)\big]
            \tag{$P_i\subseteq \ker F\chi_{B'}\cdot c$}
            \\
        =~& [x]_{P_i} \cap
            c^{-1}\big[[c(x)]_{F\chi_{B'}}\big] \cap
        c^{-1}\big[\semantics{\lambda_{F\chi_S^{B'}(c(x))}}_C(S, B'\setminus S)\big]
            \\
        =~& [x]_{P_i} \cap
            c^{-1}\big[[c(x)]_{F\chi_{B'}} \cap
        \semantics{\lambda_{F\chi_S^{B'}(c(x))}}_C(S, B'\setminus S)\big]
            \\
        =~& [x]_{P_i} \cap
            c^{-1}\big[[c(x)]_{F\chi_S^{B'}}\big]
            \tag{domain-specific interpret.~(\autoref{domainSpecificSetIndexed})}
            \\
        =~& [x]_{P_i} \cap
            [x]_{F\chi_S^{B'}\cdot c}
            \\
        =~& [x]_{P_{i+1}}
            \tag{$P_{i+1} = P_i\cap \ker(F\chi_S^{B'}\cdot c$)}
      \end{align*}    
    \end{enumerate}
    This completes the proof.
  \end{proof}
\end{proofappendix}

\begin{proofappendix}[Details for]{exCertMarkov}
  The \autoref{algoCerts} runs in $\CO(m\cdot \log n)$ producing certificates of
  a total size of $\CO(m\cdot \log n)$. When translating these certificates for
  the modalities $\fpair{a}_p$ by the translation $T$,
  we obtain certificates for the input coalgebra (\autoref{domainCertMainThm}).
  However, the formula size
  has a blow up by the additional factor $|A|$ because of the big conjunctions
  in the domain-specific interpretation (\itemref{exDomainSpecInt}{dsiMarkov}).

  This represents is a better run time than that of the
  algorithm by Desharnais \etal~\cite[Fig.~4]{desharnaisEA02}, which nests multiple loops: four
  loops over blocks all blocks seen so far and one loop over $A$, roughly
  leading to a total run time in $\CO(|A|\cdot n^4)$.
\end{proofappendix}

\nopagebreak[4]
\appendixsect{worstcase}{Worst Case Tree Size of Certificates}
\nopagebreak[4]
\takeout{}%

\begin{proofappendix}[Details for]{cleavelandMinimal}
  To verify the minimality of $\varphi = \Diamond^{n+2}\top$,
  one considers all possible replacements of subformulae of $\varphi$ by $\top$:
  \[
    \Diamond\top
    \qquad
    \Diamond\Diamond\top
    \qquad
    \ldots
    \qquad
    \Diamond^n\top
    \qquad
    \Diamond^{n+1}\top
  \]
  All of these hold at both $x$ and $y$, because $x$ can perform arbitrarily
  many transitions and $y$ can perform $n+1$ transitions.
\end{proofappendix}

\noindent We note additionally that even the optimized algorithm for
cancellative functors (cf.~\autoref{algoCertsCancel}) constructs
certificates of exponential worst-case tree size:
\begin{expl}\label{algoExpFormula}
  Define the $\R^{(-)}$-coalgebra $c$ on $C =\bigcup_{k\in \N}\{w_k,x_k,y_k,z_k\}$ by
  \[
    \begin{array}{r@{\,}l@{\,}l@{\,}l@{\,}l@{\qquad\qquad}r@{}c@{\,}c@{\,}r}
      c(w_{k+1}) = & \{w_k\mapsto 1,& x_k\mapsto 2,& y_k\mapsto 1, &z_k\mapsto 2\} & c(w_0) = \{&w_0&\mapsto & 1\}\\
      c(x_{k+1}) = & \{w_k\mapsto 1,& x_k\mapsto 2,& y_k\mapsto 2, &z_k\mapsto 1\} & c(x_0) = \{&x_0&\mapsto &2\}\\
      c(y_{k+1}) = & \{w_k\mapsto 2,& x_k\mapsto 1,& y_k\mapsto 1, &z_k\mapsto 2\} & c(y_0) = \{&y_0&\mapsto &3\}\\
      c(z_{k+1}) = & \{w_k\mapsto 2,& x_k\mapsto 1,& y_k\mapsto 2, &z_k\mapsto 1\} & c(z_0) = \{&z_0&\mapsto &4\}
    \end{array}
  \]
  The optimized \autoref{algoCertsCancel} constructs a certificate of
  size $2^k$ in the $k$-th layer. In this example, however,
  linear-sized certificates do exist for all states,~e.g.
  \[
    \semantics{\grad{2}\grad{3}^k(\grad{1}\top \vee \grad{4}\top)} = \set{x_{k+1}}.
  \]
\end{expl}

\begin{proofappendix}[Details for]{algoExpFormula}
  Define the $\R^{(-)}$-coalgebra $c\colon C\to \R^{(C)}$ on the carrier
  \[
    C := 4\times \N \cong \bigcup\set[\big]{L_k\mid k\in \N} \qquad \text{for
    }L_k = \set{w_k,x_k,y_k,z_k}.
  \]
  We put
  \[
    \begin{array}{r@{\,}l@{\,}l@{\,}l@{\,}l@{\qquad\qquad}r@{}c@{\,}c@{\,}r}
      c(w_{k+1}) = & \{w_k\mapsto 1,& x_k\mapsto 2,& y_k\mapsto 1, &z_k\mapsto 2\} & c(w_0) = \{&w_0&\mapsto & 1\}\\
      c(x_{k+1}) = & \{w_k\mapsto 1,& x_k\mapsto 2,& y_k\mapsto 2, &z_k\mapsto 1\} & c(x_0) = \{&x_0&\mapsto &2\}\\
      c(y_{k+1}) = & \{w_k\mapsto 2,& x_k\mapsto 1,& y_k\mapsto 1, &z_k\mapsto 2\} & c(y_0) = \{&y_0&\mapsto &3\}\\
      c(z_{k+1}) = & \{w_k\mapsto 2,& x_k\mapsto 1,& y_k\mapsto 2, &z_k\mapsto 1\} & c(z_0) = \{&z_0&\mapsto &4\}
    \end{array}
  \]
  For the complexity class of the formulae generated, consider the subcoalgebra
  on $L_0\cup \cdots \cup L_n$.
  
  The initial partition $P_0=\set[\big]{\set{w_0},\set{x_0}, \set{y_0}, \set{z_0},
    L_1\cup\cdots L_n}$ distinguishes on the total out-degree (being 1, 2, 3,
  4, or 6). Consider that after $i\in \N$ iterations of the main loop of the
  algorithm, the states $w_k, x_k, y_k, z_k$ have just been found to be behaviourally
  different and all states of $L_{k+1}\cup \cdots \cup L_{n}$ are still
  identified. Then the algorithm has to use some of the blocks $\set{w_k}$, $\set{x_k}$,
  $\set{y_k}$, $\set{z_k}$ as the splitter $S$ for further refinement. Assume wlog
  that $S:=\{w_k\}$ is used as the splitter, first. This will have the effect that
  $L_{k+1}\cup\cdots\cup L_{n}$ will be refined into the blocks
  \[
    \{w_{k+1},x_{k+1}\},\qquad
    \{y_{k+1},z_{k+1}\},\qquad
    L_{k+2}\cup\cdots\cup L_n.
  \]
  Assume, that the formula for $w_k$ is $\delta(\set{w_k})$ at this point (we
  omit the index, since the singleton block $\set{w_k}$ can not be refined
  further). The definition of $\delta$ in the algorithm annotates the block
  $\{w_{k+1},x_{k+1}\}$ with $\grad{1}\delta(\set{w_k})$
  and the block $\{y_{k+1},z_{k+1}\}$ with $\grad{2}\delta(\set{w_k})$.

  Splitting by $\{x_k\}$ does not lead to further refinement. However, when
  splitting by $S:= \{y_k\}$ (or equivalently $\{z_k\}$), $\set{w_{k+1},
    x_{k+1}}$ is split into $\set{w_{k+1}}$ and $\set{x_{k+1}}$ and likewise
  $\set{y_{k+1}, z_{k+1}}$ into $\set{y_{k+1}}$ and $\set{z_{k+1}}$. Let
  $\delta(\set{y_k})$ be the certificate constructed for $\set{y_k}$. This
  implies that the formulas for $\set{w_{k+1}}$ and $\set{y_{k+1}}$ are
  respectively extended by the conjunct $\grad{1}\delta(\set{y_k})$; likewise,
  the formulas for $\set{x_{k+1}}$ and $\set{z_{k+1}}$ obtain a new conjunct
  $\grad{2}\delta\set{y_k}$. Hence, for every $s\in L_{k+1}$ the tree-size of
  the formula constructed is at least:
  \[
    |\delta(\set{s})| \ge |\delta(\set{w_{k}})| + |\delta(\set{y_{k}})| .
  \]
  Thus the tree-size of the certificate constructed for cancellative functors
  may grow exponentially with the state count.

  Despite the exponential tree-size of the formulas constructed, there exist
  linearly sized certificates for all states in the above coalgebra $(C,c)$. First, we
  have
  \[
    \phi_{k} := \grad{3}^k(\grad{1}\top \vee \grad{4}\top)
    \text{ with }
    \semantics{\phi_{k}} = \set{w_k,z_k}
  \]
  This lets us define certificates for $x_{k+1}$ and $y_{k+1}$:
  \[
    \semantics{\grad{2}\phi_k} = \set{x_{k+1}}
    \quad\text{and}\quad
    \semantics{\grad{4}\phi_k} = \set{y_{k+1}}
  \]
  For the remaining two state sequences $w$ and $z$ we
  first note
  \[
    \semantics{\grad{1}\grad{4}\phi_k} = \set{w_{k+2},y_{k+2}}
  \]
  and thus have certificates
  \[
    \semantics{\phi_{k+2}\wedge \grad{1}\grad{4}\phi_k } = \set{w_{k+2}}
    \quad\text{and}\quad
    \semantics{\phi_{k+2}\wedge \neg\grad{1}\grad{4}\phi_k } = \set{z_{k+2}}.
  \]
  Since $\phi_k$ involves $k+2$ modal operators, every state in $L_k$ has a certificate
  with at most $2\cdot k + 8$ modal operators.
\end{proofappendix}

\end{document}